\PassOptionsToPackage{frozencache}{minted}
\def\arxiv{}
\pdfoutput=1
\ifdefined\arxiv
  \documentclass[a4paper,USenglish]{lipics-v2021}
\else
  \documentclass[a4paper,USenglish,anonymous]{lipics-v2021}
\fi

\newif\ifarxiv
\ifdefined\arxiv
  \arxivtrue
  \newcommand{\provsqlSite}{\url{https://provsql.org/}}
  \newcommand{\provsqlCommit}{\href{https://github.com/PierreSenellart/provsql/tree/b3445d598b9fc32330a933d1950d4e89402f1502}{\texttt{b3445d59}}}
  \newcommand{\leanDocs}{\url{https://provsql.org/lean-docs/Provenance.html}}
  \hideLIPIcs
  \nolinenumbers
\fi

\title{Provenance of HAVING Queries in~Semirings~with~Monus}

\ifarxiv
\author{Aryak Sen}{Univ.\ Grenoble Alpes, CNRS, Grenoble INP, LIG, France}%
  {aryak.sen@univ-grenoble-alpes.fr}{https://orcid.org/0009-0009-2261-682X}{}
\author{Pratik Karmakar}{National University of Singapore, Singapore \and
  CNRS@CREATE, Singapore}{pratik.karmakar@u.nus.edu}%
  {https://orcid.org/0009-0008-1111-8801}{}
\author{Silviu Maniu}{Univ.\ Grenoble Alpes, CNRS, Grenoble INP, LIG, France
  \and CNRS@CREATE \& IPAL, Singapore}{silviu.maniu@univ-grenoble-alpes.fr}%
  {https://orcid.org/0000-0002-8623-1533}{}
\author{Angelo Saadeh}{DI ENS, ENS, CNRS, PSL University, Inria, Paris, France}%
  {angelo.saadeh@ens.psl.eu}{https://orcid.org/0009-0000-2081-6232}{}
\author{Pierre Senellart}{DI ENS, ENS, CNRS, PSL University, Inria,
  Paris, France \and CNRS@CREATE \& IPAL, Singapore}{pierre@senellart.com}%
  {https://orcid.org/0000-0002-7909-5369}{}

\authorrunning{A. Sen, P. Karmakar, S. Maniu, A. Saadeh, and P. Senellart}

\Copyright{Aryak Sen, Pratik Karmakar, Silviu Maniu, Angelo Saadeh, and
  Pierre Senellart}
\else
\author{Anonymous}{Anonymous}{}{}{}
\fi



\category{}

\relatedversion{}

\ifarxiv
\supplement{}
\supplementdetails[subcategory={ProvSQL},linktext={provsql.org}]{Software}{https://provsql.org/}
\supplementdetails[subcategory={Lean formalization},linktext={provsql.org/lean-docs/Provenance.html}]{Software}{https://provsql.org/lean-docs/Provenance.html}

\funding{This research is part of the program DesCartes and is supported by
  the National Research Foundation, Prime Minister's Office, Singapore under
  its Campus for Research Excellence and Technological Enterprise (CREATE)
  program. It is also partially supported by DataGEMS, funded by the
  European Union's Horizon Europe Research and Innovation programme, under
  grant agreement 101188416, and by the French government under management
  of Agence Nationale de la Recherche (ANR) as part of the ``France 2030''
  program, reference ANR-23-IACL-0008 (PR[AI]RIE-PSAI).}
\fi

\usepackage{minted}
\setminted{fontshape=up}
\usepackage{multicol}
\usepackage{upgreek}
\usepackage{booktabs}
\usepackage{pifont}

\usepackage[only,llbracket,rrbracket]{stmaryrd}

\DeclareFontFamily{U}{FdSymbolF}{}
\DeclareFontShape{U}{FdSymbolF}{m}{n}{
  <-7.1> FdSymbolF-Book
  <7.1->  FdSymbolF-Medium
}{}
\DeclareSymbolFont{fdsymdelims}{U}{FdSymbolF}{m}{n}
\DeclareMathDelimiter{\lAngle}{\mathopen}{fdsymdelims}{"92}{fdsymdelims}{"92}
\DeclareMathDelimiter{\rAngle}{\mathclose}{fdsymdelims}{"98}{fdsymdelims}{"98}

\usepackage[bb=boondox]{mathalfa}

\usepackage[linesnumbered,algoruled,commentsnumbered,boxruled,noend]{algorithm2e}
\newcommand{\fakesql}[1]{\texttt{#1}}

\newcommand{\NN}{\mathbb{N}}
\newcommand{\RA}{\mathsf{RA}}
\newcommand{\mset}[1]{\{\!\{\,#1\,\}\!\}}
\newcommand{\ansem}[2]{\llbracket #2 \rrbracket_{#1}}
\newcommand{\angsem}[2]{\lAngle #2\rAngle^{#1}}
\newcommand{\predsem}[3]{\llbracket #3\rrbracket^{#1}_{#2}}

\newcommand{\ok}{\text{\color{green!50!black}\ding{51}}}
\newcommand{\ko}{\text{\color{red!50!black}\ding{55}}}

\newcommand{\op}{\mathrel{\uptheta}}

\renewcommand{\epsilon}{\varepsilon}
\renewcommand{\phi}{\varphi}
\renewcommand{\le}{\leqslant}
\renewcommand{\ge}{\geqslant}
\renewcommand{\leq}{\leqslant}
\renewcommand{\geq}{\geqslant}

\newcommand{\inlinepar}[1]{\smallskip\noindent\textbf{#1}}

\usepackage{xstring}
\ifarxiv
  \newcommand\leanbase{https://provsql.org/lean-docs/Provenance}
  \newcommand\leananchor[1]{\def\leancuranchor{#1}}
\else
  \newcommand\leanbase{https://anonymousprovenance.github.io/Provenance}
  \newcommand\leananchor[1]{\StrSubstitute{#1}{MinTropical}{Tropical}[\leancuranchor]}
\fi
\newcommand\leanicon{\fbox{\includegraphics[height=.6em]{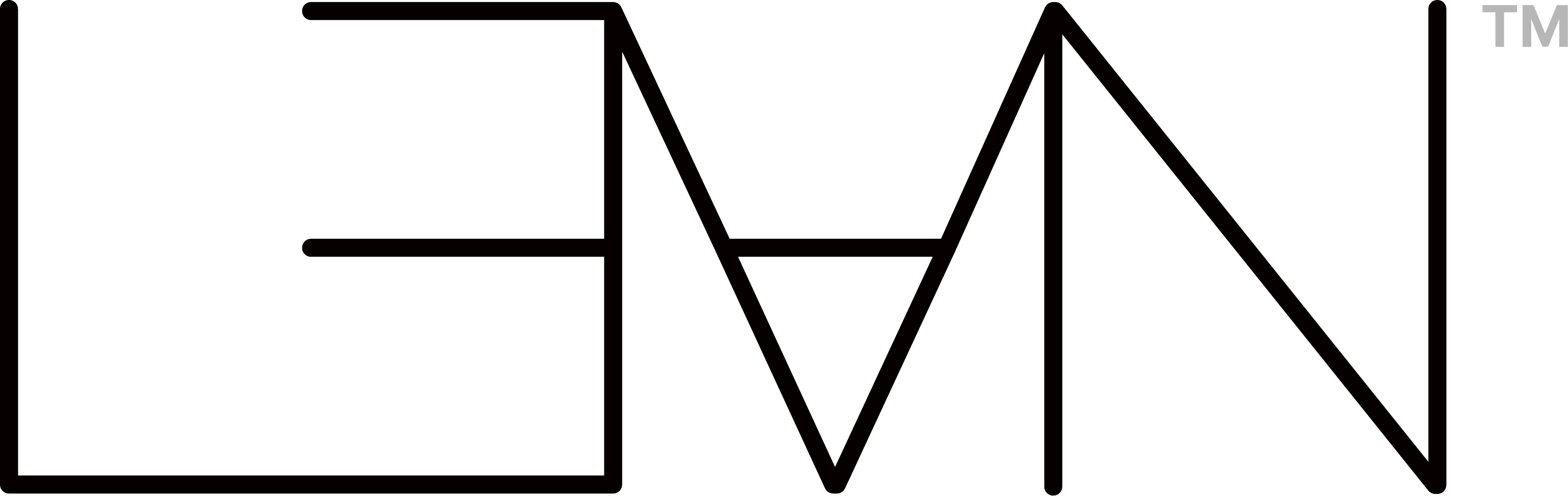}}}
\newcommand\leanmain{\href{\leanbase.html}{\leanicon}}
\newcommand\leanlink[3]{\leananchor{#2}\href{\leanbase/#1.html\#\leancuranchor}{#3}}
\newcommand\lean[2]{\leanlink{#1}{#2}{\leanicon}}

\newcommand{\sql}{\mintinline{postgresql}}

\newcommand{\Bag}{\mathrm{Bag}}
\newcommand{\Seq}{\mathrm{Seq}}
\newcommand{\msetminus}{\setminus_{\Bag}}
\newcommand{\msubseteq}{\subseteq_{\Bag}}
\newcommand{\ann}{\mathrm{ann}}

\newcommand{\defeq}{\mathrel{\mathop:}=}

\RestyleAlgo{boxed}
\DontPrintSemicolon
\SetKwInOut{Input}{Input}
\SetKwInOut{Output}{Output}
\SetKwProg{Fn}{Function}{}{}
\usepackage[appendix=append,bibliography=common]{apxproof}
\renewcommand{\appendixsectionformat}[2]{Material for Section~#1\ (#2)}

\usepackage{tikz}
\usepackage[autorun,cite,citelabel=tag]{proofgraph}

\newcommand\pgphstyle[2]{\proofgraphstyle{#1}{#2}\proofgraphstyle{#1rep}{#2}}
\pgphstyle{theorem}{shape=box,style=filled,fillcolor=gold,penwidth=2}
\pgphstyle{proposition}{shape=box,style="rounded,filled",fillcolor=lightblue}
\pgphstyle{corollary}{shape=box,style="rounded,filled",fillcolor=palegoldenrod}
\pgphstyle{lemma}{shape=ellipse,style=filled,fillcolor=lightgray}
\proofgraphignore{th:correctness}{prop:idempotent_iff_add_monus}

\makeatletter
\@for\axp@env:={theorem,lemma,corollary,proposition}\do{%
  \expandafter\let\csname\axp@env\endcsname\@undefined
  \expandafter\let\csname end\axp@env\endcsname\@undefined
}
\makeatother
\newtheoremrep{theorem}{Theorem}
\newtheoremrep{proposition}[theorem]{Proposition}
\newtheoremrep{corollary}[theorem]{Corollary}
\newtheoremrep{lemma}[theorem]{Lemma}

\begin{document}

\maketitle

\begin{abstract}
The semiring framework and its extensions form the basis of a rich collection of theoretical
results and implementations for provenance tracking of database queries.
Many real-world queries use aggregation and conditions on the aggregate
values.
Support for such queries has been proposed by introducing semimodule
elements as aggregate values and formal comparisons between aggregate values as tuple
annotations, which takes the approach outside the standard semiring framework. In this work, we
show how to introduce a semantics for the provenance of such queries in
arbitrary commutative semirings with monus (or \emph{m-semirings}), without the need for additional
operators. This semantics is shown to agree with the standard provenance of the
aggregation-free self-join rewriting of \sql{HAVING COUNT(*)} queries in
semirings that are absorptive and where times distributes over monus. We
derive algorithms for this semantics and implement them within the
ProvSQL system, with viable performance on a real-world dataset for
probabilistic query evaluation.
\end{abstract}

\section{Introduction}
Data provenance is information on ``where a piece of data came from and
the process by which it arrived in the database''~\cite{bunemankt01}. Data
provenance has numerous applications, in particular with respect to
understandability, reproducibility, and quality management of a
data-handling process~\cite{herschel2017survey}. A particular form of
data provenance, \emph{Boolean provenance}
\cite{senellart2017provenance}, sometimes also called (Boolean) \emph{lineage}
\cite{suciu2011probabilistic} is used as a basis for query
evaluation in probabilistic databases in the \emph{intensional}
\cite{jha2010bridging} approach to probabilistic query evaluation.

A particularly rich framework for data provenance, especially in the
setting of relational databases, has been that of \emph{provenance
semirings} \cite{green2007provenance} on \emph{annotated relations},
i.e., relations with an extra attribute storing the provenance, which
assigns a semantics to
positive relational algebra operators in terms of elements of an arbitrary
semiring. Different choices of this semiring lead to different forms of
provenance. This has been further extended to other settings and
classes of queries,
notably to queries with the difference operator \cite{geerts2010database}
by adding a \emph{monus}~$\ominus$ operator to the semiring, and to queries
involving grouping and aggregation \cite{amsterdamer2011provenance} by
introducing an additional $\delta$ operator for grouping as well as
the ability to turn data values into elements of a
\emph{semimodule} to reflect the result of an aggregation operation.

The next step is to define a semantics for queries involving
selections on the result of an aggregation, i.e., \sql{HAVING}
queries using the name of the corresponding SQL clause. This was done
elegantly in~\cite{amsterdamer2011provenance} by extending the semiring
with formal \emph{comparison expressions} between semimodule elements,
which encode the selection within the annotation of a tuple. Such an
expression resolves to a truth value only when the semimodule is
isomorphic to the aggregation monoid, which fails for the semirings
needed for probabilistic query evaluation; in addition, the approach is
restricted to aggregates that are commutative monoid
operations.

Provenance tracking systems on top of regular database management systems
have been successfully implemented, usually based on query rewriting to construct the
annotations, notably GProM~\cite{arab2018gprom} and
\ifarxiv our own system, \fi
ProvSQL~\cite{senellart2018provsql,sen2026provsql}. However, partly because it
is unclear what semantics should be given to the comparison operator,
aggregation support has been partial~\cite{arab2018gprom},
or has followed~\cite{amsterdamer2011provenance} to
the letter and stopped at constructing abstract formulas involving the
comparison operator, with no concrete semiring
\ifarxiv
elements~\cite{pintor2025dbms}, as ProvSQL did before this
work~\cite{sen2026provsql}.
\else
elements~\cite{sen2026provsql,pintor2025dbms}.
\fi
In contrast, most semirings used for data provenance (see
Section~\ref{sec:preliminaries}) -- though not all
semirings~\cite{amarilli2016example} -- can be equipped in a canonical way
with a monus operator, making semirings with monus (or \emph{m-semirings}
for short), despite their limitations \cite{amsterdamer2011limitations},
a usable and pragmatic provenance framework. For this reason, m-semirings were chosen
as the annotation framework
to show an analogue of Codd's theorem of the equivalence between the relational
calculus and the relational algebra on annotated
relations~\cite{badia2025codd} and\ifarxiv, in our previous work,\fi{} to
support non-monotone queries in ProvSQL~\cite{sen2026provsql}.

This paper leverages m-semirings as provenance annotation framework to provide
an implementable semantics for the comparison operator from
\cite{amsterdamer2011provenance} (and therefore for \sql{HAVING}
queries) \emph{in arbitrary commutative
m-semirings}. In more detail:
\begin{enumerate}[(i)]
\item We give a systematic account of the commutative m-semirings used
  in the literature, each with a machine-checked Lean proof of its
  definition and properties, correcting erroneous claims
  of~\cite{geerts2010database,amsterdamer2011limitations}
  (Section~\ref{sec:preliminaries}).

  \item We introduce a new semantics for the provenance of \sql{HAVING} queries
  in any
  arbitrary commutative m-semiring and for any aggregate
  function (not necessarily commutative or based on a monoid). (Section
  \ref{sec:semantics})

\item We show that, on \sql{COUNT(*)} comparisons, this semantics
  agrees with the standard provenance of the aggregation-free self-join
  rewriting when the m-semiring is absorptive and $\otimes$ distributes
  over $\ominus$; that it commutes with homomorphisms; and that it is
  compatible with probabilistic query evaluation.
  (Section~\ref{sec:results}).

\item We give algorithms for this semantics: possible-world enumeration
  in general, and algorithms specialized for \sql{COUNT}, \sql{SUM},
  \sql{MIN} and \sql{MAX} that are polynomial in identified cases,
  especially for absorptive m-semirings (Section~\ref{sec:algorithms}).

\item We implement these algorithms within ProvSQL and show
  the viability of the approach on a real-world database,
  including for probabilistic query evaluation
  (Section~\ref{sec:implementation_experiments}).
\end{enumerate}

Related work is discussed in Section~\ref{sec:related_work}.
Every statement that can be formalized in Lean is accompanied by a
machine-checked proof, available in
\ifarxiv
the open-source Lean library provenance-lean (\leanDocs)
\else
an anonymized repository
\fi
and
hyperlinked with \leanmain{} throughout the paper; the exceptions are
running-time bounds of Propositions~\ref{prop:algorithms}
and~\ref{prop:agg-cmp-poly-prob}, due to no current Lean support.
\ifarxiv
Our implementation is part of ProvSQL, available in open source at
\provsqlSite; this paper describes its revision~\provsqlCommit.
Most proofs and some additional details are in the
\hyperref[appendix]{appendix}.
\else
Our implementation
will be made available in open source
once double-blind requirements are removed. For lack of space,
most proofs and some additional details are relegated to the
\hyperref[appendix]{appendix}.
\fi

\section{Preliminaries}\label{sec:preliminaries}

\begin{toappendix}
  \label{appendix}
\end{toappendix}

We start with preliminaries on (m-)semirings, annotated relations, and
queries; the syntax and semantics of aggregate queries
from~\cite{sen2026provsql} are recalled in the appendix.

\begin{definition} \label{def:semiring}
  A \emph{semiring} is an algebraic structure $(\mathbb K,\oplus,\otimes,\mathbb 0,\mathbb 1)$ where $\mathbb K$ is a set, $\oplus$ and~$\otimes$ are binary operations on $\mathbb K$, and $\mathbb 0,\mathbb 1\in \mathbb K$, such that:
  \begin{itemize}
    \item $(\mathbb K,\oplus,\mathbb 0)$ is a commutative monoid and $(\mathbb K,\otimes,\mathbb 1)$ is a monoid;
    \item $\otimes$ distributes over $\oplus$: $\forall~a,b,c\in\mathbb K$, $a\otimes (b\oplus c) = (a\otimes b)\oplus (a\otimes c)$ and $(a\oplus b)\otimes c = (a\otimes c)\oplus (b\otimes c)$;
    \item $\mathbb 0$ is an annihilator for $\otimes$: $\forall
      a\in\mathbb K$, $a\otimes \mathbb 0 = \mathbb 0 \otimes a = \mathbb 0$.
  \end{itemize}
\end{definition}

\begin{example}\label{ex:semirings}
  Some example semirings
  \cite{green2007provenance,senellart2017provenance} are:
  the counting semiring $\mathbb N$: $(\mathbb N,+,\times,0,1)$;
  the Boolean semiring $\mathbb B$:
  $(\{\top,\bot\},\lor,\land,\bot,\top)$;
  the Why-provenance semiring over some set of variables $X$:
      $(2^{2^X},\cup,\Cup,\emptyset,\{\emptyset\})$ where $\Cup$ is defined as follows: $A\Cup B \defeq \{a\cup b \mid a\in A, b\in B\}$;
  the Boolean-function semiring $\mathbb B[X]$ over a set of variables $X$: Boolean functions
  over $X$ (propositional formulas up to logical equivalence) with
  $(\lor,\land,\bot,\top)$.
\end{example}

Following~\cite{geerts2010database}, we extend the definition of a semiring with a monus operation $\ominus$;
note that this definition, when possible, uniquely defines such an
operation \lean{SemiringWithMonus}{monus_unique}:
\begin{definition} \label{def:m-semiring}
  Let $(\mathbb K, \oplus, \otimes, \mathbb 0, \mathbb 1)$ be a semiring.
  For $a, b\in\mathbb K$, define
$a\le b$ as $\exists c\in\mathbb K, {b = a\oplus c}$. If $\leq$ is an order
relation,\footnote{$\leq$ is always a
pre-order, but not always antisymmetric: on the ring $\mathbb Z$ of
integers, any two elements are related in both directions, e.g.,
$0\le 1\le 0$ with $0\neq 1$
\lean{SemiringWithMonus}{natural_preorder_not_antisymm}.}
it is called
the \emph{natural order} of the semiring, and the semiring is said to be
\emph{naturally ordered}.
An \emph{m-semiring} (or \emph{semiring with monus})
is defined as a naturally ordered semiring in which we can introduce an
operation $\ominus$ satisfying the \emph{residuation law}:
\(
  x\ominus y \le z \iff x \le y\oplus z
  \).
\end{definition}

It is possible to construct
a naturally ordered semiring where a $\ominus$ cannot be
introduced~\cite{amarilli2016example}.
An m-semiring $(\mathbb K, \oplus, \otimes, \ominus, \mathbb 0, \mathbb
1)$ satisfies some basic properties, for $a, b, c\in\mathbb K$:
\vspace*{-1em}
\begin{multicols}{2}
  \begin{itemize}
    \item \lean{SemiringWithMonus}{add_monus}
      $a\oplus(b\ominus a) = b\oplus(a\ominus b)$;
    \item \lean{SemiringWithMonus}{monus_add}
      $(a\ominus b)\ominus c = a\ominus(b\oplus c)$;
    \item \lean{SemiringWithMonus}{monus_self}
      \lean{SemiringWithMonus}{zero_monus} $a\ominus a = \mathbb 0\ominus a=\mathbb 0$;
    \item \lean{SemiringWithMonus}{monus_zero} $a\ominus \mathbb 0=a$.
  \end{itemize}
\end{multicols}
\vspace*{-1em}
\noindent In addition to these, m-semirings may have other properties. Useful ones are:
\begin{itemize}
  \item \textbf{Commutativity:} $\forall a,b\in \mathbb K$, $a\otimes b = b\otimes a$.
  \item \textbf{Idempotence:} \lean{SemiringWithMonus}{idempotent} $\forall a\in \mathbb K$, $a\oplus a = a$.
  \item \textbf{Absorptivity:} \lean{SemiringWithMonus}{absorptive} $\forall a\in \mathbb K$, $\mathbb 1 \oplus a = \mathbb 1$.
  \item \textbf{Distributivity of $\otimes$ over $\ominus$}:
    \lean{SemiringWithMonus}{mul_sub_left_distributive} $\forall a,b,c\in
    \mathbb K$, $a\otimes (b\ominus c) = (a\otimes b) \ominus (a\otimes
    c)$.
  \item \textbf{Characteristic $0$:} \(\forall n\geq 1,
    \bigoplus_{i=1}^n\mathbb{1}\neq \mathbb{0}\).
\end{itemize}

\begin{table}
  \centering
  \begin{tabular}{rlcccc}
\toprule
\bfseries M-semiring&
\bfseries Definition of $a\ominus b$&
\bfseries Idem.&
\bfseries Abs.&
\bfseries Distr.&
\bfseries Char.~$0$\\
\midrule
\leanlink{Semirings/Bool}{instCommSemiringWithMonusBool}{$\mathbb{B}$}&
$a\land\lnot b$&
\leanlink{Semirings/Bool}{Bool.idempotent}{\ok}&
\leanlink{Semirings/Bool}{Bool.absorptive}{\ok}&
\leanlink{Semirings/Bool}{Bool.mul_sub_left_distributive}{\ok}&
\leanlink{Semirings/Bool}{Bool.instCharPZero}{\ok}\\
\leanlink{Semirings/BoolFunc}{instCommSemiringWithMonusBoolFunc}{$\mathbb{B}[X]$}&
$\nu\mapsto a(\nu)\land\lnot b(\nu)$&
\leanlink{Semirings/BoolFunc}{BoolFunc.idempotent}{\ok}&
\leanlink{Semirings/BoolFunc}{BoolFunc.absorptive}{\ok}&
\leanlink{Semirings/BoolFunc}{BoolFunc.mul_sub_left_distributive}{\ok}&
\leanlink{Semirings/BoolFunc}{BoolFunc.instCharPZero}{\ok}\\
\leanlink{Semirings/Why}{instCommSemiringWithMonusWhy}{$\mathrm{Why}[X]$}&
$a\setminus b$&
\leanlink{Semirings/Why}{Why.idempotent}{\ok}&
\leanlink{Semirings/Why}{Why.not_absorptive}{\ko}&
\leanlink{Semirings/Why}{Why.not_mul_sub_left_distributive}{\ko}&
\leanlink{Semirings/Why}{Why.instCharPZero}{\ok}\\
\leanlink{Semirings/Which}{instCommSemiringWithMonusWhich}{$\mathrm{Which}[X]$}&
$a$ if $\bot\in\{a,b\}$; $\bot$ if $a\subseteq b$; else $a\setminus b$&
\leanlink{Semirings/Which}{Which.idempotent}{\ok}&
\leanlink{Semirings/Which}{Which.not_absorptive}{\ko}&
\leanlink{Semirings/Which}{Which.not_mul_sub_left_distributive}{\ko}&
\leanlink{Semirings/Which}{Which.instCharPZero}{\ok}\\
\leanlink{Semirings/IntervalUnion}{instCommSemiringWithMonusIntervalUnionOfDenselyOrderedOfBoundedOrder}{Temporal}&
$a\setminus b$&
\leanlink{Semirings/IntervalUnion}{IntervalUnion.idempotent}{\ok}&
\leanlink{Semirings/IntervalUnion}{IntervalUnion.absorptive}{\ok}&
\leanlink{Semirings/IntervalUnion}{IntervalUnion.mul_sub_left_distributive}{\ok}&
\leanlink{Semirings/IntervalUnion}{IntervalUnion.instCharPZero}{\ok}\\
\leanlink{Semirings/Tropical}{instCommSemiringWithMonusMinTropicalWithTopNat}{Tropical ($\mathbb{N}$)}&
if $a\geq b$ then $+\infty$ else $a$&
\leanlink{Semirings/Tropical}{MinTropical.idempotent}{\ok}&
\leanlink{Semirings/Tropical}{MinTropicalN.absorptive}{\ok}&
\leanlink{Semirings/Tropical}{MinTropicalN.mul_sub_left_distributive}{\ok}&
\leanlink{Semirings/Tropical}{MinTropicalN.instCharPZero}{\ok}\\
\leanlink{Semirings/Tropical}{instCommSemiringWithMonusMinTropicalWithTopReal}{Tropical ($\mathbb{R}$)}&
if $a\geq b$ then $+\infty$ else $a$&
\leanlink{Semirings/Tropical}{MinTropical.idempotent}{\ok}&
\leanlink{Semirings/Tropical}{MinTropicalR.not_absorptive}{\ko}&
\leanlink{Semirings/Tropical}{MinTropicalR.mul_sub_left_distributive}{\ok}&
\leanlink{Semirings/Tropical}{MinTropicalR.instCharPZero}{\ok}\\
\leanlink{Semirings/Viterbi}{Viterbi.instCommSemiringWithMonus}{Viterbi}&
if $a\leq b$ then $0$ else $a$&
\leanlink{Semirings/Viterbi}{Viterbi.idempotent}{\ok}&
\leanlink{Semirings/Viterbi}{Viterbi.absorptive}{\ok}&
\leanlink{Semirings/Viterbi}{Viterbi.mul_sub_left_distributive}{\ok}&
\leanlink{Semirings/Viterbi}{Viterbi.instCharPZero}{\ok}\\
\leanlink{Semirings/MinMax}{instCommSemiringWithMonusMinMax}{Security}&
if $a\geq b$ then $\mathbb 0$ else $a$&
\leanlink{Semirings/MinMax}{MinMax.idempotent}{\ok}&
\leanlink{Semirings/MinMax}{MinMax.absorptive}{\ok}&
\leanlink{Semirings/MinMax}{TVL.not_mul_sub_left_distributive}{\ko}&
\leanlink{Semirings/MinMax}{MaxMin.instCharPZero}{\ok}\\
\leanlink{Semirings/Nat}{instCommSemiringWithMonusNat}{$\mathbb{N}$}&
$\max(a-b, 0)$&
\leanlink{Semirings/Nat}{Nat.not_idempotent}{\ko}&
\leanlink{Semirings/Nat}{Nat.not_absorptive}{\ko}&
\leanlink{Semirings/Nat}{Nat.mul_sub_left_distributive}{\ok}&
\leanlink{Semirings/Nat}{Nat.charP_zero}{\ok}\\
\leanlink{Semirings/How}{instCommSemiringWithMonusMvPolynomialNat}{$\mathbb{N}[X]$}&
\(\sum_{\text{$M$}}\max(\mathrm{c}_M(a)-\mathrm{c}_M(b), 0)M\)&
\leanlink{Semirings/How}{How.not_idempotent}{\ko}&
\leanlink{Semirings/How}{How.not_absorptive}{\ko}&
\leanlink{Semirings/How}{How.not_mul_sub_left_distributive}{\ko}&
\leanlink{Semirings/How}{How.charP_zero}{\ok}\\
\leanlink{Semirings/Lukasiewicz}{instCommSemiringWithMonusLukasiewicz}{$\textrm{\L ukasiewicz}$}&
if $a\leq b$ then $0$ else $a$&
\leanlink{Semirings/Lukasiewicz}{Lukasiewicz.idempotent}{\ok}&
\leanlink{Semirings/Lukasiewicz}{Lukasiewicz.absorptive}{\ok}&
\leanlink{Semirings/Lukasiewicz}{Lukasiewicz.mul_sub_left_distributive}{\ok}&
\leanlink{Semirings/Lukasiewicz}{Lukasiewicz.instCharPZero}{\ok}\\
\bottomrule
\end{tabular}
\caption{Properties of common commutative m-semirings: idempotence,
  absorptivity, distributivity of $\otimes$ over $\ominus$, characteristic
  $0$; every entry links to its machine-checked proof \leanmain}
\label{tab:semirings}
\end{table}

In this work, we restrict to \emph{commutative} semirings, which is
actually the setting of most works on data provenance
\cite{green2007provenance,geerts2010database}, though not of
\cite{sen2026provsql}.
All semirings from Example~\ref{ex:semirings} are commutative
and can be extended to m-semirings, as can all the semirings of
Table~\ref{tab:semirings} and, to our knowledge, all other semirings that
have been used to annotate data (rings, which are not naturally ordered,
are not among them).
Table~\ref{tab:semirings} summarizes the definitions of the $\ominus$
operator and whether they satisfy the other four properties above for
the semirings in Example~\ref{ex:semirings}, along with the
Which[$X$] semiring \cite{green2017provenance}, the
Temporal semiring of unions of intervals
used in temporal databases
\cite{widiaatmaja2025demonstration}, the tropical semirings
\cite{green2007provenance}, the Viterbi semiring~\cite{green2017provenance}, the security or access control semiring
\cite{green2017provenance}, the $\mathbb{N}[X]$ or How[$X$] semiring
\cite{green2007provenance} and the \L ukasiewicz semiring~\cite{gradel2025provenance}.
We note that only the Boolean, Boolean-function, Temporal, Tropical (over~$\mathbb{N}$), Viterbi, and \L
ukasiewicz m-semirings are idempotent, absorptive and have $\otimes$
distributive over $\ominus$.
These properties will be important for the semantics of \sql{HAVING}
queries. Note that this table, along with the formal Lean proofs that
ensure its validity, definitively fixes
a
few erroneous statements in the literature (see textual proofs in
appendix):
\cite[Example~4]{geerts2010database}
wrongly claimed that the tropical semiring over reals could not be equipped with
$\ominus$; \cite{amsterdamer2011limitations} wrongly claimed that
both $\mathbb{N}[X]$ and Why[$X$] had $\otimes$ distributive over $\ominus$
(called Axiom A13 in~\cite{amsterdamer2011limitations}). Such errors are
easy to make given some counter-intuitive properties of
these basic building blocks; clean, formally proven statements
are one of our contributions.

For queries with grouping and aggregation, an additional useful notion
\cite{amsterdamer2011provenance} is that of \emph{$\delta$-semirings},
which is a semiring with an additional unary operator $\delta:\mathbb
K\to \mathbb K$, used to capture the semantics of the existence of a
group. Two properties are required in \cite{amsterdamer2011provenance}:
(i) $\delta(\mathbb
0)=\mathbb 0$ and (ii) $\delta(\mathbb 1 \oplus \dots \oplus \mathbb
1)=\mathbb 1$. Non-trivial semirings in which such a $\delta$ can be defined are exactly
those of characteristic $0$
\lean{SemiringWithMonus}{delta_exists_iff_charP_zero}.
These properties leave $\delta$ largely underdetermined. We require one
further axiom, \emph{$\delta$-absorption}: $\forall a,b\in\mathbb K$,
$a\otimes\delta(a\oplus b)=a$; a group-existence factor is thus redundant
next to any annotation that already contains an occurrence of that group.
The characterization above is unaffected by this strengthening: in an
m-semiring, a~$\delta$ satisfying all three axioms exists exactly when the
characteristic is $0$, since the support indicator (see below) satisfies
them all \lean{SemiringWithMonus}{isDelta_exists_iff_charP_zero}.
The two usual definitions of $\delta$ are the identity,
$\delta(x)\defeq x$, and the \emph{support indicator},
$\delta(x)\defeq\mathbb 1$ for $x\neq\mathbb 0$ and
$\delta(\mathbb 0)\defeq\mathbb 0$. We choose the former whenever it
satisfies all properties (e.g., in $\mathbb{B}[X]$
\lean{Semirings/BoolFunc}{BoolFunc.isDelta_id}), the latter otherwise.

A \emph{homomorphism of m-semirings}
\lean{SemiringWithMonus}{SemiringWithMonusHom} is a semiring
homomorphism $h:\mathbb K\to\mathbb K'$ that also commutes with
$\ominus$ and with $\delta$.

\begin{toappendix}
  \subsection{Corrections to the Literature}
  \label{appendix:corrections}
  Section~\ref{sec:preliminaries} announced corrections to three
  published claims; for each, we recall the claim, locate the error, and
  prove the correct statement, mirroring the Lean development.

  \paragraph*{The tropical semiring over $\mathbb R$ admits a monus
  \lean{Semirings/Tropical}{instCommSemiringWithMonusMinTropicalWithTopReal}.}
  \cite[Example~4]{geerts2010database} considers
  $(\mathbb R\cup\{+\infty\},\min,+,+\infty,0)$, establishes that it is
  naturally ordered, and asserts that the monus is nonetheless
  ill-defined in it: $x\ominus y$ is to be the least $z$ with
  $x\leq y\oplus z$, and the argument observes that
  $\{z \mid x\leq\min(y,z)\}$ ``is not bounded below since one can take
  arbitrary small values for~$z$''. The set is indeed unbounded below in
  the usual order of the reals, but the least element must be taken in
  the \emph{natural} order $\leq$ of the semiring, which is the reverse
  of the usual order, and it always exists, as we now show (the flaw in the argument was already
  observed in~\cite{amarilli2016example}, and confirmed by an author
  of~\cite{geerts2010database}).

  Write $\mathbb K \defeq \mathbb R\cup\{+\infty\}$ with $a\oplus b\defeq
  \min(a,b)$, $\mathbb 0\defeq+\infty$, $a\otimes b\defeq a+b$ and $\mathbb
  1\defeq 0$; we write $\leq_{\mathbb R}$ for the usual order on $\mathbb K$,
  to distinguish it from the natural order $\leq$ of
  Definition~\ref{def:m-semiring}.

  The natural order is the \emph{reverse} of the usual one. Indeed, $a\leq b$
  holds iff $\min(a,c)=b$ for some $c\in\mathbb K$: if $b\leq_{\mathbb R}a$,
  take $c\defeq b$; and if $b>_{\mathbb R}a$ then $\min(a,c)\leq_{\mathbb
  R}a<_{\mathbb R}b$ for every $c$, so no such $c$ exists. Hence $a\leq b$
  iff $b\leq_{\mathbb R}a$, and therefore the semiring is naturally ordered.

  Now,
  \[
    a\ominus b\defeq
    \begin{cases}
      +\infty & \text{if $a\geq_{\mathbb R}b$,}\\
      a & \text{otherwise.}
    \end{cases}
  \]
  We observe by the residuation law $a\ominus b\leq c\iff a\leq b\oplus c$, which,
  rewritten in the usual order, is $c\leq_{\mathbb R}a\ominus b\iff
  \min(b,c)\leq_{\mathbb R}a$. If $a\geq_{\mathbb R}b$, the left-hand side is
  $c\leq_{\mathbb R}+\infty$ and the right-hand side follows from
  $\min(b,c)\leq_{\mathbb R}b\leq_{\mathbb R}a$, so both hold. If
  $a<_{\mathbb R}b$, the left-hand side is $c\leq_{\mathbb R}a$; when
  $c\leq_{\mathbb R}a$ we get $\min(b,c)\leq_{\mathbb R}c\leq_{\mathbb R}a$,
  and when $c>_{\mathbb R}a$ both $b$ and $c$ exceed $a$, so
  $\min(b,c)>_{\mathbb R}a$. The two sides therefore agree in every case, and
  $\mathbb K$ is an m-semiring.

  Nothing above uses any property of $\mathbb R$ beyond it being a linearly
  ordered commutative additive monoid, so the same construction equips the
  tropical semiring over $\mathbb Q$, over $\mathbb Z$, and over $\mathbb N$, with a monus.

  The erroneous claim has propagated: the initial arXiv version of a
  recent generalization of Codd's theorem to
  semirings~\cite{badia2025coddarxiv} reproduces it, along with the
  argument of~\cite[Example~4]{geerts2010database}; the published
  version~\cite{badia2025codd} states, without further comment, exactly
  the monus above.

  \paragraph*{In $\mathbb N[X]$, $\otimes$ does not distribute over $\ominus$
  \lean{Semirings/How}{How.not_mul_sub_left_distributive}.}
  \cite{amsterdamer2011limitations} axiomatizes difference on annotated
  relations through thirteen axioms A1--A13 on m-semirings, A13 being
  the distributivity of $\otimes$ over~$\ominus$, and its Table~3 lists
  $\mathbb N[X]$ among the m-semirings satisfying A13. It does not.
  Here $\ominus$ is coefficient-wise truncated subtraction: for every
  monomial $M$, $c_M(a\ominus b)=\max\bigl(c_M(a)-c_M(b),0\bigr)$, writing
  $c_M(\cdot)$ for the coefficient of $M$. A single variable $x$ suffices.
  Take $b\defeq x$, $c\defeq\mathbb 1$ and $a\defeq b\oplus c=x+1$. Then
  $b\ominus c=x$, since $c_x(b)-c_x(c)=1-0$ and $c_{\mathbb 1}(b)-c_{\mathbb
  1}(c)=0-1$ truncates to $0$, so
  \[
    a\otimes(b\ominus c)=(x+1)\,x=x^2+x .
  \]
  On the other hand $a\otimes b=x^2+x$ and $a\otimes c=x+1$, so
  \[
    (a\otimes b)\ominus(a\otimes c)=(x^2+x)\ominus(x+1)=x^2 ,
  \]
  the coefficient of $x$ truncating to $1-1=0$. The two results differ in the
  coefficient of $x$.

  \paragraph*{In Why[$X$], $\otimes$ does not distribute over $\ominus$
  \lean{Semirings/Why}{Why.not_mul_sub_left_distributive}.}
  Table~3 of \cite{amsterdamer2011limitations} also lists Why[$X$] among
  the m-semirings satisfying its axiom A13.
  Recall the semiring $(2^{2^X},\cup,\Cup,\emptyset,\{\emptyset\})$ of
  Example~\ref{ex:semirings}, whose monus is set difference, $A\ominus
  B=A\setminus B$. We note that $\mathbb 0=\emptyset$ is the empty set whereas $\mathbb 1=\{\emptyset\}$ is the singleton containing
  the empty set. Again one variable $x$ suffices: take
  $a\defeq\{\{x\}\}$, $b\defeq\{\emptyset\}$ and $c\defeq\{\{x\}\}$. Then
  $b\ominus c=\{\emptyset\}$ and
  \[
    a\otimes(b\ominus c)=\{\{x\}\}\Cup\{\emptyset\}=\{\{x\}\cup\emptyset\}
    =\{\{x\}\} .
  \]
  But $a\otimes b=\{\{x\}\}\Cup\{\emptyset\}=\{\{x\}\}$ and $a\otimes
  c=\{\{x\}\}\Cup\{\{x\}\}=\{\{x\}\cup\{x\}\}=\{\{x\}\}$, so
  \[
    (a\otimes b)\ominus(a\otimes c)=\{\{x\}\}\setminus\{\{x\}\}=\emptyset ,
  \]
and $\{\{x\}\}\neq\emptyset$.
\end{toappendix}

\inlinepar{Useful results on m-semirings.} We discuss some
additional results on m-semirings. First we show that idempotence is equivalent to $\ominus$ being right-distributive over $\oplus$.
\begin{toappendix}
  \subsection{Proofs of the Useful Results on M-Semirings}
  We prove here the three useful results on m-semirings stated in
  Section~\ref{sec:preliminaries}.
\end{toappendix}
\begin{propositionrep}\label{prop:idempotent_iff_add_monus}
  \lean{SemiringWithMonus}{idempotent_iff_add_monus}
  An m-semiring $\mathbb K$ is idempotent if and only if $\ominus$
  is right-distributive over $\oplus$ in $\mathbb K$, i.e., if and only
  if $\forall a, b,
c\in\mathbb K, (a\oplus b)\ominus c=(a\ominus c)\oplus (b\ominus c)$.
\end{propositionrep}
\begin{appendixproof}
  We show both implications.

  \textbf{($\Rightarrow$)}
  \lean{SemiringWithMonus}{add_monus_of_idempotent}
  Assume $\mathbb K$ to be idempotent.

  Consider $a,b \in\mathbb K$ such as $a\leq b$. Then $\exists c, a+c=b$.
  So $a \oplus b = a \oplus (a \oplus c)=(a\oplus a)\oplus c=a \oplus c = b$.
  Conversely, if $a\oplus b = b$, we trivially have $a\leq b$.

  Now consider an element $u\in\mathbb K$ such that $a \leq u$ and $b
  \leq u$, i.e.,
  $a \oplus u = u$ and $b\oplus u =u$. Observe that $a\oplus b\oplus
  u=a\oplus u=u$, so that $a\oplus b\leq u$.

  But since $a\leq a \oplus b$ and $b \leq a \oplus b$,
  $a \oplus b $ is the least such $u$ for which $a\leq u $ and $b \leq u$.
  In other words, $\oplus$ is the join of the semilattice $(\mathbb K,\leq)$.

  \medskip
  We show $(a\oplus b)\ominus c=(a\ominus c)\oplus (b\ominus c)$ by
  showing both inequalities.

  (i) By definition of monus : $x \ominus y \leq z \iff x \leq y \oplus
  z$. Note that $a\leq a\oplus b$ and $a\oplus b\leq c \oplus ((a\oplus
  b)\ominus c)$ by applying the residuation law to the inequality
  $(a\oplus b)\ominus c\leq (a\oplus b)\ominus c$. This means $a\leq
  c\oplus((a\oplus b)\ominus c)$ which implies $a\ominus c\leq (a\oplus b)\ominus c$.
  Similarly, $b\ominus c\leq (a\oplus b)\ominus c$. By using the fact
  that $\oplus$ is the join, we get
  $(a\ominus c)\oplus (b\ominus c) \leq (a\oplus b) \ominus c$.

  (ii) We have $a\leq c \oplus (a\ominus c)$ and $b\leq c\oplus (b\ominus c)$. So, $a\oplus b \leq c\oplus c \oplus (a\ominus c) \oplus (b\ominus c)$.
  By idempotence of $\oplus$,  $a\oplus b \leq c \oplus (a\ominus c)
  \oplus (b\ominus c)$, or by the residuation law
$(a\oplus b) \ominus c \leq (a\ominus c)\oplus (b\ominus c) $.

\bigskip
\textbf{($\Leftarrow$)}
  \lean{SemiringWithMonus}{idempotent_of_add_monus}
  Assume that $\ominus$ is right-distributive over $\oplus$.
For any arbitrary $a\in \mathbb K$, we thus have:
\[
  (a\oplus a)\ominus a
  = (a\ominus a)\oplus(a\ominus a).
\]
Since $a\ominus a=\mathbb 0$, the right-hand side is $\mathbb 0$. And
$(a\oplus a)\ominus a = \mathbb 0\implies a\oplus a\leq a$ by the
residuation law.
  But, trivially, $a\le a\oplus a$ so the antisymmetry of $(\mathbb K,\leq)$ gives $a\oplus a=a$.
  Thus, \(\forall a\in\mathbb K, a \oplus a=a.\)
\end{appendixproof}
\begin{corollaryrep}\label{lem:finite-add-monus}
\lean{SemiringWithMonus}{add_monus_of_idempotent_multiset}
Let \(\mathbb K\) be an
idempotent
m-semiring.
Then, for all $n \ge 1$ and $a_1,\dots,a_n,c \in\nobreak K$,
\(
\left(\bigoplus_{i=1}^n a_i\right)\ominus c
=
\bigoplus_{i=1}^n (a_i \ominus c).
\)
\end{corollaryrep}

\begin{appendixproof}
We proceed by induction on \(n \ge 1\).

For \(n=1\), the identity is trivial.

Assume the identity holds for \(n-1\geq 1\).
Then
\begin{align*}
\Bigl(\bigoplus_{i=1}^n a_i\Bigr)\ominus c
&= \Bigl(\bigl(\bigoplus_{i=1}^{n-1} a_i\bigr)\oplus a_n\Bigr)\ominus c \\
&= \Bigl(\bigoplus_{i=1}^{n-1} a_i\Bigr)\ominus c \;\oplus\; (a_n\ominus c) \\
&= \bigoplus_{i=1}^{n-1} (a_i\ominus c) \;\oplus\; (a_n\ominus c)
= \bigoplus_{i=1}^n (a_i\ominus c),
\end{align*}
where the second equality uses
Proposition~\ref{prop:idempotent_iff_add_monus} and the third uses the
induction hypothesis.
\end{appendixproof}

\noindent We also show elementary consequences of the residual
characterization of $\ominus$ in m-semirings.
\begin{lemmarep}\label{lem:monus-basic}
  In any m-semiring $\mathbb K$, for all $x,y,y'\in\mathbb K$:
  (i)~$x\ominus y \le x$ \lean{SemiringWithMonus}{monus_le};
  (ii)~$x \le y\oplus (x\ominus y)$ \lean{SemiringWithMonus}{le_plus_monus};
  (iii)~if $y\le y'$ then $x\ominus y'\le x\ominus y$
  \lean{SemiringWithMonus}{monus_antitone}.
\end{lemmarep}

\begin{appendixproof}
  (i)~By the residuation law (Definition~\ref{def:m-semiring}) with
  $z=x$, $x\ominus y\le x$ if and only if $x\le y\oplus x$, which holds by
  definition of $\leq$ and commutativity of $\oplus$.
  (ii)~Apply the residuation law with $z=x\ominus y$: since
  $x\ominus y\le x\ominus y$, we obtain $x\le y\oplus(x\ominus y)$.
  (iii)~By~(ii), $x\le y\oplus(x\ominus y)$; since $y\le y'$ and $\oplus$
  is monotone with respect to the natural order,
  $x\le y'\oplus(x\ominus y)$, and the residuation law gives
  $x\ominus y'\le x\ominus y$.
\end{appendixproof}

\inlinepar{Annotated relations and semantics.}
We use the notion of annotated relations initially introduced in
\cite{green2007provenance}, relying on the setting (in particular where
relations are multisets of tuples) and definitions introduced in detail in
\cite[Section IV]{sen2026provsql}.
In short, if $\mathcal{V}$ is the (usually countably infinite) set of
data values used within unannotated relations, given a commutative m-semiring
$\mathbb K$ and a set $\mathcal V_{\mathbb K}\supseteq \mathcal V$ of
values (either~$\mathcal V$ or an extension of it) a
$\mathbb K$-relation of arity $k\in\mathbb N$ over $\mathcal V_{\mathbb
K}$ is a \emph{multiset} of tuples over~$\mathcal V_{\mathbb K}^k\times\mathbb
K$ which we usually denote $(u,\alpha)$ with $u\in \mathcal V_{\mathbb K}^k$ and $\alpha\in \mathbb K$.
We call $\alpha$ the annotation of the tuple $u$. As
in~\cite{green2007provenance}, a tuple annotated by $\mathbb 0$ is
considered absent: two $\mathbb K$-relations that differ only by tuples
annotated by $\mathbb 0$ are identified.
A $\mathbb K$-instance, usually denoted as $\hat I$, over a database
schema $\mathcal D$, is a function mapping each relation symbol $R$ of $\mathcal D$ to a $\mathbb K$-relation over~$\mathcal V_{\mathbb K}$ of arity $\mathcal D(R)$.

Unlike in~\cite{green2007provenance}, where a $\mathbb K$-relation is a
function from tuples to~$\mathbb K$, a tuple may thus occur several
times, with different annotations. Aggregation requires it: an aggregate
is computed over the occurrences of a group and not only over its distinct
values. Over an idempotent semiring such as $\mathbb B[X]$ the
multiplicities cannot be carried by the annotations, as the semimodule
construction of~\cite{amsterdamer2011provenance} then forces the
aggregation monoid to be idempotent as
well~\cite[Proposition~3.11]{amsterdamer2011provenance}, thus excluding
\sql{COUNT} and \sql{SUM}. It also matches the bag semantics of SQL,
under which a projection may produce two occurrences of a tuple with
distinct provenance, and merging them would change the value of
\sql{COUNT(*)}~\cite{sen2026provsql}.

The elements of $\mathcal V$ are the \emph{regular data values}: those
that a term can denote and that a comparison can test directly.
We adopt the unnamed perspective and consider, instead of attribute names, their \emph{positional indices}, denoted by $\#i$ for $i\in \mathbb N^*$.
A \emph{term} is then an expression involving values from $\mathcal V$
and positional indices, and the \emph{max-index} of a term is the largest
positional index appearing in it ($0$ if none appear). Given a tuple
$u=(u_1,\dots,u_k)$ and a term $t$ of max-index $\leq k$ we can define $t(u)\in \mathcal V$ that evaluates the term expression after replacing the positional indices in it with the respective attribute values of the tuple. Given a finite sequence of tuples $(u_{i_1},\dots,u_{i_n})$ we define the pointwise evaluation of a well-defined term $t$ on them as $\pi_{t}((u_{i_1},\dots,u_{i_n}))=(t(u_{i_1}),\dots,t(u_{i_n}))$.

Queries are written in the relational algebra under multiset semantics,
extended with a duplicate elimination operator $\epsilon$
\cite{DBLP:conf/pods/DayalGK82,DBLP:journals/iandc/GrumbachM99}, a
multiset difference operator $-$, and a grouping and aggregation
operator $\gamma$ inspired by \cite{DBLP:journals/tcs/Libkin03}; we write
$\RA_k$ for the set of queries of arity $k\in\mathbb N$.
The semantics of a query $q$ over annotated relations $\hat I$, denoted as
$\angsem{\hat I}{q}$, is defined over $\mathbb K$-relations by induction
on these operators, following \cite{sen2026provsql}, which extends
\cite{green2007provenance,geerts2010database,amsterdamer2011provenance}.
The full syntax of $\RA_k$ and the definition of $\angsem{\hat
I}{\cdot}$ on every operator are recalled in the appendix; we describe
here only the
aggregation operator, on which this paper builds. Since we allow
aggregate functions that are not commutative (Section~\ref{sec:semantics}),
the operator carries a total order~$\preceq$ on tuples fixing the order
in which the values of a group are aggregated, as the standard SQL
``\sql{AGG(...) WITHIN GROUP (ORDER BY ...)}'' clause or PostgreSQL's
``\sql{AGG(... ORDER BY ...)}'' do; the order is irrelevant for
commutative aggregates.
\smallskip
\begin{description}
  \item[syntax] for $k\in\mathbb N$, $q\in\RA_k$, distinct grouping
    indices $(i_j)_{1\leq j\leq m}$ with $1\leq i_j\leq k$, terms
    $t_1,\dots,t_n$ of max-index $\leq k$, aggregate functions
    $f_1,\dots,f_n$ from finite sequences of values to values, and a
    total order $\preceq$ on $\mathcal V^k$,
    $\gamma^{\preceq}_{i_1,\dots,i_m}[t_1{:}f_1,\dots,t_n{:}f_n](q)\in\RA_{m+n}$;
  \item[semantics] for a group key $\vec v\defeq (v_1,\dots,v_m)$, let
    $U^{\preceq}_{\angsem{\hat I}{q},\vec v}$, abbreviated~$U$, be the
    sequence of the tuples $(u,\alpha)\in\angsem{\hat I}{q}$ such that
    $(u_{i_1},\dots,u_{i_m})=\vec v$, taken with their multiplicity and
    ordered by~$\preceq$ on~$u$ (occurrences of equal tuples in an
    arbitrary order; the grouping indices are those of the operator). Then
    $\angsem{\hat I}{\gamma^{\preceq}_{i_1,\dots,i_m}[t_1{:}f_1,\dots,t_n{:}f_n](q)}$
    is defined as
    \[
      \Bigl\{\,\bigl(\vec v,\lAngle f_1,t_1,U\rAngle,\dots,
      \lAngle f_n,t_n,U\rAngle,\delta(\beta)\bigr)\ \Big|\
      (\vec v,\beta)\in\angsem{\hat
      I}{\epsilon(\Pi_{\#i_1,\dots,\#i_m}(q))}\,\Bigr\}
    \]
    where the \emph{aggregate value} $\lAngle f_l,t_l,U\rAngle$ is a
    formal object recording the aggregate function $f_l$ together with
    the annotated sequence $\bigl((t_l(u),\alpha)\bigr)_{(u,\alpha)\in U}$
    of the values it applies to. The group key is annotated by
    $\delta(\beta)$, with $\beta$ the $\oplus$-sum of the annotations of
    the tuples of the group: the group exists exactly when one of its
    tuples does.
\end{description}
An aggregate value is not a value of~$\mathcal V$ and cannot be compared
directly: which values are actually aggregated depends on which
occurrences of the group are present, and this is what the annotations
record. In \cite{sen2026provsql}, following
\cite{amsterdamer2011provenance}, aggregate functions are commutative
monoid operations and the aggregate value is the element
$\bigoplus_{(u,\alpha)\in U}t_l(u)*\alpha$ of a semimodule, ``$*$''
being a tensor product combining a value with an annotation; that element
is determined by $\lAngle f_l,t_l,U\rAngle$.
We take $\mathcal V_{\mathbb K}$ to be $\mathcal V$ extended with the
aggregate values, and call an attribute, a term, or a side of a
comparison \emph{regular} when it only involves values of~$\mathcal V$,
and \emph{aggregate} otherwise; in the absence of aggregation, $\mathcal
V_{\mathbb K}=\mathcal V$ and everything is regular.

Every operator of this algebra retains in this paper the semantics of
\cite{sen2026provsql}, up to this representation of aggregate values. The
\emph{only} operator whose
semantics we change is \emph{selection}: where
\cite{sen2026provsql} defines
\(
  \angsem{\hat I}{\sigma_\phi(q)}\defeq
  \mset{(u,\alpha)\mid (u,\alpha)\in\angsem{\hat I}{q},\ \phi(u)},
\)
a filter that can only test regular data values,
Section~\ref{sec:semantics} gives a \emph{possible-world semantics} for
predicates $\psi$ that may in addition compare \emph{aggregate} values,
i.e., \sql{HAVING} conditions, assigning them a $\mathbb K$-annotation
instead of filtering on them. A predicate $\phi$ of the algebra above is
exactly a $\psi$ with no aggregate comparison; the extension is
conservative, such that, on those, the two definitions coincide.
We use three classes of queries throughout. An \emph{aggregate-free
query} uses no grouping and aggregation. An \emph{aggregate query} may
use it, but only as its last operator (\emph{terminal aggregation}),
which is the case covered by \cite{sen2026provsql};
terminal
aggregation excludes any selection, join or further grouping on an
aggregate value, and in particular \sql{HAVING}, an ordinary part of SQL
(two of the 22 TPC-H queries~\cite{tpch}, Q11 and Q18, use it).
A
\emph{\sql{HAVING} query} may in addition apply a selection to aggregate
expressions (\sql{HAVING} conditions), comparing terms over the
aggregate values of a group with constants, with the group key, or with
each other.
Giving a provenance
semantics to \sql{HAVING} queries is the object of this paper.

\begin{toappendix}
  \subsection{Relational Algebra with Aggregation and its Semantics}
  \label{appendix:algebra}

  We recall in this subsection, for self-containedness, the relational
  algebra with aggregation of \cite{sen2026provsql} and its semantics over
  ordinary and over annotated relations. All of this material is from
  \cite{sen2026provsql}, itself building
  on \cite{green2007provenance,geerts2010database,amsterdamer2011provenance};
  Section~\ref{sec:semantics} changes only the semantics of the selection
  operator.

  Given a database schema $\mathcal{D}$, the language $\RA_k$ of
  queries of arity $k\in\mathbb{N}$ is
  defined recursively as follows (\lean{Query}{Query} for the
  aggregate-free fragment, \lean{AggQuery}{AggQuery} with the aggregation
  operator):
  \begin{description}
    \item[relation] for any relation $R$ in the domain of $\mathcal{D}$,
      $R\in\RA_{\mathcal{D}(R)}$;
    \item[projection] for $k\in\mathbb N$, $q\in\RA_{k}$,
      $t_1,\dots, t_n$ terms of max-index~\mbox{$\leq k$},
      $\Pi_{t_1,\dots,t_n}(q)\in\RA_n$;
    \item[selection] for $k\in\mathbb N$, $q\in\RA_k$, and $\phi$ a
      Boolean combination of (in)equality comparisons between terms
      of max-index $\leq k$, $\sigma_\phi(q)\in\RA_k$;
    \item[cross product] for $k_1,k_2\in\mathbb N$, $q_1\in\RA_{k_1}$,
      $q_2\in\RA_{k_2}$, $q_1\times q_2\in\RA_{k_1+k_2}$;
    \item[multiset sum] for $k\in\mathbb N$, $q_1,q_2\in\RA_k$,
      $q_1\uplus q_2\in\RA_k$;
    \item[duplicate elimination] for $k\in\mathbb N$,
      $q\in\RA_k$, $\epsilon(q)\in\RA_k$;
    \item[multiset difference] for $k\in\mathbb N$, $q_1,q_2\in\RA_k$,
      $q_1 - q_2\in\RA_k$;
    \item[aggregation] for $k\in\mathbb N$, $q\in\RA_k$, distinct
      $(i_j)_{1\leq j\leq m}$ with $1\leq i_j\leq k$, terms $t_1,\dots,
      t_n$ of max-index $\leq k$, functions $f_1,\dots,f_n$ from finite
      sequences of values to values, and a total order $\preceq$ on
      $\mathcal V^k$,
      $\gamma^{\preceq}_{i_1,\dots,i_m}[t_1{:}f_1,\dots,t_n{:}f_n](q)\in\RA_{m+n}$.
  \end{description}
  Some other operators are syntactic sugar:
  the \emph{join} $q_1\bowtie_\phi q_2 \defeq \sigma_\phi(q_1\times q_2)$ and
  the \emph{set union} $q_1\cup q_2\defeq \epsilon(q_1\uplus q_2)$.
  Note that the multiset difference above is neither the usual one nor
  the one of the \sql{EXCEPT ALL} operator of SQL, but the one of the
  \sql{NOT IN} operator: the standard definition
  $\ansem{I}{q_1-q_2}(t)=\max(0,\ansem{I}{q_1}(t)-\ansem{I}{q_2}(t))$
  makes provenance computation intractable. The difference disappears
  under set semantics: $\epsilon(q_1-q_2)$ matches \sql{EXCEPT}.

  \paragraph*{Semantics over ordinary relations.} The semantics
  $\ansem{I}{\cdot}$ of these queries on an ordinary instance $I$
  over $\mathcal{D}$, where a relation of arity $k$ is a finite multiset
  of tuples of $\mathcal{V}^k$, is defined by induction
  \lean{Query}{Query.evaluate} \lean{AggQuery}{AggQuery.evaluatePlain}:
  \begin{description}
    \item[relation] $\ansem{I}{R}\defeq I(R)$;
    \item[projection] $\ansem{I}{\Pi_{t_1,\dots,t_n}(q)}\defeq
      \mset{(t_1(u),\dots,t_n(u))\mid u\in\ansem{I}{q}}$;
    \item[selection] $\ansem{I}{\sigma_\phi(q)}\defeq
      \mset{u\mid u\in\ansem{I}{q},\ \phi(u)}$;
    \item[cross product] $\ansem{I}{q_1\times q_2} \defeq
      \ansem{I}{q_1}\times\ansem{I}{q_2}$;
    \item[multiset sum] $\ansem{I}{q_1\uplus
      q_2}\defeq\ansem{I}{q_1}\uplus\ansem{I}{q_2}$;
    \item[duplicate elimination] $\ansem{I}{\epsilon(q)}$ maps $u$ to $1$
      if $\ansem{I}{q}(u)>0$, and to $0$ otherwise;
    \item[multiset difference] $\ansem{I}{q_1-q_2}$ maps $u$ to $0$ if
      $\ansem{I}{q_2}(u)>0$, and to $\ansem{I}{q_1}(u)$ otherwise;
    \item[aggregation] with $\vec v\defeq(v_1,\dots,v_m)$ and
      $B_l(\vec v)$ the sequence of the values $t_l(u)$ for
      $u\in\ansem{I}{q}$ with $(u_{i_1},\dots,u_{i_m})=\vec v$, taken
      with their multiplicity and ordered by $\preceq$,
      \begin{align*}
        &\ansem{I}{\gamma^{\preceq}_{i_1,\dots,i_m}[t_1{:}f_1,\dots,t_n{:}f_n](q)}\defeq\\
        &\quad\Bigl\{\,\bigl(\vec v,f_1(B_1(\vec
        v)),\dots,f_n(B_n(\vec v))\bigr)\ \Big|\
        \vec v\in\ansem{I}{\epsilon(\Pi_{\#i_1,\dots,\#i_m}(q))}\,\Bigr\}.
      \end{align*}
  \end{description}

  \paragraph*{Semantics over annotated relations.} The semantics
  $\angsem{\hat I}{\cdot}$ of queries on a
  $\mathbb{K}$-instance $\hat I$ over $\mathcal{D}$ is defined by
  induction as follows; for the definition to be meaningful, $\mathbb{K}$
  needs to be a semiring, an m-semiring for the multiset difference
  operator, and a $\delta$-semiring for the aggregation operator (we say
  that $\mathbb{K}$ is \emph{appropriate} for $q$ if this is the case)
  \lean{QueryAnnotatedDatabase}{Query.evaluateAnnotated}
  \lean{AggQuery}{AggQuery.evaluateAnnotated}:
  \begin{description}
    \item[relation] $\angsem{\hat I}{R}\defeq\hat I(R)$;
    \item[projection] $\angsem{\hat I}{\Pi_{t_1,\dots,t_n}(q)}\defeq
      \mset{(t_1(u),\dots,t_n(u), \alpha) \mid (u,\alpha)\in\angsem{\hat
      I}{q}}$;
    \item[selection] $\angsem{\hat I}{\sigma_\phi(q)}\defeq
      \mset{(u,\alpha)\mid (u,\alpha)\in\angsem{\hat I}{q},\ \phi(u)}$
      \emph{(this is the operator whose semantics
      Section~\ref{sec:semantics} extends)};
    \item[cross product] $\angsem{\hat I}{q_1\times q_2}
      \defeq\mset{(u,v,\alpha_1\otimes\alpha_2) \mid
      (u,\alpha_1)\in\angsem{\hat I}{q_1},\ (v,\alpha_2)\in\angsem{\hat
      I}{q_2}}$;
    \item[multiset sum] $\angsem{\hat I}{q_1\uplus
      q_2}\defeq\angsem{\hat I}{q_1}\uplus\angsem{\hat I}{q_2}$;
    \item[duplicate elimination]
      \(
        \angsem{\hat
        I}{\epsilon(q)}\defeq\bigcup_{u\mid\exists\alpha\,(u,\alpha)\in
        \angsem{\hat I}{q}}\bigl\{\,\bigl(u,\bigoplus_{\alpha\mid(u,\alpha)\in\angsem{\hat
        I}{q}}\alpha\bigr)\,\bigr\};
      \)
    \item[multiset difference]
      \(
        \angsem{\hat I}{q_1-q_2}\defeq
        \mset{\bigl(u,\alpha\ominus\bigoplus_{\beta\mid(u,\beta)\in\angsem{\hat
        I}{q_2}}\beta\bigr)\mid(u,\alpha)\in \angsem{\hat I}{q_1}};
      \)
    \item[aggregation] as in Section~\ref{sec:preliminaries}: with
      $\vec v\defeq(v_1,\dots,v_m)$ and $U\defeq U^{\preceq}_{\angsem{\hat
      I}{q},\vec v}$,
      \begin{align*}
        &\angsem{\hat
        I}{\gamma^{\preceq}_{i_1,\dots,i_m}[t_1{:}f_1,\dots,t_n{:}f_n](q)}\defeq\\
        &\quad\Bigl\{\,\bigl(\vec v,\lAngle f_1,t_1,U\rAngle,\dots,
        \lAngle f_n,t_n,U\rAngle,\delta(\beta)\bigr)\ \Big|\
        (\vec v,\beta)\in\angsem{\hat
        I}{\epsilon(\Pi_{\#i_1,\dots,\#i_m}(q))}\,\Bigr\}.
      \end{align*}
      In \cite{sen2026provsql}, the $f_l$ are commutative monoid
      operations, $\gamma$ carries no order, and the aggregate value is the
      semimodule element $\hat f_l\bigl(\mset{t_l(u)*\alpha \mid
      (u,\alpha)\in U}\bigr)$, where ``$*$'' denotes the tensor product
      combining data values of $\mathcal{V}$ with $\delta$-semiring
      annotations from $\mathbb{K}$ in a semimodule
      $\mathcal{V}_{\mathbb{K}}$, and $\hat f_l$ is $f_l$ lifted to
      $\mathcal{V}_\mathbb{K}$; see~\cite{amsterdamer2011provenance} for
      details about provenance semimodules.
  \end{description}
\end{toappendix}

\section{A Possible-World Semantics for Aggregate Comparisons}
\label{sec:semantics}

In this section we give a possible-world semantics for selection
predicates comparing aggregate values, in arbitrary commutative
m-semirings. Unlike~\cite{amsterdamer2011provenance}, where aggregates
are commutative monoid operations lifted to a semimodule, which excludes
non-associative operations such as \sql{AVG} and non-commutative ones
such as \fakesql{ARRAY\_AGG}, our semantics does not depend on the
semimodule structure and accepts any function on \emph{sequences} of
values.

We fix a commutative m-semiring
$(\mathbb{K},\oplus,\otimes,\ominus,\mathbb{0},\mathbb{1})$, as discussed in the preliminaries, along with a data domain $\mathcal{V}$.
For any set~$S$, we write $\Seq(S)$ the set of finite sequences over~$S$.
We define aggregate functions and comparison operators:
\begin{definition}
An \emph{aggregate function} is a function  $f:\Seq(\mathcal{V})\to\mathcal{V}$.
  For any comparison operator ${\op}\in\{<,\leq,=,\geq,>,\neq\}$ defined
  over $\mathcal{V}$, i.e., $\op: \mathcal{V}^2\rightarrow
  \{\bot,\top\}$, we denote
\(
  \chi_{\op}:\mathcal{V}\times\mathcal{V}\to \mathbb{K}
  \) as the mapping that sends $(a,b)$ to $\mathbb{1}$ if $a\op b$, and to
  $\mathbb{0}$ otherwise.
\end{definition}

\begin{example}\label{ex:aggregates}
Aggregates such as \sql{SUM}, \sql{COUNT}, \sql{MIN} and \sql{MAX}
are \emph{commutative}: their value is unchanged under any reordering of
the input sequence, as for instance the \sql{SUM} aggregate
$f_\Sigma\bigl((x_1,\dots,x_n)\bigr)\defeq x_1+\dots+x_n$.
But
$\fakesql{PICKFIRST}\bigl((x_1,\dots,x_n)\bigr)\defeq x_1$, returning
the first value of its input, is \emph{non-commutative}: its result
depends on the order.
\end{example}

A \sql{HAVING} predicate cannot be evaluated on the annotations of a
group taken as a whole. The annotations can be combined by $\oplus$ and $\otimes$
into a single element of $\mathbb K$, but no element of $\mathbb K$
records how many occurrences are counted or which values are summed, so
there is nothing to which, e.g., $\text{\sql{COUNT(*)}}\geq 2$ could be applied. This
is why~\cite{amsterdamer2011provenance} leaves comparisons
uninterpreted, as formal expressions over semimodule elements. To give the predicate
a value in $\mathbb K$ instead, we split the group into its
\emph{possible worlds}. A world is a subsequence $W$ of the occurrences of
the group, on which the aggregate is an ordinary value and the comparison
an ordinary test; its annotation states that the occurrences of $W$ are
present and the others are not, the latter through the monus. The
predicate provenance is the $\oplus$-sum of the annotations of the worlds
satisfying the comparison.
We extend the semantics of the \emph{selection} operator accordingly, and
first define possible worlds of the sequence of occurrences of a group:
\begin{definition}\label{def:possible_world}
  Let $M$ be a $\mathbb K$-relation of arity~$k$ over $\mathcal V$,
  $(i_j)_{1\leq j \leq m}$ be distinct grouping indices, $\preceq$ a
  total order on $\mathcal{V}^k$, $\vec v\in\mathcal V^m$ a group key,
  and $U\defeq U^{\preceq}_{M,\vec v}$ the $\preceq$-ordered sequence of
  the occurrences of the group (Section~\ref{sec:preliminaries}).
  A \emph{possible world}~$W$ of $U$ is a subsequence
  of $U$, denoted $W\sqsubseteq U$. Its \emph{$\mathbb{K}$-annotation} is:
\(
\ann_U(W) \defeq
\Bigl( \bigotimes_{(u,\alpha)\in W} \alpha\Bigr)
 \otimes
\Bigl( \mathbb{1} \ominus \bigoplus_{(u,\alpha)\in U\setminus W}
  \alpha\Bigr)
\)
where $U\setminus W$ denotes the occurrences of $U$ not retained
in~$W$, and empty $\bigotimes$-products (resp.\ $\bigoplus$-sums) equal
$\mathbb{1}$ (resp.\ $\mathbb{0}$). A tuple occurring several times
in~$U$ contributes one factor per occurrence; the order of the
$\bigotimes$-product is immaterial, $\otimes$ being commutative.
Moreover, given a term $t$ and an aggregate function $f$, the
\emph{aggregate value} of $f$ over $t$ in~$W$ is
\(
\mathrm{agg}_{t,f}(W)\defeq f\bigl((t(u))_{(u,\alpha)\in W}\bigr),
\)
i.e., $f$ applied to the sequence of $t$-values of~$W$.
\end{definition}

We now give the semantics of \emph{selection} over comparison predicates.
A term of a query evaluates, on a tuple, to an element of the value
domain $\mathcal V_{\mathbb K}\supseteq\mathcal V$ of $\mathbb
K$-relations (Section~\ref{sec:preliminaries}): besides the
\emph{regular} values of $\mathcal V$, this domain contains the
\emph{non-regular} values produced by aggregation, i.e., aggregate values
represented as formal aggregations of sequences of pairs formed of a
value and a semiring element (the straightforward extension of the
semimodule values of~\cite{amsterdamer2011provenance}). For simplicity
of presentation we give
the semantics of a single comparison applied directly to a grouping; we
then discuss the
general case.

\begin{definition}
\label{def:semantics}
\lean{HavingSemantics}{Having.havingProv}
Let $\hat I$ be a $\mathbb K$-instance, $k\in\NN$, $q$ an arity-$k$
query, and let
$\psi$ be a selection predicate over arity~$k$. To $\psi$ we associate a
\emph{predicate provenance} $\predsem{\hat I}{q}{\psi}\colon\mathcal V_{\mathbb K}^k\to\mathbb K$.
An atomic comparison $t\op t'$ is \emph{regular} when both $t$ and $t'$
are of regular type, i.e., mention only regular columns, and is an
\emph{aggregate comparison} as soon as one of them mentions an aggregate
column.

If $\psi=(t\op t')$ is regular, then
\(\predsem{\hat I}{q}{\psi}(u)\defeq \chi_{\op}\bigl(t(u),t'(u)\bigr).\)
If $\psi=(t\op t')$ is an aggregate comparison, we
require $q$ to be a grouping query
$q=\gamma^{\preceq}_{i_1,\dots,i_m}[t_1{:}f_1,\dots,t_n{:}f_n](q')$ with
$q'$ an arity-$k'$ \emph{aggregate-free} query, i.e., one without
grouping and aggregation, whose order $\preceq$ is the one used for non-commutative
aggregates (and is irrelevant otherwise); its output column $m+l$ holds
the aggregate value of $f_l$ over $t_l$, for $1\le l\le n$. For an output tuple
$u=(\vec v,b_1,\dots,b_n)$ with group key $\vec v=(v_1,\dots,v_m)$, let
$U\defeq U^{\preceq}_{\angsem{\hat I}{q'},\,\vec v}$. For a possible
world $W\sqsubseteq U$, the \emph{value in~$W$} of a side
$s\in\{t,t'\}$, $\mathrm{val}_s(W)$, is $s$ evaluated on $u$ after
replacing each aggregate value $b_l$ by $\mathrm{agg}_{t_l,f_l}(W)$; in
particular, $\mathrm{val}_s(W)=\mathrm{agg}_{t_l,f_l}(W)$ if $s$ is the
aggregate column $\#(m+l)$, and $\mathrm{val}_s(W)=s(u)$ if $s$ is
regular.
Then
\(
\predsem{\hat I}{q}{\psi}(u)\defeq
\bigoplus_{\emptyset\neq W\sqsubseteq U}
\ann_U(W)\otimes
\chi_{\op}\bigl(\mathrm{val}_t(W),\,\mathrm{val}_{t'}(W)\bigr).
\) We then define:
\[
\angsem{\hat I}{\sigma_\psi(q)}\defeq
\begin{cases}
\mset{(u,\ \alpha)\mid(u,\alpha)\in\angsem{\hat I}{q}, \psi(u)}
  & \text{if $\psi$ has no aggregate comparison,}\\[1ex]
\mset{(u,\ \predsem{\hat I}{q}{\psi}(u))\mid(u,\alpha)\in\angsem{\hat I}{q}}
  & \text{otherwise.}
\end{cases}
\]
\end{definition}
The first case is identical to the selection semantics
of~\cite{sen2026provsql}.
In the second case, the possible-world sum ranges over non-empty worlds only, so it
already enforces group existence (the role of the $\delta(\cdot)$-type
factor $\alpha$), making the extra $\otimes$ by $\alpha$ redundant by
the $\delta$-absorption axiom of Section~\ref{sec:preliminaries}. The
annotation~$\alpha$ is thus dropped. A
group for which no world satisfies $\psi$ receives the annotation
$\mathbb 0$, i.e., is absent from the result
(Section~\ref{sec:preliminaries}).

\begin{example}\label{ex:semantics}
Consider a group with key $\vec v$ with pre-aggregation occurrences
$U=\bigl((u_1,x_1),(u_2,x_2),$ $(u_3,x_3)\bigr)$, annotated in the
Boolean-function semiring $\mathbb B[X]$ by distinct variables
$x_1,x_2,x_3$; the grouping outputs for this group the tuple
$u=(\vec v,3)$, on which the predicate is evaluated. Take the predicate
$\psi=\bigl(\text{\sql{COUNT(*)}}\ge 2\bigr)$, for which the aggregate value on a
world $W$ is $|W|$. The non-empty
worlds $W\sqsubseteq U$ with $|W|\ge 2$ are the three pairs and the full
triple (the others contribute $\chi_{\ge}=\mathbb 0$); for each, $\ann_U(W)$
multiplies the annotations of the kept occurrences by the factor
$\mathbb 1\ominus(\cdot)$ over the discarded ones:
\begin{align*}
  \ann_U(\{1,2\})&=x_1\land x_2\land\lnot x_3, &
  \ann_U(\{1,3\})&=x_1\land x_3\land\lnot x_2,\\
  \ann_U(\{2,3\})&=x_2\land x_3\land\lnot x_1, &
  \ann_U(\{1,2,3\})&=x_1\land x_2\land x_3.
\end{align*}
Their $\oplus$-sum is
$\predsem{\hat I}{q}{\psi}(u)=
(x_1\land x_2)\lor(x_1\land x_3)\lor(x_2\land x_3)$,
the Boolean function ``at least two of $x_1,x_2,x_3$ hold'', as one would
expect of $\text{\sql{COUNT(*)}}\ge 2$. Notice that the negative factors
$\lnot x_i$ have cancelled out: since $\mathbb B[X]$ is absorptive, the
provenance is already captured by the \emph{minimal} satisfying worlds
(here the three pairs), with no use of $\ominus$. This is precisely the
property exploited by the results of Section~\ref{sec:results} and the
polynomial-time algorithms of Section~\ref{sec:algorithms}; in a
non-absorptive semiring the discarded-occurrence factors do not cancel
and this simplification fails. In~\cite{amsterdamer2011provenance}, the
same comparison is the formal element
$[x_1*1+x_2*1+x_3*1\geq\mathbb 1*2]$ of a semiring extended with such
comparison expressions, $*$ being the tensor of
Section~\ref{sec:preliminaries}; it resolves to $\mathbb 0$ or
$\mathbb 1$ only when the semimodule $\mathbb K\otimes M$ is isomorphic to
the aggregation monoid~$M$ (e.g., for $\mathbb N$ and \sql{COUNT}), and
remains a symbol otherwise, in particular in $\mathbb B[X]$.
\end{example}

\inlinepar{Boolean combinations.}
\sql{COUNT(*) >= 2 AND SUM(a) < 10} is a \sql{HAVING} clause too; its
connectives are interpreted by the operations used for join, union and
difference, which in $\mathbb B[X]$ are conjunction, disjunction and
negation. For selection predicates
$\psi_1,\psi_2$ and $u\in\mathcal V_{\mathbb K}^k$,
\[
\predsem{\hat I}{q}{\psi_1\land\psi_2}(u)\defeq
  \predsem{\hat I}{q}{\psi_1}(u)\otimes\predsem{\hat I}{q}{\psi_2}(u);
\qquad
\predsem{\hat I}{q}{\psi_1\lor\psi_2}(u)\defeq
  \predsem{\hat I}{q}{\psi_1}(u)\oplus\predsem{\hat I}{q}{\psi_2}(u).
\]
Negation is pushed to the atoms: $\neg$ traverses $\land$ and $\lor$ by
De Morgan duality, and on an atomic comparison
$\neg(t\op t')\defeq(t\,\overline{\op}\,t')$, where $\overline{\op}$
is the complementary operator ($\overline{\le}$ is~$>$).

\inlinepar{Selections not immediately following aggregation.}
Definition~\ref{def:semantics} covers a comparison applied directly to a
grouping, which is the typical \sql{HAVING}. More generally, a non-regular
aggregate value produced by some $\gamma^{\preceq}$ may be carried, as an
element of $\mathcal V_{\mathbb K}$, through further operators (projection,
join, union, additional selections) before being compared;
our implementation materializes the aggregate value as a gate in
the provenance circuit and evaluates the comparison downstream. The
intended semantics is unchanged: the possible worlds relevant to such a
comparison are still those of the group of the originating
$\gamma^{\preceq}$. We give in the appendix a general semantics for
arbitrary \sql{HAVING} queries that extends
Definition~\ref{def:semantics}, along with the query rewriting
on provenance circuits.

\begin{toappendix}
\label{appendix:semantics}
This appendix defines the general semantics announced in
Section~\ref{sec:semantics}, shows that it coincides with
Definition~\ref{def:semantics} on terminal comparisons, and gives the
corresponding rewriting on provenance circuits.
Three ingredients make the general semantics precise: aggregate values
carry the occurrences of the group that produced them, so that a
comparison performed arbitrarily far downstream can still range over that
group's possible worlds; the columns holding them are distinguished from
regular ones by a \emph{kind}, so that the scope restrictions become
typing conditions rather than side conditions on queries; and the
group-existence factor $\delta(\cdot)$ of a group is kept \emph{pending}
in the annotation until a downstream comparison supersedes it, as in the
second case of Definition~\ref{def:semantics}.

\begin{definition}[General semantics]\label{def:general}
\lean{AggQuery}{AggQuery.evaluate}
Every output column of a query has a \emph{kind}: \emph{regular} (a value
of~$\mathcal V$),
\emph{aggregate} (a token, see below), or \emph{provenance} (a value used
as an annotation), mirroring the regular, aggregate-token and \sql{uuid}
column types of ProvSQL. Queries are indexed by the vector of kinds of
their columns, so that the scope restrictions
are typing conditions: $\gamma^{\preceq}$ groups and aggregates an
all-regular input, $\varepsilon$ and $-$ apply to all-regular relations,
projection copies aggregate columns verbatim, and a selection atom is
either a comparison between terms over regular columns, or a comparison
of one bare aggregate column against a term over regular columns. This
is the form we formalize; an atom whose sides are terms over several
aggregate columns of the same grouping, as allowed by
Definition~\ref{def:semantics}, is evaluated by the same possible-world
sum, all its tokens carrying the same occurrence sequence~$U$.

An \emph{aggregate token} is a pair
$\lAngle f, U\rAngle$ of an aggregate function and the $\preceq$-ordered
occurrence sequence of the group that produced it, each occurrence
recording the value $t(u)$ of the aggregated term together with its
annotation; its value in a world $W\sqsubseteq U$ is
$\mathrm{agg}_{t,f}(W)$.

The annotation of a tuple is kept in \emph{factored} form
$\lAngle\beta,P\rAngle$, with $\beta\in\mathbb K$
and $P$ a multiset of occurrence sequences, one for each group whose
tokens have not been compared yet. It denotes the element
\(
  \beta\otimes\bigotimes_{U\in P}
    \delta\bigl(\bigoplus_{(u,\alpha)\in U}\alpha\bigr)
\)
of $\mathbb K$. Operators act as follows:
\begin{itemize}
\item $\gamma^{\preceq}_{i_1,\dots,i_m}[t_1{:}f_1,\dots,t_n{:}f_n]$
  produces one tuple per group, carrying the group key followed by the
  tokens $\lAngle f_l,U\rAngle$, annotated by $\lAngle\mathbb 1,\{U\}\rAngle$:
  as long as it is not compared, a group is annotated by
  $\delta(\bigoplus_{(u,\alpha)\in U}\alpha)$, which is the terminal-aggregation
  semantics of~\cite{sen2026provsql}.
\item A selection $\sigma_\psi$ whose predicate has no aggregate atom
  filters classically. Otherwise it multiplies the predicate provenance
  $\predsem{\hat I}{q}{\psi}$ into~$\beta$, an aggregate atom
  contributing the possible-world sum of Definition~\ref{def:semantics}
  over the occurrence sequence of \emph{its own} token, and Boolean
  connectives being
  interpreted as above. It removes an
  occurrence sequence $U$ from~$P$ exactly when every token compared
  by~$\psi$ has occurrence sequence~$U$ and $\psi$ \emph{entails the
  existence} of the group: the
  possible-world sum then already ranges over the non-empty worlds
  of~$U$, so it subsumes the factor $\delta(\bigoplus_{(u,\alpha)\in
  U}\alpha)$, which the $\delta$-absorption axiom of Section~\ref{sec:preliminaries}
  makes redundant, and which would otherwise be counted twice in a
  non-idempotent m-semiring. A predicate comparing tokens of several
  groups, or one that can hold in a world where the group is empty --
  an aggregate atom disjoined with a regular atom -- entails no such
  existence and supersedes nothing.
\item The other operators combine the concrete parts~$\beta$ exactly as
  in~\cite{sen2026provsql} and concatenate the pending sequences; a
  projection dropping the last copy of an aggregate column multiplies the
  corresponding pending factor into~$\beta$, that group being no longer
  comparable.
\end{itemize}
\end{definition}

We call \emph{site} a selection with an aggregate comparison applied
directly to a grouping, $\sigma_\psi(\gamma^{\preceq}(q'))$; a comparison
occurring further downstream has the same predicate provenance as the
site formed by that comparison and its originating grouping.
Definition~\ref{def:general} is a conservative extension of
Definition~\ref{def:semantics}, in that the two agree on every site.
\begin{proposition}\label{prop:general-agrees}
  \lean{AggQueryBridges}{AggQuery.havingSite_evaluateAnnotated}
  Let $\psi$ be a single aggregate comparison and
  \[
    Q=\sigma_\psi\bigl(
        \gamma^{\preceq}_{i_1,\dots,i_m}[t_1{:}f_1,\dots,t_n{:}f_n](q')
      \bigr).
  \]
  Then, under Definition~\ref{def:general}, $\angsem{\hat I}{Q}$ consists
  of one tuple per group key of $\angsem{\hat I}{q'}$, carrying the key
  followed by the whole-group aggregate values, annotated by
  $\predsem{\hat I}{\gamma^{\preceq}(q')}{\psi}$ as given by
  Definition~\ref{def:semantics}.
\end{proposition}
\begin{proof}
  A single aggregate atom compares one token, whose occurrence sequence
  is that of the pending factor $\{U\}$ introduced by $\gamma^{\preceq}$,
  and an aggregate atom entails the existence of its group; the factor is
  therefore superseded and the annotation is the predicate provenance
  alone. The aggregate columns are read at the whole-group world, which
  gives the aggregate values of Definition~\ref{def:semantics}.
\end{proof}

Consequently every result of Section~\ref{sec:results}, stated for a
comparison applied to a grouping, applies verbatim to the corresponding
site of an arbitrary query.

\paragraph*{Materializing a site in the provenance circuit.} The
implementation does not compute the possible-world sum eagerly. Instead, it
rewrites the query so that the aggregate columns hold tokens built by
ProvSQL's aggregate-token constructor and the provenance column of a
grouping holds its group-existence guard. A predicate $\psi$ with at
least one aggregate atom is then compiled into a
\emph{gate term}~$\llbracket\psi\rrbracket$ where an aggregate atom becomes a
comparison gate applied to the token and to the compared regular term, a
regular atom becomes an indicator gate $\chi_{\op}$, $\land$ becomes
$\otimes$, $\lor$ becomes $\oplus$, and $\neg$ is pushed to the atoms
with complementation of the comparison operator. The provenance column of
the rewritten site is then
$\llbracket\psi\rrbracket$ alone when $\psi$ entails the existence of the
group, and $\llbracket\psi\rrbracket$ times the guard otherwise -- the
same dichotomy as in Definition~\ref{def:general}, now realized on gates.
\begin{proposition}\label{prop:site-rewriting}
  \lean{AggQueryClosure}{AggQuery.havingPredRew_valid}
  For an arbitrary selection predicate with at least one aggregate atom,
  regular atoms mixed in included, the rewritten site computes the
  annotation that Definition~\ref{def:general} assigns to
  $\sigma_\psi(\gamma^{\preceq}(q'))$, relative to the comparison and
  indicator gates.
\end{proposition}
\noindent
Correctness relative to the two gate primitives is exactly the sense in
which a rewriting-based system such as ProvSQL can be correct: the gates
are the primitives that the evaluation of the circuit, and hence
Section~\ref{sec:algorithms}, has to implement.
\end{toappendix}

\inlinepar{Limitations.} Grouping and ordering only make sense for
regular values from $\mathcal{V}$, not for aggregate values; queries
aggregating or grouping over aggregate values are thus rejected.
Extending the semantics to support such queries is non-trivial,
as order interacts in complex ways with uncertainty and provenance
\cite{DBLP:journals/tcs/AmarilliBDS19}.
The definitions also assume at least one grouping index ($m\geq 1$).
For \emph{scalar} aggregation (no \sql{GROUP BY}), SQL forms a single
group that exists even on empty input, so that the possible-world sum
should range over all worlds, the empty one included, and the
group-existence factor is~$\mathbb 1$; our implementation follows this
convention, which the formal treatment of this paper leaves aside.

\section{Main Results}
\label{sec:results}
A \emph{simple} \sql{HAVING} query counts the occurrences of each group
and keeps the groups whose count satisfies a comparison with a constant:
$\sigma_{\#2\op C}(\gamma^{\preceq}_{\#1}[1{:}+](q))$, \sql{COUNT(*)} being
the sum of the constant~$1$ over the group, for an arbitrary
order~$\preceq$. Such a query needs no aggregation at all:
$\text{\sql{COUNT(*)}}\geq C$ holds if and only if $C$ distinct occurrences of
the group exist, which a $C$-fold self-join of $q$ on the group key
expresses, followed by duplicate elimination, and $\leq$ and $=$ follow
by difference. The provenance of this rewriting is fixed by the semantics
of Section~\ref{sec:preliminaries}
\cite{green2007provenance,geerts2010database,sen2026provsql}, and it is
\ifarxiv
how we had to express such queries in earlier work~\cite{yunus2025using}.
\else
how \cite{yunus2025using} had to express such queries.
\fi We show that our
semantics agrees with it exactly when the m-semiring is absorptive and
$\otimes$ distributes over~$\ominus$. We state the result for a single
grouping column and an occurrence identifier as second column; several
grouping columns are handled identically.
\begin{toappendix}
\subsection{Proof of Theorem~\ref{th:correctness}}
The proof of Theorem~\ref{th:correctness} below rests on the following
lemmas; $A_W$, $T_U(W)$, $F_C(U)$ and $S_C(U)$ are as defined in that
proof.
\begin{lemma}\label{lem:upward-expansion}
\lean{Having}{Having.upward_expansion}
Let $(\mathbb K,\oplus,\otimes,\ominus,\mathbb 0,\mathbb 1)$ be a
commutative idempotent m-semiring,~$U$ a finite set, and
$(\alpha_w)_{w\in U}$ a family of elements of $\mathbb K$. For
$W\subseteq U$ define~$A_W \defeq\nobreak \bigotimes_{w\in W}\alpha_w$
and~$T_U(W) \defeq A_W \ominus \bigoplus_{w\in U\setminus W} A_{W\cup\{w\}}$.
Then $A_V \;\le\; \bigoplus_{\substack{W:\;V\subseteq W\subseteq U}} T_U(W)$
for every~$V\subseteq\nobreak U$.
\end{lemma}

\begin{proof}
From the properties of m-semirings $\ominus$ is the (right) residual of $\oplus$, i.e., for all $x,y,z\in \mathbb K$,
\begin{equation}\label{eq:residuation}
  x\ominus y \le z \quad\Longleftrightarrow\quad x \le y\oplus z .
\end{equation}

We prove
\begin{equation}\label{eq:upward-expansion}
  A_V \;\le\; \bigoplus_{\substack{W:\;V\subseteq W\subseteq U}} T_U(W).
\end{equation}
by induction on $d\defeq|U\setminus V|$.

\smallskip\noindent
\textbf{Base case ($d=0$).}
Then $V=U$. Since $U\setminus U=\varnothing$ and $\bigoplus_{\varnothing}(\cdot)=\mathbb 0$,
\[
  T_U(U)=A_U\ominus \mathbb 0.
\]
We claim $A_U\ominus\mathbb 0=A_U$. Indeed, by \eqref{eq:residuation}, for any $z\in \mathbb K$,
\[
  A_U\ominus\mathbb 0 \le z \iff A_U \le \mathbb 0\oplus z = z,
\]
so $A_U\ominus\mathbb 0$ is the greatest element below $A_U$, hence equals $A_U$.
Therefore,
\[
  \bigoplus_{\substack{W:\;U\subseteq W\subseteq U}} T_U(W)=T_U(U)=A_U=A_V,
\]
and \eqref{eq:upward-expansion} holds.

\smallskip\noindent
\textbf{Inductive step.}
Assume $d>0$ and that \eqref{eq:upward-expansion} holds for all pairs $(U,V')$
with $V\subsetneq V'\subseteq U$ and $|U\setminus V'|<d$.
Let
\[
  Y \defeq \bigoplus_{w\in U\setminus V} A_{V\cup\{w\}} \in \mathbb K.
\]
From Lemma~\ref{lem:monus-basic}(ii) we obtain the standard inequality
\begin{equation}\label{eq:M2-lemma1}
  X \le Y \oplus (X\ominus Y)\qquad\text{for all }X,Y\in \mathbb K,
\end{equation}
by instantiating \eqref{eq:residuation} with $z=X\ominus Y$ (since $X\ominus Y\le X\ominus Y$).
Applying \eqref{eq:M2-lemma1} to $X=A_V$ gives
\[
  A_V \le Y \oplus (A_V\ominus Y).
\]
By definition of $T_U(V)$ we have $A_V\ominus Y = T_U(V)$, hence
\begin{equation}\label{eq:AV-bound}
  A_V \le T_U(V)\ \oplus\ \bigoplus_{w\in U\setminus V} A_{V\cup\{w\}}.
\end{equation}

For each $w\in U\setminus V$ set $V_w\defeq V\cup\{w\}$. Then $|U\setminus V_w|=d-1$,
so the induction hypothesis yields
\[
  A_{V_w}\le \bigoplus_{\substack{W:\;V_w\subseteq W\subseteq U}} T_U(W).
\]
Taking the $\oplus$-sum over all $w\in U\setminus V$ and using monotonicity of $\oplus$,
\[
  \bigoplus_{w\in U\setminus V} A_{V\cup\{w\}}
  =
  \bigoplus_{w\in U\setminus V} A_{V_w}
  \le
  \bigoplus_{w\in U\setminus V}\ \bigoplus_{\substack{W:\;V_w\subseteq W\subseteq U}} T_U(W).
\]
Every $W$ appearing on the right satisfies $V\subseteq V_w\subseteq W\subseteq U$, hence $V\subseteq W\subseteq U$.
Therefore, using associativity/commutativity/idempotence of $\oplus$ to ignore duplicates and reorder summands,
\[
  \bigoplus_{w\in U\setminus V}\ \bigoplus_{\substack{W:\;V_w\subseteq W\subseteq U}} T_U(W)
  \le
  \bigoplus_{\substack{W:\;V\subseteq W\subseteq U}} T_U(W).
\]
Combining with \eqref{eq:AV-bound} yields
\[
  A_V \le T_U(V)\ \oplus\ \bigoplus_{\substack{W:\;V\subseteq W\subseteq U}} T_U(W)
  =
  \bigoplus_{\substack{W:\;V\subseteq W\subseteq U}} T_U(W),
\]
since the right-hand $\oplus$-sum already includes the summand $T_U(V)$ (take $W=V$).
This proves \eqref{eq:upward-expansion} for $V$ and completes the induction.
\end{proof}
With Lemma~\ref{lem:upward-expansion} in hand, we now establish matching
recurrences for $F_C(U)$ and $S_C(U)$.
\begin{lemma}[Recurrence for $F_C$]\label{lem:FC_recurrence}
\lean{Having}{Having.FC_recurrence}
Let $(\mathbb K,\oplus,\otimes,\ominus,\mathbb 0,\mathbb 1)$ be a
commutative idempotent m-semiring in which $\otimes$ distributes over
$\ominus$,~$U$ a non-empty finite set, and $(\alpha_w)_{w\in U}$ a family
of elements of~$\mathbb K$. For $W\subseteq U$ define~$A_W\defeq\bigotimes_{w\in W}\alpha_w$ and
$T_U(W)\defeq A_W\ominus \bigoplus_{w\in U\setminus W}A_{W\cup\{w\}}$, and
for $C\ge 0$ set
$F_C(U)\defeq\bigoplus_{\substack{W\subseteq U\\|W|\ge C}} T_U(W)$.
Then, for any $u\in U$ and $U'\defeq U\setminus\{u\}$,
$F_C(U)=F_C(U')\oplus\bigl(F_{C-1}(U')\otimes\alpha_u\bigr)$ for every $C\ge 1$.
\end{lemma}

\begin{proof}
  Fix $u\in U$ such that $U'=U\setminus\{u\}$ and a $C\in\mathbb N_{\ge 1}$.
Split  $F_C(U)$ according to whether $u\in W$:
\begin{equation}\label{eq:split-I-II}
F_C(U)
=
\underbrace{\bigoplus_{\substack{W\subseteq U'\\|W|\ge C}}T_U(W)}_{(I).\;u\notin W}
\ \oplus\
\underbrace{\bigoplus_{\substack{W'\subseteq U'\\|W'|\ge C-1}}T_U(W'\cup\{u\})}_{(II).\;u\in W}.
\end{equation}

\paragraph*{Simplify $(II)$: }
Fix $W'\subseteq U'$. Then $A_{W'\cup\{u\}}=A_{W'}\otimes\alpha_u$, and
$U\setminus(W'\cup\{u\})=U'\setminus W'$. Hence, using distributivity of $\otimes$ over $\oplus$,
\[
\bigoplus_{w\in U\setminus(W'\cup\{u\})}A_{W'\cup\{u,w\}}
=
\bigoplus_{w\in U'\setminus W'}\bigl(A_{W'\cup\{w\}}\otimes\alpha_u\bigr)
=
\Bigl(\bigoplus_{w\in U'\setminus W'}A_{W'\cup\{w\}}\Bigr)\otimes\alpha_u.
\]
Therefore, using distributivity of $\otimes$ over $\ominus$,
\begin{align*}
T_U(W'\cup\{u\})
&=
(A_{W'}\otimes\alpha_u)\ominus
\Bigl(\bigoplus_{w\in U'\setminus W'}A_{W'\cup\{w\}}\Bigr)\otimes\alpha_u\\
&=
\Bigl(A_{W'}\ominus\bigoplus_{w\in U'\setminus W'}A_{W'\cup\{w\}}\Bigr)\otimes\alpha_u
=
T_{U'}(W')\otimes\alpha_u.
\end{align*}
Taking $\oplus$ over all $W'\subseteq U'$ with $|W'|\ge C-1$ and factoring $\alpha_u$ using
distributivity of $\otimes$ over $\oplus$ yields
\begin{equation}\label{eq:II-simplifies}
(II)=\Bigl(\bigoplus_{\substack{W'\subseteq U'\\|W'|\ge C-1}}T_{U'}(W')\Bigr)\otimes\alpha_u
=F_{C-1}(U')\otimes\alpha_u.
\end{equation}

\paragraph*{Relate $T_U(W)$ and $T_{U'}(W)$ for $W\subseteq U'$: }
Fix $W\subseteq U'$. Since $U\setminus W=(U'\setminus W)\cup\{u\}$, we have
\[
\bigoplus_{w\in U\setminus W}A_{W\cup\{w\}}
=
\Bigl(\bigoplus_{w\in U'\setminus W}A_{W\cup\{w\}}\Bigr)\oplus (A_W\otimes\alpha_u).
\]
Using
$x\ominus(y\oplus z)=(x\ominus y)\ominus z$, we get
\begin{equation}\label{eq:TU-rewrite}
T_U(W)
=
A_W\ominus\Bigl(\bigoplus_{w\in U'\setminus W}A_{W\cup\{w\}}\oplus(A_W\otimes\alpha_u)\Bigr)
=
T_{U'}(W)\ominus(A_W\otimes\alpha_u).
\end{equation}

\paragraph*{Show $F_C(U)\le F_C(U')\oplus(F_{C-1}(U')\otimes\alpha_u)$: }
  By \eqref{eq:TU-rewrite} and the inequality $x\ominus y\le x$ (shown in Lemma~\ref{lem:monus-basic}(i)),
we have $T_U(W)\le T_{U'}(W)$ for all $W\subseteq U'$.
Taking $\oplus$ over all $W\subseteq U'$ with $|W|\ge C$ yields $(I)\le F_C(U')$.
Together with \eqref{eq:split-I-II} and \eqref{eq:II-simplifies}, this gives
\begin{equation}\label{eq:leq-direction}
F_C(U)\le F_C(U')\oplus(F_{C-1}(U')\otimes\alpha_u).
\end{equation}

\paragraph*{Show $F_C(U')\le F_C(U)$: }
Fix $W\subseteq U'$ with $|W|\ge C$.
By the inequality $x\le y\oplus(x\ominus y)$ (shown in Lemma~\ref{lem:monus-basic}(ii)) applied to
$x=T_{U'}(W)$ and $y=A_W\otimes\alpha_u$, and using \eqref{eq:TU-rewrite}, we get
\begin{equation}\label{eq:TUprime-bound}
T_{U'}(W)\le (A_W\otimes\alpha_u)\oplus T_U(W).
\end{equation}
Since $|W|\ge C$, the term $T_U(W)$ is one of the summands of $F_C(U)$, hence $T_U(W)\le F_C(U)$.
Moreover $A_W\otimes\alpha_u=A_{W\cup\{u\}}$, and $|W\cup\{u\}|\ge C$.
By Lemma~\ref{lem:upward-expansion} applied to $V=W\cup\{u\}$ we have
\[
A_{W\cup\{u\}}
\le
\bigoplus_{\substack{Y:\;W\cup\{u\}\subseteq Y\subseteq U}} T_U(Y)
\le
\bigoplus_{\substack{Y\subseteq U\\|Y|\ge C}}T_U(Y)
=
F_C(U),
\]
so $A_W\otimes\alpha_u\le F_C(U)$.
Plugging these two bounds into \eqref{eq:TUprime-bound} yields $T_{U'}(W)\le F_C(U)$.
Taking $\oplus$ over all $W\subseteq U'$ with $|W|\ge C$ gives $F_C(U')\le F_C(U)$.

\paragraph*{The recurrence: }
From \eqref{eq:leq-direction}, $F_C(U')\le F_C(U)$ and
$F_{C-1}(U')\otimes\alpha_u=(II)\le F_C(U)$, we obtain
$F_C(U')\oplus(F_{C-1}(U')\otimes\alpha_u)\le F_C(U)$.
Together with \eqref{eq:leq-direction}, we get
\[
F_C(U)=F_C(U')\oplus(F_{C-1}(U')\otimes\alpha_u),
\]
as required.
\end{proof}
\begin{lemma}[Recurrence for $S_C$]\label{lem:SC_recurrence}
\lean{Having}{Having.SC_recurrence}
Let $(\mathbb K,\oplus,\otimes,\ominus,\mathbb 0,\mathbb 1)$ be a
commutative m-semiring and $(\alpha_w)_{w\in U}$ a family of elements of
$\mathbb K$ indexed by a non-empty finite set~$U$. For
$W\subseteq U$ define~$A_W\defeq\nobreak\bigotimes_{w\in W}\alpha_w$, and for
$C\ge 0$ set $S_C(U)\defeq\bigoplus_{\substack{W\subseteq U\\|W|=C}} A_W$.
Then, for any $u\in U$ and $U'\defeq U\setminus\{u\}$,
$S_C(U)=S_C(U')\oplus\bigl(S_{C-1}(U')\otimes\alpha_u\bigr)$ for every $C\ge 1$.
\end{lemma}

\begin{proof}
Fix $C\ge 1$ and $u\in U$, and write $U=U'\cup\{u\}$ with $U'=U\setminus\{u\}$.

Define two index families:
\[
\mathcal A\defeq\{\,W\subseteq U \mid |W|=C,\ u\notin W\,\},
\qquad
\mathcal B\defeq\{\,W\subseteq U \mid |W|=C,\ u\in W\,\}.
\]
Then $\mathcal A$ and $\mathcal B$ are disjoint and their union is the full family
\[
\mathcal A\ \dot\cup\ \mathcal B=\{\,W\subseteq U\mid |W|=C\,\},
\]
because every $W\subseteq U$ either contains $u$ or does not.

Therefore, by associativity and commutativity of $\oplus$ (finite sum over a disjoint
union equals the $\oplus$ of the two partial sums),
\begin{equation}\label{eq:split-SC}
S_C(U)
=
\bigoplus_{W\in\mathcal A} A_W\ \oplus\ \bigoplus_{W\in\mathcal B} A_W.
\end{equation}

\smallskip\noindent
\emph{(i) The $\mathcal A$-part.}
If $W\in\mathcal A$, then $u\notin W$ and $W\subseteq U'=U\setminus\{u\}$.
Conversely, if $W\subseteq U'$ and $|W|=C$, then $W\subseteq U$ and $u\notin W$.
Hence
\[
\mathcal A=\{\,W\subseteq U'\mid |W|=C\,\},
\]
so
\begin{equation}\label{eq:A-part}
\bigoplus_{W\in\mathcal A}A_W
=
\bigoplus_{\substack{W\subseteq U'\\|W|=C}}A_W
=
S_C(U').
\end{equation}

\smallskip\noindent
\emph{(ii) The $\mathcal B$-part.}
Define a map
\[
\phi:\{\,W'\subseteq U'\mid |W'|=C-1\,\}\longrightarrow \mathcal B,
\qquad
\phi(W')\defeq W'\cup\{u\}.
\]
This map is a bijection:
\begin{itemize}
\item It is well-defined because if $W'\subseteq U'$ and $|W'|=C-1$, then
$u\notin W'$ and thus $|W'\cup\{u\}|=|W'|+1=C$, and $\phi(W')$ contains $u$.
\item It is injective because if $W'_1\cup\{u\}=W'_2\cup\{u\}$, removing $u$ gives
$W'_1=W'_2$.
\item It is surjective because if $W\in\mathcal B$, then $u\in W$ and setting
$W'\defeq W\setminus\{u\}$ gives $W'\subseteq U'$ and $|W'|=C-1$, with $W=\phi(W')$.
\end{itemize}
Hence we may re-index the sum over $\mathcal B$ via $\phi$:
\begin{equation}\label{eq:B-reindex}
\bigoplus_{W\in\mathcal B}A_W
=
\bigoplus_{\substack{W'\subseteq U'\\|W'|=C-1}} A_{\phi(W')}
=
\bigoplus_{\substack{W'\subseteq U'\\|W'|=C-1}} A_{W'\cup\{u\}}.
\end{equation}
For each $W'\subseteq U'$ we have, by commutativity and associativity of $\otimes$,
\[
A_{W'\cup\{u\}}
=
\Bigl(\bigotimes_{w\in W'}\alpha_w \Bigr)\otimes \alpha_u
=
A_{W'}\otimes \alpha_u.
\]
Substituting into \eqref{eq:B-reindex} and using distributivity of $\otimes$ over $\oplus$,
\begin{equation}\label{eq:B-part}
\bigoplus_{W\in\mathcal B}A_W
=
\bigoplus_{\substack{W'\subseteq U'\\|W'|=C-1}} (A_{W'}\otimes\alpha_u)
=
\Bigl(\bigoplus_{\substack{W'\subseteq U'\\|W'|=C-1}} A_{W'}\Bigr)\otimes\alpha_u
=
S_{C-1}(U')\otimes\alpha_u.
\end{equation}

Finally, combining \eqref{eq:split-SC}, \eqref{eq:A-part}, and \eqref{eq:B-part}, we obtain
\[
S_C(U)=S_C(U')\ \oplus\ \bigl(S_{C-1}(U')\otimes\alpha_u\bigr),
\]
as required.
\end{proof}
The last ingredient, used for the $=$ and $\le$ cases, is a per-world
upper bound: the contribution $T_U(W)$ of a single world $W$ of size at
most~$C$ is below the difference $S_j(U)\ominus S_{C+1}(U)$, for every
level $j\le|W|$.
\begin{lemma}[Per-world upper bound]\label{lem:world-bound}
\lean{Having}{Having.world_bound}
Let $(\mathbb K,\oplus,\otimes,\ominus,\mathbb 0,\mathbb 1)$ be a
commutative absorptive m-semiring in which $\otimes$ distributes over
$\ominus$, $U$ a finite set, and $(\alpha_w)_{w\in U}$ a family of
elements of~$\mathbb K$, with $A_W$, $T_U(W)$, and $S_C(U)$ defined as in
Lemmas~\ref{lem:FC_recurrence} and~\ref{lem:SC_recurrence}. Then for
every $W\subseteq U$ and all integers $j,C$ with $j\le|W|\le C$,
\[
T_U(W)\ \le\ S_j(U)\ominus S_{C+1}(U).
\]
\end{lemma}

\begin{proof}
First note that, since $\mathbb K$ is absorptive, every $y\in\mathbb K$
satisfies $y\le\mathbb 1$, so by monotonicity of $\otimes$ (a consequence
of its distributivity over $\oplus$) we have $x\otimes y\le x$ for all
$x,y$: products only decrease as factors are added.

Write $E_W\defeq\bigoplus_{w\in U\setminus W}\alpha_w$. By distributivity
of $\otimes$ over $\oplus$,
$\bigoplus_{w\in U\setminus W}A_{W\cup\{w\}}=A_W\otimes E_W$, so by
distributivity of $\otimes$ over $\ominus$,
\begin{equation}\label{eq:TW-factored}
T_U(W)=A_W\ominus(A_W\otimes E_W)=A_W\otimes(\mathbb 1\ominus E_W).
\end{equation}

We now make two observations.

\emph{(1) $A_W\le S_j(U)$.} Choose any $W'\subseteq W$ with $|W'|=j$.
Then $A_W=A_{W'}\otimes A_{W\setminus W'}\le A_{W'}\le S_j(U)$, the last
step because $A_{W'}$ is one of the summands of $S_j(U)$.

\emph{(2) $S_{C+1}(U)\le S_j(U)\otimes E_W$.} Let $V\subseteq U$ with
$|V|=C+1$ (if no such $V$ exists, $S_{C+1}(U)=\mathbb 0$ and the claim is
trivial). Since $|V|>|W|$, there is some $v\in V\setminus W$, and since
$|V\setminus\{v\}|=C\ge j$, there is some $V'\subseteq V\setminus\{v\}$
with $|V'|=j$. Then
\[
A_V=A_{V'}\otimes\alpha_v\otimes A_{V\setminus(V'\cup\{v\})}
\le A_{V'}\otimes\alpha_v\le S_j(U)\otimes E_W,
\]
using $A_{V'}\le S_j(U)$, $\alpha_v\le E_W$ (as $v\in U\setminus W$), and
monotonicity of $\otimes$. Taking the $\oplus$-sum over all such $V$
(idempotence) proves the claim.

Combining: by \eqref{eq:TW-factored}, observation (1) and monotonicity of
$\otimes$, then distributivity of $\otimes$ over $\ominus$, then
observation (2) and antitonicity of $\ominus$ in its second argument
(Lemma~\ref{lem:monus-basic}(iii)),
\[
T_U(W)=A_W\otimes(\mathbb 1\ominus E_W)
\le S_j(U)\otimes(\mathbb 1\ominus E_W)
= S_j(U)\ominus\bigl(S_j(U)\otimes E_W\bigr)
\le S_j(U)\ominus S_{C+1}(U).\qedhere
\]
\end{proof}
Note the order of the steps in the proof of
Lemma~\ref{lem:world-bound}: writing
$E_W\defeq\bigoplus_{w\in U\setminus W}\alpha_w$, bounding $A_W$ by
$S_j(U)$ \emph{directly inside the subtrahend} of
$T_U(W)=A_W\ominus(A_W\otimes E_W)$ would move the monus in the wrong
direction; the bound is instead applied to the factored form
$T_U(W)=A_W\otimes(\mathbb 1\ominus E_W)$, whose second factor does not
depend on $A_W$, and distributivity of $\otimes$ over $\ominus$ is only
used afterwards, when the subtrahend has grown to $S_j(U)\otimes E_W$,
which dominates $S_{C+1}(U)$.

\end{toappendix}

\begin{theoremrep}\label{th:correctness}
  Let $q$ be an aggregate-free query of arity~$2$ whose second column is
  an \emph{occurrence identifier}: the tuples of $\angsem{\hat I}{q}$ are
  pairwise distinct (e.g., a key attribute of a base relation carried
  along). For $C\in\mathbb{N}_{\ge1}$ and ${\op}\in \{\leq,
  =,\geq\}$ consider the queries:
\begin{align*}
    Q^{\op C}_1&=\Pi_{\#1}(\sigma_{\#2\op C}(\gamma^{\preceq}_{\#1}[1:+](q)))
    & Q^{\leq C}_2&= Q^{\geq 1}_2 - Q^{\geq (C+1)}_2\\
    Q^{\geq C}_2&=\varepsilon\bigl(\Pi_{\#1}(q\bowtie_{\phi_1}q\bowtie_{\phi_2}\cdots\bowtie_{\phi_{C-1}}q)\bigr)
    & Q^{=C}_2&=Q^{\geq C}_2 - Q^{\geq (C+1)}_2
\end{align*}
where $\varepsilon$ denotes duplicate elimination, $\phi_i$ is the join
condition $\#1=\#(2i+1)\land\#(2i)<\#(2i+2)$, and the chain is empty for
$C=1$, so that $Q^{\geq 1}_2=\varepsilon(\Pi_{\#1}(q))$; the strict
comparisons on the identifier column select each $C$-element set
of occurrences of a group once.
    Then:
    \begin{enumerate}[(i)]
      \item
        \lean{HavingQueryCorrectness}{Query.joinCount_correct}
        \lean{HavingJoinCompositional}{joinCountQueryPadded_correct}
        For any commutative m-semiring $\mathbb{K}$ which is
        absorptive and such that $\otimes$ is
        distributive over $\ominus$, for any
        $\mathbb{K}$-instance $\hat I$,
        $\angsem{\hat I}{Q^{\op C}_1}=\angsem{\hat I}{Q^{\op C}_2}$.
      \item
        \lean{HavingQueryCounterexamples}{HavingQueryCounterexamples.ChainFive.query_counterexample}
        There exist a commutative m-semiring $\mathbb{K}_1$ that is
        absorptive but where $\otimes$ is not distributive over
        $\ominus$, and a $\mathbb{K}_1$-instance $\hat I_1$ such that
        $\angsem{\hat I_1}{Q^{=1}_1}\neq\angsem{\hat I_1}{Q^{=1}_2}$.
      \item
        \lean{HavingQueryCounterexamples}{HavingQueryCounterexamples.MinTropicalZ.query_counterexample}
        There exist a commutative idempotent m-semiring $\mathbb{K}_2$
        with $\otimes$ distributive over $\ominus$ but not absorptive,
        and a $\mathbb{K}_2$-instance $\hat I_2$ such that
        $\angsem{\hat I_2}{Q^{\geq 1}_1}\neq\angsem{\hat I_2}{Q^{\geq 1}_2}$.
    \end{enumerate}
\end{theoremrep}
\begin{proof}
We proceed with our proof for part $(i)$ of Theorem~\ref{th:correctness}.
\paragraph*{Proof of (i)}
For any $\mathbb{K}$-instance $\hat I$, we prove the equality $\angsem{\hat I}{Q^{\op C}_1}=\angsem{\hat I}{Q^{\op C}_2}$ by showing
  that, for any tuple in the query result, the corresponding provenance expressions, instantiated on any commutative m-semiring $\mathbb{K}$ under the assumptions of $(i)$, are the same.
    In other words, we need to show that the provenance annotation of each group produced by
    the grouping operator under our \sql{HAVING} semantics of
    Definition~\ref{def:semantics} coincides with
    that produced by the self-join query. Recall from
    Section~\ref{sec:preliminaries} that a tuple annotated by $\mathbb 0$
    is absent from a $\mathbb K$-relation: a group whose predicate
    provenance is $\mathbb 0$ is thus absent from $\angsem{\hat
    I}{Q^{\op C}_1}$, as it is from $\angsem{\hat I}{Q^{\op C}_2}$, which
    produces no row for it. The Lean development states the identity per
    group key \lean{HavingQueryCorrectness}{Query.joinCount_correct} and,
    row for row, for $Q^{\op C}_2$ padded with a $\mathbb 0$-annotated
    row per group key
    \lean{HavingJoinCompositional}{joinCountQueryPadded_correct}.
    Since \(q\) is an arity~$2$ query, every tuple in $\angsem{\hat
    I}{q}$ has the form $(u_1,u_2,\alpha)$ where $\alpha\in \mathbb K$.
Fix an arbitrary group key $\vec v=(a)$.

Let $U\defeq U^{\preceq}_{\angsem{\hat I}{q},(a)}$ be the sequence of
occurrences of the group with key~$a$ in $\angsem{\hat I}{q}$ under
order~$\preceq$ (Definition~\ref{def:possible_world}).

\paragraph*{Provenance of the \sql{HAVING} query using possible-world semantics:}
For the query
$
Q^{\op C}_1
=
\Pi_{\#1}\bigl(
\sigma_{\#2 \op C}(
  \gamma^{\preceq}_{\#1}[1:+](q)
)\bigr)
$
\begin{multline}\label{eq:Q1-prov}
\angsem{\hat I}{Q^{\op C}_1}(a)
=\predsem{\hat I}{\gamma^{\preceq}_{\#1}[1:+](q)}{\#2\op C}(u)
=\bigoplus_{\emptyset\neq W\sqsubseteq U}
\Bigl(
\ann_U(W)\otimes
\chi_{\op}\bigl(\mathrm{agg}_{1,f_\Sigma}(W),\ C\bigr)
\Bigr)\\
=\bigoplus_{\emptyset\neq W\sqsubseteq U}
\Bigl(
\ann_U(W)\otimes
\chi_{\op}\bigl(|W|,\ C\bigr)
\Bigr)
\end{multline}
where $f_\Sigma$ denotes the commutative \sql{COUNT} aggregate, so that
$\mathrm{agg}_{1,f_\Sigma}(W)=|W|$.
\\
For each possible world $W\sqsubseteq U$, write
$A_W\defeq\bigotimes_{w\in W}\alpha_w$, where $\alpha_w$ is the annotation
of occurrence~$w$. Then
\begin{align*}
T_U(W)&\defeq\ A_W \otimes \Bigl(\mathbb{1}_{\mathbb{K}} \ominus \bigoplus_{w\in U\setminus W}\alpha_w\Bigr)\\
&= A_W\ominus \bigoplus_{w\in U\setminus W}A_{W\cup\{w\}} \;\textit{[by distributivity of $\otimes$ over $\ominus$ and then of $\otimes$ over $\oplus$]}
\end{align*}
so that $\ann_U(W)=T_U(W)$.

\lean{Having}{Having.T_eq_mul_one_monus_sum}
Since $\chi_{\op}(|W|,C)=\mathbb{1}_\mathbb{K}$ iff $|W|\op C$ is true and $\mathbb{0}_\mathbb{K}$ otherwise, \eqref{eq:Q1-prov} becomes
\begin{equation}\label{eq:pw_tw}
\angsem{\hat I}{Q^{\op C}_1}(a)
=
\bigoplus_{\substack{W\sqsubseteq U,\\ W\neq\emptyset,\\|W|\op C}} T_{U}(W).
\end{equation}
\paragraph*{Provenance of the \sql{JOIN} queries:}
For every $W\sqsubseteq U $ and $C\ge 1$ we define,
\[
S_C(U)\defeq\bigoplus_{\substack{W\sqsubseteq U \\|W|=C}} A_W
\quad
(\text{with }S_C(U)=\mathbb{0}_\mathbb{K}\text{ if }C>|U|).
\]
Consider the query $Q^{\ge C}_2$.  A join-valuation contributing to the output
tuple in $\angsem{\hat I}{Q^{\ge C}_2}$  with key $(a)$ corresponds to choosing $C$ tuples from $\angsem{\hat I}{q}$ with first
attribute $a$ and strictly increasing identifiers (due to the chain of
join conditions $\#2<\#4<\cdots$); identifiers being pairwise distinct,
each $C$-element set of occurrences is selected exactly once, in the order
of its identifiers. Thus this corresponds to choosing a
$C$-element subset $W\sqsubseteq U$, which contributes the product $A_W$.
Under multiset semantics, $\Pi_{\#1}$ produces one occurrence of the output tuple $(a)$ for each such subset; the duplicate elimination $\varepsilon$ then merges these occurrences into a single tuple whose annotation is their $\oplus$-sum.
Therefore,
\[
\angsem{\hat I}{Q^{\ge C}_2}(a)
=S_C(U).
\tag{$J_{\ge}$}
\]
By definition,
\[
Q^{=C}_2 = Q^{\ge C}_2 - Q^{\ge (C+1)}_2,
\qquad
Q^{\le C}_2 = Q^{\ge 1}_2 - Q^{\ge (C+1)}_2,
\]
Hence, using $(J_{\ge})$,
\[
\angsem{\hat I}{Q^{=C}_2}(a)
=
S_C(U) \ominus S_{C+1}(U),
\quad
\angsem{\hat I}{Q^{\le C}_2}(a)
=
S_1(U) \ominus S_{C+1}(U).
\tag{$J_{=,\le}$}
\]
\paragraph*{\textbf{When $\op$ is $\ge$:}}
From \eqref{eq:pw_tw} with $\op$ being $\ge$, we set
$
F_C(U)\defeq \bigoplus_{\substack{W\sqsubseteq U\\|W|\ge C}} T_{U}(W).
$
So,
\begin{equation}\label{eq:PW_ge_equals_F}
\angsem{\hat I}{Q^{\ge C}_1}(a)=F_C(U).
\tag{$PW_{\ge}$}
\end{equation}
We already showed for the join query that
\[
\angsem{\hat I}{Q^{\ge C}_2}(a)=S_C(U)
\tag{$J_{\ge}$}
\]
The quantities $A_W$, $T_U(W)$, $F_C(U)$ and $S_C(U)$ depend only on which
occurrences are selected, not on their order; we may therefore regard $U$ as
the finite \emph{set} of its occurrences and each possible world
$W\sqsubseteq U$ as a subset $W\subseteq U$, as in
Lemmas~\ref{lem:FC_recurrence} and~\ref{lem:SC_recurrence}.
By Lemma~\ref{lem:FC_recurrence}, for an arbitrary choice of $u\in U$ with $U'=U\setminus\{u\}$, for all $C\ge 1$,
\[
F_C(U)=F_C(U')\ \oplus\ \bigl(F_{C-1}(U')\otimes\alpha_u\bigr).
\tag{R$_F$}
\]
Similarly, by Lemma~\ref{lem:SC_recurrence}, $S_C(U)$ can be shown to satisfy the same recurrence:
\[
S_C(U)=S_C(U')\ \oplus\ \bigl(S_{C-1}(U')\otimes\alpha_u\bigr)
\qquad(C\ge 1).
\tag{R$_S$}
\]
We prove that $F_C(U)=S_C(U)$ by induction on $|U|$, for all finite $U$ and all $C\ge 1$ \lean{Having}{Having.F_eq_S}.
If $|U|=0$ then $U=\varnothing$ and $F_C(\varnothing)=\mathbb 0=S_C(\varnothing)$ for $C\ge 1$
(empty $\oplus$-sum).
Assume $|U|>0$ and pick $u\in U$, let $U'=U\setminus\{u\}$.
Applying (R$_F$) and (R$_S$),
\[
F_C(U)=F_C(U')\oplus(F_{C-1}(U')\otimes\alpha_u),
\qquad
S_C(U)=S_C(U')\oplus(S_{C-1}(U')\otimes\alpha_u).
\]
The induction hypothesis gives $F_C(U')=S_C(U')$, and for $C\ge 2$ also
$F_{C-1}(U')=S_{C-1}(U')$. For $C=1$ the lower term is the count-$0$
quantity, which is not covered by the induction hypothesis but holds
directly: $S_0(U')=A_\varnothing=\mathbb 1$, while
$F_0(U')=\bigoplus_{W\subseteq U'}T_{U'}(W)=\mathbb 1$ because $\mathbb K$ is
absorptive (so that $\mathbb 1\oplus x=\mathbb 1$)
\lean{Having}{Having.F_zero_eq_one}; this is the only step requiring
absorptivity rather than just the idempotence it implies. In every case the two right-hand sides coincide,
so $F_C(U)=S_C(U)$.
Thus $F_C(U)=S_C(U)$ for all $C\ge 1$. Hence by \eqref{eq:PW_ge_equals_F} and $(J_{\ge})$,
\[
\angsem{\hat I}{Q^{\ge C}_1}(a)=\angsem{\hat I}{Q^{\ge C}_2}(a).
\]
\paragraph*{\textbf{When $\op$ is $=$:}}
Since $\chi_{=}(n,C)$ is $\mathbb 1$ when $n=C$ and $\mathbb 0$ otherwise, the
\sql{HAVING} side is simply the restriction of the sum to the worlds of size
exactly~$C$:
\[
\angsem{\hat I}{Q^{=C}_1}(a)
=\bigoplus_{W\sqsubseteq U}\bigl(T_{U}(W)\otimes\chi_{=}(|W|,C)\bigr)
=\bigoplus_{\substack{W\sqsubseteq U,\\ |W|=C}} T_{U}(W)
\defeq G_C(U).
\]
On the join side, $(J_{=,\le})$ together with the already established
$\ge$-case gives
\[
\angsem{\hat I}{Q^{=C}_2}(a)
=\angsem{\hat I}{Q^{\ge C}_2}(a)\ominus\angsem{\hat I}{Q^{\ge(C+1)}_2}(a)
=S_C(U)\ominus S_{C+1}(U).
\]
The $=$ case therefore amounts to the identity
\begin{equation}\label{eq:exactly-C}
G_C(U)=S_C(U)\ominus S_{C+1}(U),
\end{equation}
which does not follow from the $\ge$-case alone: it is proven by combining
a per-world upper bound (Lemma~\ref{lem:world-bound}) with
right-distributivity of $\ominus$ over~$\oplus$
\lean{Having}{Having.G_eq_S_monus_S}.

On the one hand, Lemma~\ref{lem:world-bound}, applied with $j=C$, gives
$T_U(W)\le S_C(U)\ominus S_{C+1}(U)$ for every world $W\sqsubseteq U$ with
$|W|=C$; taking the $\oplus$-sum over all such $W$ yields
$G_C(U)\le S_C(U)\ominus S_{C+1}(U)$.

On the other hand, for every $W\sqsubseteq U$ with $|W|=C$, each summand of
$P_W\defeq\bigoplus_{w\in U\setminus W}A_{W\cup\{w\}}$ is an $A_V$ with
$|V|=C+1$, so $P_W\le S_{C+1}(U)$ and, by antitonicity of $\ominus$ in its
second argument (Lemma~\ref{lem:monus-basic}(iii)),
$A_W\ominus S_{C+1}(U)\le A_W\ominus P_W=T_U(W)$.
Corollary~\ref{lem:finite-add-monus}, applied to
$S_C(U)=\bigoplus_{|W|=C}A_W$ with the \emph{fixed} subtrahend $S_{C+1}(U)$
(legitimate because absorptivity implies idempotence), then yields
\[
S_C(U)\ominus S_{C+1}(U)
=\bigoplus_{\substack{W\sqsubseteq U,\\|W|=C}}\bigl(A_W\ominus S_{C+1}(U)\bigr)
\le\bigoplus_{\substack{W\sqsubseteq U,\\|W|=C}}T_U(W)=G_C(U).
\]
This proves \eqref{eq:exactly-C}, whence
$\angsem{\hat I}{Q^{=C}_1}(a)=\angsem{\hat I}{Q^{=C}_2}(a)$.

\paragraph*{\textbf{When $\op$ is $\le$:}}
Since groups produced by $\gamma^{\preceq}_{\#1}$ are non-empty, only
$n=|W|\ge 1$ occurs, and for such $n$ the condition $n\le C$ is equivalent to
$n\ge 1$ and not $n\ge C+1$. As in the $=$ case, the \sql{HAVING} side is the
restriction of the sum to the worlds of size between $1$ and~$C$,
\[
\angsem{\hat I}{Q^{\le C}_1}(a)
=\bigoplus_{\substack{W\sqsubseteq U,\\ 1\le|W|\le C}} T_{U}(W),
\]
while $(J_{=,\le})$ and the $\ge$-case give
$\angsem{\hat I}{Q^{\le C}_2}(a)=S_1(U)\ominus S_{C+1}(U)$. The $\le$ case
therefore amounts to the identity
\begin{equation}\label{eq:at-most-C}
\bigoplus_{1\le|W|\le C} T_{U}(W)
=S_1(U)\ominus S_{C+1}(U),
\end{equation}
which is proven like \eqref{eq:exactly-C}
\lean{Having}{Having.atMost_eq_S_monus_S}.
Lemma~\ref{lem:world-bound}, now applied with $j=1$, gives
$T_U(W)\le S_1(U)\ominus S_{C+1}(U)$ for every $W$ with $1\le|W|\le C$,
whence
$\bigoplus_{1\le|W|\le C}T_U(W)\le S_1(U)\ominus S_{C+1}(U)$.
Conversely, since $F_1(U)=S_1(U)$ by the $\ge$-case,
Corollary~\ref{lem:finite-add-monus} yields
\[
S_1(U)\ominus S_{C+1}(U)
=\bigoplus_{\substack{W\sqsubseteq U,\\|W|\ge 1}}\bigl(T_U(W)\ominus S_{C+1}(U)\bigr).
\]
For $|W|\ge C+1$, we have $T_U(W)\le A_W\le A_{W'}\le S_{C+1}(U)$ where
$W'$ is any subset of $W$ with $|W'|=C+1$ (using
Lemma~\ref{lem:monus-basic}(i), then the fact that in an absorptive m-semiring
products only decrease as factors are added), so
$T_U(W)\ominus S_{C+1}(U)=\mathbb 0$ by the residuation law; for
$1\le|W|\le C$, $T_U(W)\ominus S_{C+1}(U)\le T_U(W)$ by
Lemma~\ref{lem:monus-basic}(i). Hence
$S_1(U)\ominus S_{C+1}(U)\le\bigoplus_{1\le|W|\le C}T_U(W)$, proving
\eqref{eq:at-most-C} and thus
$\angsem{\hat I}{Q^{\le C}_1}(a)=\angsem{\hat I}{Q^{\le C}_2}(a)$.

\smallskip
Since group key $(a)$ was arbitrary, we conclude
$\angsem{\hat I}{Q^{\op C}_1}=\angsem{\hat I}{Q^{\op C}_2}$ for
${\op}\in\{\ge,=,\le\}$.
\paragraph*{Proof of (ii)}
\lean{HavingQueryCounterexamples}{HavingQueryCounterexamples.ChainFive.query_counterexample}
To isolate distributivity we give an \emph{absorptive} counterexample.
Let $\mathbb K_1$ be the five-element chain
$\mathbb 0=e_0<e_1<e_2<e_3<e_4=\mathbb 1$ with
\begin{itemize}
  \item $e_i\oplus e_j=e_{\max(i,j)}$;
  \item $\otimes$ commutative, of unit~$e_4$ and annihilator~$e_0$, defined on the remaining elements by $e_3\otimes e_3=e_3$, $e_3\otimes e_2=e_3\otimes e_1=e_1$, and $e_i\otimes e_j=\mathbb 0$ for all other $1\le i,j\le 3$.
\end{itemize}
This is a commutative m-semiring
\lean{Semirings/ChainFive}{ChainFive.instCommSemiringWithMonus}: the
natural order is the order of the chain, and $\ominus$ is determined by
it, $e_i\ominus e_j=\mathbb 0$ if $i\le j$ and $e_i\ominus e_j=e_i$
otherwise.
Since $\oplus$ is $\max$ and $\mathbb 1=e_4$ is the greatest element,
$\mathbb K_1$ is absorptive \lean{Semirings/ChainFive}{ChainFive.absorptive},
hence idempotent. It is not $\otimes$/$\ominus$-distributive
\lean{Semirings/ChainFive}{ChainFive.not_mul_sub_left_distributive}: since
$\mathbb 1\ominus e_3=\mathbb 1$, we get
$e_3\otimes(\mathbb 1\ominus e_3)=e_3$, whereas
$(e_3\otimes\mathbb 1)\ominus(e_3\otimes e_3)=e_3\ominus e_3=\mathbb 0$.
Let $\hat I_1$ have a group with
key~$u$ consisting of three occurrences annotated $\alpha_1=e_2$ and
$\alpha_2=\alpha_3=e_3$. Only the singleton possible worlds
satisfy the \sql{HAVING} condition $\text{\sql{COUNT(*)}}=1$, so
\[
\angsem{\hat I_1}{Q^{=1}_1}(u)
=\bigoplus_{i}\alpha_i\otimes\Bigl(\mathbb 1\ominus\bigoplus_{j\neq i}\alpha_j\Bigr),
\qquad
\angsem{\hat I_1}{Q^{=1}_2}(u)
=\Bigl(\bigoplus_i\alpha_i\Bigr)\ominus\Bigl(\bigoplus_{i<j}\alpha_i\otimes\alpha_j\Bigr),
\]
the latter because $Q^{=1}_2=Q^{\ge1}_2-Q^{\ge2}_2$ annotates~$u$ by
$\bigoplus_i\alpha_i$ minus $\bigoplus_{i<j}\alpha_i\otimes\alpha_j$.
Here $\bigoplus_{j\neq i}\alpha_j=e_3$ for every~$i$ and
$\mathbb 1\ominus e_3=\mathbb 1$, so
\[
\angsem{\hat I_1}{Q^{=1}_1}(u)
=e_2\oplus e_3\oplus e_3=e_3,
\]
while $\bigoplus_i\alpha_i=e_3$ and
$\bigoplus_{i<j}\alpha_i\otimes\alpha_j=e_1\oplus e_1\oplus e_3=e_3$, so
\[
\angsem{\hat I_1}{Q^{=1}_2}(u)=e_3\ominus e_3=\mathbb 0\neq e_3.
\]
Hence $\angsem{\hat I_1}{Q^{=1}_1}\neq\angsem{\hat I_1}{Q^{=1}_2}$.
\paragraph*{Proof of (iii)}
To isolate absorptivity we give a \emph{distributive} counterexample.
\lean{HavingQueryCounterexamples}{HavingQueryCounterexamples.MinTropicalZ.query_counterexample}
Consider the tropical semiring
$\mathbb K_2=(\mathbb{Z}\cup\{\infty\},\min,+,\infty,0)$ over the
integers, equipped with the same monus as the tropical semiring over
$\mathbb R$ of Table~\ref{tab:semirings} (the construction of
Appendix~\ref{appendix:corrections} applies to any linearly ordered
commutative additive monoid); the computation below is identical over
$\mathbb R$ \lean{Semirings/Tropical}{MinTropicalR.F_ne_S}.
It is idempotent, $\otimes$ distributes over $\ominus$
\lean{HavingQueryCounterexamples}{HavingQueryCounterexamples.MinTropicalZ.mul_sub_left_distributive}, and it is
not absorptive
\lean{HavingQueryCounterexamples}{HavingQueryCounterexamples.MinTropicalZ.not_absorptive_witness}:
$\mathbb 1\oplus(-1)=\min(0,-1)=-1\neq\mathbb 1$.
We show that there exists a $\mathbb K_2$-instance $\hat I$ such that
\[
\angsem{\hat I}{Q^{\geq 1}_1}
\neq
\angsem{\hat I}{Q^{\geq 1}_2}.
\]

Consider a group with key $u$ containing exactly two tuples,
both annotated by
$
t_1=t_2=-1$.
So,
\[
\angsem{\hat I}{Q^{\ge1}_2}(u)
=
t_1\oplus t_2
=
\min(-1,-1)
=
-1.
\]
But
\[
\angsem{\hat I}{Q^{\ge1}_1}(u)
=
(t_1\otimes(\mathbb 1 \ominus t_2)) \oplus (t_2\otimes(\mathbb 1 \ominus t_1))\oplus (t_1\otimes t_2\otimes(\mathbb 1 \ominus \mathbb 0))
=
\min(\infty,\infty,-2)
=
-2.
\]

Hence,
$
\angsem{\hat I}{Q^{\ge1}_1}
\neq
\angsem{\hat I}{Q^{\ge1}_2}.
$
\end{proof}

The three operators of the theorem suffice: $\text{\sql{COUNT(*)}}>C$ is
$\text{\sql{COUNT(*)}}\geq C+1$, $\text{\sql{COUNT(*)}}<C$ is
$\text{\sql{COUNT(*)}}\leq C-1$, and $\neq C$ is the disjunction of the two,
both for $Q_1$ (Definition~\ref{def:semantics}) and for the self-join
rewriting (set union $\cup$ of the two rewritings); the Lean statement
covers the six operators directly.

This mismatch across m-semirings is to be expected: outside of
the Boolean provenance world, semiring provenance depends on \emph{the way
a query is written}, not just on query equivalence.
Theorem~\ref{th:correctness} shows that our semantics is correct for
Boolean provenance (the Boolean-function m-semiring satisfies both
hypotheses) and generalizes%
\footnote{This is a proper generalization: the tropical m-semiring
  over $\mathbb{N}$, the Viterbi and the \L ukasiewicz m-semirings are
  absorptive with $\otimes$ distributive over $\ominus$
  (Table~\ref{tab:semirings}), but no semiring homomorphism from
  $\mathbb{B}[X]$ to any of them sends the variables of $X$ to arbitrary
  values
  \lean{Semirings/Tropical}{MinTropicalN.no_hom_from_BoolFunc}
  \lean{Semirings/Viterbi}{Viterbi.no_hom_from_BoolFunc}
\lean{Semirings/Lukasiewicz}{Lukasiewicz.no_hom_from_BoolFunc}.}
this equivalence to
other m-semirings. For other m-semirings, we simply choose the proposed
semantics as an implementable semantics for \sql{HAVING}.

\begin{proofsketch}
Fix a group key $a$ and let $U$ be the sequence of occurrences of its
group in~$\angsem{\hat I}{q}$; the order of $U$ plays no role here, so we
regard $U$ as a finite \emph{set} and its worlds as subsets $W\subseteq U$,
writing $A_W\defeq\bigotimes_{w\in W}\alpha_w$ for the product of their
annotations. Three quantities are compared:
$T_U(W)\defeq A_W\ominus\bigoplus_{w\in U\setminus W}A_{W\cup\{w\}}$, the
world annotation $\ann_U(W)$ in expanded form, which distributivity of
$\otimes$ over $\ominus$ provides
\lean{Having}{Having.T_eq_mul_one_monus_sum};
$F_C(U)\defeq\bigoplus_{|W|\ge C}T_U(W)$, the \sql{HAVING} side
$\angsem{\hat I}{Q^{\ge C}_1}(a)$ by Definition~\ref{def:semantics}; and
$S_C(U)\defeq\bigoplus_{|W|=C}A_W$, the self-join side
$\angsem{\hat I}{Q^{\ge C}_2}(a)$, since the self-join exhibits each
$C$-subset once and $\varepsilon$ merges the copies of the key by~$\oplus$.
The $\ge$ case of (i) is the identity $F_C(U)=S_C(U)$
\lean{Having}{Having.F_eq_S}: both sides satisfy the same include/exclude
recurrence $X_C(U)=X_C(U')\oplus\bigl(X_{C-1}(U')\otimes\alpha_u\bigr)$ for
$u\in U$ and $U'\defeq U\setminus\{u\}$ (proved in the appendix, the
recurrence for $F_C$ relying on the expansion
$A_V\le\bigoplus_{V\subseteq W\subseteq U}T_U(W)$ for $V\subseteq U$,
both obtained from the elementary monus inequalities of
Lemma~\ref{lem:monus-basic}(i)--(ii)),
and the recurrences bottom out at
$S_0(U')=\mathbb 1$ and $F_0(U')=\mathbb 1$, the latter being the only step
that needs absorptivity rather than the idempotence it implies
\lean{Having}{Having.F_zero_eq_one}. The $=$ and $\le$ cases are the
identities $\bigoplus_{|W|=C}T_U(W)=S_C(U)\ominus S_{C+1}(U)$ and
$\bigoplus_{1\le|W|\le C}T_U(W)=S_1(U)\ominus S_{C+1}(U)$
\lean{Having}{Having.G_eq_S_monus_S}
\lean{Having}{Having.atMost_eq_S_monus_S}: one inequality follows from the
per-world bound $T_U(W)\le S_j(U)\ominus S_{C+1}(U)$ for $j\le|W|\le C$
(proved in the appendix), the other from the right-distributivity of
$\ominus$ over $\oplus$, which idempotence grants
(Proposition~\ref{prop:idempotent_iff_add_monus}, in $n$-ary form
Corollary~\ref{lem:finite-add-monus}), and the
antitonicity of $\ominus$ (Lemma~\ref{lem:monus-basic}(iii)). The key $a$
being arbitrary, this proves~(i). For~(ii) we exhibit a counterexample in
a five-element absorptive chain in which $\otimes$ does not distribute
over $\ominus$
\lean{HavingQueryCounterexamples}{HavingQueryCounterexamples.ChainFive.query_counterexample},
and for~(iii) one in the tropical semiring over~$\mathbb Z\cup\{\infty\}$
\lean{HavingQueryCounterexamples}{HavingQueryCounterexamples.MinTropicalZ.query_counterexample}.
\end{proofsketch}

For \emph{monotone} \sql{HAVING} conditions -- preserved as a group
grows, such as $\text{\sql{COUNT(*)}}\geq C$ -- absorptivity
suffices: part~(i) holds without distributivity of $\otimes$
over~$\ominus$, and extends to every monotone condition classically
rewritable into a positive aggregate-free query (see appendix).

\begin{toappendix}
\subsection{Monotone Conditions: Absorptivity Suffices}
\label{appendix:monotone}

This appendix isolates the class of \sql{HAVING} conditions for which
the first hypothesis of Theorem~\ref{th:correctness}(i), absorptivity,
suffices on its own. The role of the second hypothesis, distributivity
of $\otimes$ over~$\ominus$, is tied to the difference operator: the
rewritings $Q^{\leq C}_2$ and $Q^{=C}_2$ use it, and part~(ii) of the
theorem shows that distributivity is then unavoidable. When the
rewriting is \emph{positive}, as $Q^{\geq C}_2$ is, distributivity can
be dropped: the monus cancels from the possible-world sum.

\begin{definition}\label{def:monotone}
  In the setting of Definition~\ref{def:semantics}, an atomic
  comparison $t\op t'$ \emph{holds in} a non-empty world
  $W\sqsubseteq U$ if $\mathrm{val}_t(W)\op\mathrm{val}_{t'}(W)$; it
  is \emph{monotone} if, for every occurrence sequence~$U$, the family
  of the non-empty worlds in which it holds is closed under
  supersequences within~$U$. A \emph{monotone condition} is a Boolean
  combination, using $\land$ and~$\lor$ only, of monotone atomic
  comparisons.
\end{definition}
Note that ``holds in'' is used only atom by atom: the provenance of a
compound condition remains that of Section~\ref{sec:semantics},
interpreting $\land$ by $\otimes$ and $\lor$ by $\oplus$ of the
atoms' predicate provenances; summing $\ann_U(W)$ over the worlds
where the whole combination holds would define a different
semantics.
Monotone atomic comparisons include $\text{\sql{COUNT(*)}}\geq C$
and~$>C$, since adding occurrences only increases the count;
$\text{\sql{MIN}}(t)\leq c$ and~$<c$, since a new occurrence can only
lower the minimum; and symmetrically $\text{\sql{MAX}}(t)\geq c$
and~$>c$. In contrast, an added occurrence can falsify
$\text{\sql{COUNT(*)}}\leq C$ or $=C$, and negation exchanges
monotone and non-monotone atoms, whence its exclusion.
$\text{\sql{SUM}}(t)\geq c$ is monotone when the aggregated values
are non-negative, though it admits no aggregate-free rewriting in
general.

The key observation is that on families of worlds closed under
supersets, the monus contributes nothing to the possible-world sum:
\begin{proposition}[monus cancellation]\label{prop:monus-cancellation}
  Let $(\mathbb K,\oplus,\otimes,\ominus,\mathbb 0,\mathbb 1)$ be a
  commutative m-semiring, $U$ a finite set and $(\alpha_w)_{w\in U}$ a
  family of elements of~$\mathbb K$; for $W\subseteq U$, let
  $A_W\defeq\bigotimes_{w\in W}\alpha_w$,
  $E_W\defeq\bigoplus_{w\in U\setminus W}\alpha_w$,
  $\ann_U(W)=A_W\otimes(\mathbb 1\ominus E_W)$ as in
  Definition~\ref{def:possible_world}, and
  $T_U(W)\defeq A_W\ominus\bigoplus_{w\in U\setminus W}A_{W\cup\{w\}}$
  as in Lemma~\ref{lem:upward-expansion}.
  \begin{enumerate}[(i)]
    \item
      \lean{Having}{Having.monus_factor_le}
      $T_U(W)\leq\ann_U(W)\leq A_W$ for every $W\subseteq U$.
    \item
      \lean{Having}{Having.witness_identity}
      \lean{Having}{Having.sum_T_eq_sum_A}
      If $\mathbb K$ is idempotent then, for every family
      $\mathcal V$ of non-empty subsets of~$U$ closed under supersets
      (i.e., $W\in\mathcal V$ and $W\subseteq W'\subseteq U$ imply
      $W'\in\mathcal V$),
      \[
        \bigoplus_{W\in\mathcal V}\ann_U(W)
        =\bigoplus_{W\in\mathcal V}T_U(W)
        =\bigoplus_{W\in\mathcal V}A_W.
      \]
    \item
      \lean{Having}{Having.witness_minimal}
      If $\mathbb K$ is moreover absorptive, these sums further
      equal $\bigoplus_{W\in\min\mathcal V}A_W$, where
      $\min\mathcal V$ denotes the set of $\subseteq$-minimal elements
      of~$\mathcal V$.
  \end{enumerate}
\end{proposition}
\begin{proof}
  (i) For the right inequality, $\mathbb 1\ominus E_W\le\mathbb 1$ by
  Lemma~\ref{lem:monus-basic}(i), and $\otimes$ is monotone. For the
  left one, distributivity of $\otimes$ over $\oplus$ gives
  $\bigoplus_{w\in U\setminus W}A_{W\cup\{w\}}=A_W\otimes E_W$, so by
  the residuation law~\eqref{eq:residuation} it suffices that
  $A_W\le(A_W\otimes E_W)\oplus\bigl(A_W\otimes(\mathbb 1\ominus
  E_W)\bigr)=A_W\otimes\bigl(E_W\oplus(\mathbb 1\ominus E_W)\bigr)$,
  which holds since $\mathbb 1\le E_W\oplus(\mathbb 1\ominus E_W)$ by
  Lemma~\ref{lem:monus-basic}(ii) and $\otimes$ is monotone.

  (ii) Summing~(i) over $\mathcal V$ yields
  $\bigoplus_{\mathcal V}T_U(W)\le\bigoplus_{\mathcal V}\ann_U(W)
  \le\bigoplus_{\mathcal V}A_W$. Conversely, for each
  $W\in\mathcal V$, Lemma~\ref{lem:upward-expansion} gives
  $A_W\le\bigoplus_{W\subseteq W'\subseteq U}T_U(W')$, and every such
  $W'$ belongs to~$\mathcal V$ by closure under supersets; summing
  over $\mathcal V$ and collapsing duplicate summands by idempotence
  of~$\oplus$,
  $\bigoplus_{\mathcal V}A_W\le\bigoplus_{\mathcal V}T_U(W)$, which
  closes the cycle of inequalities.

  (iii) Absorptivity gives $a\otimes b\le a$, hence $A_W\le A_{W'}$
  whenever $W'\subseteq W$: every summand of
  $\bigoplus_{\mathcal V}A_W$ is dominated by one over
  $\min\mathcal V$, and $x\le y$ implies $x\oplus y=y$.
\end{proof}
Neither absorptivity nor distributivity of $\otimes$ over~$\ominus$
enters parts~(i) and~(ii). The proposition is also indifferent to
rewritability: it applies, e.g., to $\text{\sql{SUM}}(t)\geq c$ over
non-negative values, for which no aggregate-free rewriting is
available. It yields the announced refinement of
Theorem~\ref{th:correctness}, in which the distributivity hypothesis
is dropped:

\begin{theorem}\label{th:monotone}
  Let $\mathbb K$ be a commutative \emph{absorptive} m-semiring and
  $\hat I$ a $\mathbb K$-instance.
  \begin{enumerate}[(i)]
    \item
      \lean{HavingQueryCorrectness}{Query.joinCount_monotone_correct}
      \lean{HavingJoinCompositional}{countHaving_site_rewrite_monotone}
      In the setting of Theorem~\ref{th:correctness},
      $\angsem{\hat I}{Q^{\geq C}_1}=\angsem{\hat I}{Q^{\geq C}_2}$.
    \item
      \lean{HavingMonotone}{MonoCond.site_rewrite}
      More generally, let $q$ be an aggregate-free query with a
      group-key column $\#1$ and an occurrence-identifier column, let
      $\gamma^{\preceq}_{\#1}[\dots](q)$ compute the aggregates
      mentioned by a monotone condition $\psi$ over the atoms
      $\text{\sql{COUNT(*)}}\geq C$, $>C$,
      $\text{\sql{MIN}}(t)\leq c$, $<c$, $\text{\sql{MAX}}(t)\geq c$,
      $>c$, and let the positive aggregate-free rewriting $Q_\psi$ be
      defined compositionally: an atom on \sql{COUNT(*)} is rewritten
      as $Q^{\geq C}_2$ of Theorem~\ref{th:correctness}, an atom on
      \sql{MIN} or \sql{MAX} as
      $\varepsilon(\Pi_{\#1}(\sigma_{t\op c}(q)))$, a conjunction as
      the join of the two rewritings on the group key, projected back
      to it and duplicate-eliminated, and a
      disjunction as their set union. Then
      $\angsem{\hat
      I}{\Pi_{\#1}(\sigma_\psi(\gamma^{\preceq}_{\#1}[\dots](q)))}
      =\angsem{\hat I}{Q_\psi}$.
  \end{enumerate}
\end{theorem}
\begin{proof}
  Fix a group key~$a$ and let
  $U\defeq U^{\preceq}_{\angsem{\hat I}{q},(a)}$ be the sequence of
  the occurrences of its group, with annotations
  $(\alpha_w)_{w\in U}$; as in the proof of
  Theorem~\ref{th:correctness}, a group whose annotation on either
  side is $\mathbb 0$ is absent from it, so it suffices to prove the
  annotations equal, key by key.

  (i) By Definition~\ref{def:semantics}, the annotation of the group
  on the left-hand side is $\bigoplus_{W\in\mathcal V}\ann_U(W)$,
  where $\mathcal V\defeq\{W\neq\emptyset:|W|\geq C\}$. This family
  is closed under supersets
  and its minimal elements are the $C$-element subsets, so
  Proposition~\ref{prop:monus-cancellation}(ii)--(iii) (absorptivity
  implies idempotence) give
  $\bigoplus_{\mathcal V}\ann_U(W)=\bigoplus_{|W|=C}A_W=S_C(U)$. On
  the right-hand side, as computed in the proof of
  Theorem~\ref{th:correctness}, the join chain of $Q^{\geq C}_2$
  derives each $C$-element subset of the group exactly once, thanks to
  the strict comparisons on the identifier column, and $\varepsilon$
  merges the copies of the key by~$\oplus$, giving $S_C(U)$ as well;
  this computation uses no hypothesis on the m-semiring.

  (ii)
  \lean{Having}{Having.sum_ann_meet}
  \lean{HavingSemantics}{Having.havingProv_existential}
  For an atom $\text{\sql{MIN}}(t)\leq c$ (the other \sql{MIN}
  and \sql{MAX} atoms are identical, and \sql{COUNT(*)} atoms are
  part~(i), with $>C$ being $\geq C+1$), the valid worlds are the
  non-empty $W$ meeting the set
  $S\defeq\{w\in U: t(w)\leq c\}$ of qualifying occurrences; this
  family is closed under supersets and its minimal elements are the
  singletons $\{w\}$, $w\in S$, so
  Proposition~\ref{prop:monus-cancellation} gives the predicate
  provenance $\bigoplus_{w\in S}\alpha_w$, which is exactly the
  annotation of key~$a$ in
  $\varepsilon(\Pi_{\#1}(\sigma_{t\leq c}(q)))$. For the Boolean
  combinations, the semantics of Section~\ref{sec:semantics}
  interprets $\land$ by $\otimes$ and $\lor$ by $\oplus$ of the
  predicate provenances, which is precisely what joining the two
  rewritings on the group key, respectively taking their set union,
  computes on the annotation of each key; the claim follows by
  induction on~$\psi$. As for Theorem~\ref{th:correctness}, the Lean
  statement pads the rewriting with a $\mathbb 0$-annotated row per
  failing group key, matching the convention that
  $\mathbb 0$-annotated tuples are absent; it fixes a three-column
  base query (group key, aggregated value, occurrence identifier), as
  Theorem~\ref{th:correctness} fixes a two-column one.
\end{proof}

Three remarks are in order. First, the equality of
Theorem~\ref{th:monotone}(ii) is for the rewritings exhibited, and
cannot hold for every classically equivalent positive rewriting, since
provenance depends on how a query is written: rewriting
$\text{\sql{COUNT(*)}}\geq 1$ as
$\varepsilon(\Pi_{\#1}(\sigma_{\#2=\#4}(q\bowtie_{\#1=\#3}q)))$, which
pairs each occurrence with itself, is classically equivalent to
$\varepsilon(\Pi_{\#1}(q))$ but annotates the group by
$\bigoplus_w\alpha_w\otimes\alpha_w$, which in the (absorptive)
tropical semiring over~$\mathbb N$ is $\min_w2\alpha_w$ rather than
$\min_w\alpha_w$. The rewritings of Theorem~\ref{th:monotone} use
every occurrence at most once per derivation, as the tie-break chain
of $Q^{\geq C}_2$ does. Second, no assumption is placed on the base
query~$q$, which may in particular use the difference operator: the
proof manipulates the annotations of the group's occurrences as
opaque elements of~$\mathbb K$. Third, the restriction to positive
rewritings is what the counterexample of
Theorem~\ref{th:correctness}(ii) shows to be unavoidable: as soon as
the rewriting requires the difference operator, as for $\leq$
and~$=$, distributivity of $\otimes$ over~$\ominus$ becomes necessary
even over absorptive m-semirings. The separation is sharp: on the very
m-semiring and instance of that counterexample,
$\text{\sql{COUNT(*)}}\geq 1$ and $\geq 2$ agree with their rewritings
\lean{HavingQueryCounterexamples}{HavingQueryCounterexamples.ChainFive.query_ge_agree}
\lean{HavingQueryCounterexamples}{HavingQueryCounterexamples.ChainFive.query_ge_two_agree}
while $=1$ does not.
\end{toappendix}

\inlinepar{Specialization, probabilistic evaluation, and complexity.}
Much of the practical value of semiring provenance comes from the fact
that provenance computed once in a sufficiently general semiring can be
specialized to any other one by applying a homomorphism, instead of
re-evaluating the query. This is the ``compile once, evaluate many''
principle that ProvSQL implements by storing a single provenance circuit
and evaluating it in the semiring of interest~\cite{sen2026provsql}. The
principle survives the introduction of aggregate comparisons.
For a homomorphism of m-semirings $h:\mathbb K\to\mathbb K'$, we write $h(\hat I)$ for the $\mathbb
K'$-instance obtained from~$\hat I$ by applying $h$ to every annotation.

\begin{toappendix}
\subsection{Specialization, Probabilistic Evaluation, and Complexity}
We prove here the results of Section~\ref{sec:results} on specialization
by homomorphisms, probabilistic evaluation, and complexity.
\end{toappendix}
\begin{propositionrep}\label{prop:hom}
  \lean{AggQueryHom}{AggQuery.evaluateAnnotated_hom}
  Let $h:\mathbb K\to\mathbb K'$ be an m-semiring homomorphism. Then, for
  every \sql{HAVING} query $Q$
  and every $\mathbb K$-instance $\hat I$,
  \(
    \angsem{h(\hat I)}{Q}=h\bigl(\angsem{\hat I}{Q}\bigr).
  \)
\end{propositionrep}

\begin{appendixproof}
The syntax of Definition~\ref{def:general} mentions no annotation, so both
sides evaluate the same query and only the instance differs; we proceed
by structural induction on~$Q$. The induction hypothesis cannot be an
equality of intermediate results, because the semantics does not only
combine annotations, it also \emph{inspects} them, in two places:
\begin{enumerate}[(a)]
\item the supersede condition tests occurrence sequences for equality,
  so a non-injective $h$ may identify two groups that $\mathbb K$
  distinguishes, and drop on the $\mathbb K'$ side a pending factor that
  the $\mathbb K$ side keeps;
\item Definition~\ref{def:possible_world} orders occurrences of equal
  tuples arbitrarily; any concrete choice may break ties using the
  annotations, and $h$ need not preserve it.
\end{enumerate}

\smallskip\noindent\emph{The invariant.} Say that an annotated tuple
$(u',\alpha')$ over $\mathbb K'$ \emph{simulates} an annotated tuple
$(u,\alpha)$ over $\mathbb K$ when
\begin{enumerate}[(i)]
\item $u'$ and $u$ agree on their regular columns;
\item in each aggregate column, the two tokens have the same aggregate
  function, and the occurrence sequence of the one in~$u'$ is obtained
  from the $h$-image of the one in~$u$ by permuting occurrences
  \emph{within blocks of equal value};
\item $\alpha'=h(\alpha)$, where the annotations are the elements of
  $\mathbb K'$ and $\mathbb K$ \emph{denoted} by the factored forms of
  Definition~\ref{def:general}, i.e., taken after the pending factors
  have been multiplied in.
\end{enumerate}
We prove by induction on $Q$ that the tuples of
$\angsem{h(\hat I)}{Q}$ and those of $\angsem{\hat I}{Q}$ are in a
multiplicity-preserving bijection under which each simulates its
counterpart. This suffices: (i) and (ii) give equal data parts, since the
value of an aggregate column is the aggregate of the values of its
token's occurrences and a permutation within blocks of equal value leaves
that sequence of values unchanged, and (iii) is the required equality of
annotations.

\smallskip\noindent\emph{Two neutrality lemmas.} Both say that a
rearrangement of the factored annotation does not change the element it
denotes, on either side separately.
\begin{enumerate}[(1)]
\item \emph{Cashing is neutral} \lean{AggQueryHom}{GenAnn.finalize_cash}:
  for a factored annotation $\lAngle\beta,P\rAngle$ and any
  $P_0\msubseteq P$, the pair
  $\lAngle\beta\otimes\bigotimes_{U\in P_0}\delta(\bigoplus_{(u,\alpha)\in
  U}\alpha),\,P\msetminus P_0\rAngle$ denotes the same element, by
  associativity and commutativity of~$\otimes$.
\item \emph{Superseding is neutral}
  \lean{AggQueryHom}{GenPred.predsem_delta_absorb}: if $\psi$ entails the
  existence of the group of $U$ and every token it compares has
  occurrence sequence~$U$, then
  \[
    \predsem{\hat I}{q}{\psi}(u)\otimes
      \delta\bigl(\bigoplus_{(u,\alpha)\in U}\alpha\bigr)
    =\predsem{\hat I}{q}{\psi}(u).
  \]
  Dropping the pending factor of $U$ next to $\predsem{\hat
  I}{q}{\psi}(u)$ is therefore invisible, whether or not the evaluation
  actually drops it.
\end{enumerate}
Lemma~(2) is the substantive use of the $\delta$-absorption axiom of
Section~\ref{sec:preliminaries}. Negation being pushed to the atoms,
$\psi$ is a $\land/\lor$-combination of atoms. For an aggregate atom, by
distributivity of $\otimes$ over $\oplus$ it suffices to absorb the guard
into one summand $\ann_U(W)\otimes\chi_{\op}(\cdots)$ of the
possible-world sum, with $W\neq\emptyset$: picking an occurrence $w\in W$,
its annotation $\alpha_w$ is a factor of $\ann_U(W)$, and, writing
$\bigoplus_{(u,\alpha)\in U}\alpha=\alpha_w\oplus\gamma$, $\delta$-absorption gives
$\alpha_w\otimes\delta(\alpha_w\oplus\gamma)=\alpha_w$. A conjunction
entails existence as soon as one conjunct does, and that conjunct absorbs
the guard while the other factor is left untouched; a disjunction entails
only when both disjuncts do, and then each summand absorbs it. A regular
atom entails nothing, and indeed its characteristic value $\chi_{\op}$
does not absorb the guard -- which is why the entailment test is part of
the semantics.

\smallskip\noindent\emph{The induction.} We record first that the
operators which inspect \emph{data} -- $\varepsilon$, $-$ and
$\gamma^{\preceq}$ -- apply to all-regular relations by the kind
discipline of Definition~\ref{def:general}. Clause~(ii) is vacuous there,
so on their inputs simulation is plain equality of tuples together with
$\alpha'=h(\alpha)$.
\begin{itemize}
\item \emph{Base relations.} Every tuple of $h(\hat I)$ is the $h$-image
  of the corresponding tuple of $\hat I$, with nothing pending.
\item \emph{$\gamma^{\preceq}$.} By the remark above the two inputs carry
  the same tuples, so the two sides produce the same group keys. Fix a
  key: the occurrence sequence built on the $\mathbb K'$ side and the
  $h$-image of the one built on the $\mathbb K$ side are both sorted by
  the tuple $u$ and contain the same occurrences, hence differ only
  inside blocks of equal tuples; as the aggregated term reads the
  tuple only, this is a permutation of the kind allowed by~(ii).
  This is where obstruction~(b) is discharged. The output annotation is
  $\delta(\bigoplus_{(u,\alpha)\in U}\alpha)$ on either side, and $h$
  commutes with $\oplus$ and with~$\delta$, which gives~(iii).
\item \emph{Selection.} If $\psi$ has no aggregate atom, its terms range
  over regular columns only, whose values simulation preserves by~(i), so
  the two sides retain matching tuples. Otherwise, by
  lemmas~(1) and~(2), the annotation produced on each side denotes
  $\predsem{\hat I}{q}{\psi}(u)$ times the element denoted by the input
  annotation, \emph{regardless of which pending factors that side
  superseded}; this discharges obstruction~(a). It remains to see that
  the predicate provenance itself commutes with~$h$, which it does
  because it is built from $\oplus$, $\otimes$, $\ominus$ and values
  $\chi_{\op}$ that are $\mathbb 0$ or $\mathbb 1$ on either side, the
  aggregate values being untouched by the pushforward; and it is
  invariant under the permutations of~(ii), by the same argument as in
  the $\gamma^{\preceq}$ case.
\item \emph{Projection.} Terms over regular columns evaluate equally
  by~(i) and~(ii), and aggregate columns are copied verbatim. The two
  sides may cash different pending factors, since they compare occurrence
  sequences, but by lemma~(1) this does not change the denoted
  annotation.
\item \emph{Product and union.} The product multiplies the concrete parts
  of the two annotations and concatenates their pending factors, so the
  denoted annotations are multiplied, and $h$ commutes with $\otimes$;
  the union is the multiset sum of the two results and combines no
  annotation at all.
\item \emph{$\varepsilon$ and $-$.} All-regular, so by the remark above
  the two inputs have equal data parts and $h$-related annotations.
  Duplicate elimination $\oplus$-sums the annotations of equal tuples,
  and $h$ commutes with $\oplus$; difference subtracts them with
  $\ominus$, and $h$ commutes with $\ominus$ by definition.
\end{itemize}
Every operator of Definition~\ref{def:general} is covered, which
concludes the induction and the proof.
\end{appendixproof}

In addition to this important result, our m-semiring semantics
for \sql{HAVING} also provides an implementable semantics for probabilistic query
evaluation over tuple-independent probabilistic databases, in the classical
sense of~\cite{suciu2011probabilistic}: below, $\Pr(I\models Q)$ is the probability
that a Boolean query $Q$ holds on a random sub-instance of~$I$, and the
probability of a Boolean function is that of its being true under a random
valuation.

\begin{propositionrep}
  \label{prop:pqe-equiv}
  \lean{AggQueryProbability}{AggQuery.boolean_pqe}
  Let $Q$ be a \emph{Boolean} \sql{HAVING} query
  and $I$ be an instance where each tuple $u$ is associated with an
  independent probability $\Pr(u)\in [0,1]$.
  Then $\Pr(I\models\nobreak Q)=\Pr(\angsem{\hat I}{Q})$, where $\hat I$ is a
  $\mathbb{B}[X]$-instance whose tuples are annotated with independent
  variables~$x\in X$ with the same probability as the corresponding tuple
  in the probabilistic instance.
\end{propositionrep}
\begin{proof}
  No restriction is placed on where the aggregate comparisons occur:
  $Q$ is interpreted by Definition~\ref{def:general}, so a comparison may
  be applied to tokens carried through projections, joins, unions and
  further selections, arbitrarily far from the grouping that produced
  them.

By definition \cite{suciu2011probabilistic}, probabilistic query
evaluation under the possible-world semantics is
\[
  \Pr(I\models Q)=\sum_{\substack{S\msubseteq I,\\ S\models Q}}\Pr(S),
  \qquad
  \Pr(S)\defeq\prod_{u\in S}\Pr(u)\prod_{u\in I\msetminus S}(1-\Pr(u)),
\]
where $S\msubseteq I$ ranges over the sub-instances of $I$ and $S\models Q$
means that $Q$ evaluates to true on $S$. As $\hat I$ is tuple-independent,
these sub-instances are in bijection with the valuations
$\nu\colon X\to\{\bot,\top\}$ through
$S_\nu\defeq\mset{u\in\hat I\mid\nu(x_u)=\top}$, with
$\Pr(S_\nu)=\Pr(\nu)=\prod_{u:\,\nu(x_u)=\top}\Pr(u)\prod_{u:\,\nu(x_u)=\bot}(1-\Pr(u))$.
Since, by definition, the probability of a Boolean function satisfies
$\Pr(\angsem{\hat I}{Q})=\sum_{\nu\models\angsem{\hat I}{Q}}\Pr(\nu)$,
comparing the two sums shows that the claim
$\Pr(I\models Q)=\Pr(\angsem{\hat I}{Q})$ is \emph{equivalent} to the
characteristic property
\begin{equation}\label{eq:char-prop}
  \nu\bigl(\angsem{\hat I}{Q}\bigr)=\top
  \iff
  S_\nu\models Q
  \qquad\text{for every valuation }\nu,
\end{equation}
i.e., that the Boolean provenance is true under $\nu$ exactly on the possible
worlds where $Q$ holds.

For every operator of $Q$ but the aggregate comparison, \eqref{eq:char-prop}
is exactly the compatibility of the $\mathbb{B}[X]$-semantics with
probabilistic query evaluation for queries \emph{without} aggregate
comparison, established in~\cite[Theorem~12]{sen2026provsql}. It thus only
remains to check that \eqref{eq:char-prop} is preserved by the
aggregate-comparison selection of Definition~\ref{def:semantics}, i.e., that
our \sql{HAVING} semantics matches the Boolean one.

Let $\psi=(t\op t')$ be such a comparison, and let $U$ be the occurrence
sequence carried by the token it compares, i.e., the possible-world
sequence $U^{\preceq}_{\angsem{\hat I}{q'},\vec v}$ of the group of the
originating $\gamma^{\preceq}(q')$, wherever that grouping sits in $Q$.
Fix a valuation~$\nu$. In $\mathbb{B}[X]$, where $\oplus=\lor$, $\otimes=\land$ and
$a\ominus b=a\land\lnot b$, the annotation of a world $W\sqsubseteq U$ is
\[
  \ann_U(W)=\Bigl(\bigwedge_{(u,x_u)\in W}x_u\Bigr)\land
            \Bigl(\bigwedge_{(u,x_u)\in U\setminus W}\lnot x_u\Bigr),
\]
which is true under $\nu$ if and only if $W$ is exactly the set of occurrences
of the group present in $S_\nu$, namely
$W=W_\nu\defeq\mset{(u,x_u)\in U\mid\nu(x_u)=\top}$. Hence, in the predicate
provenance
\(
  \predsem{\hat I}{q}{\psi}(u)=\bigoplus_{\emptyset\neq W\sqsubseteq U}
  \ann_U(W)\otimes\chi_{\op}(\mathrm{val}_t(W),\mathrm{val}_{t'}(W)),
\)
a single disjunct is active under $\nu$ -- the one for $W_\nu$ -- so
\[
  \nu\bigl(\predsem{\hat I}{q}{\psi}(u)\bigr)=\top
  \iff
  W_\nu\neq\emptyset \;\text{ and }\;
  \mathrm{val}_t(W_\nu)\op\mathrm{val}_{t'}(W_\nu),
\]
that is, if and only if the group of~$u$ is non-empty in $S_\nu$ and its
aggregate values there satisfy the comparison. This is precisely the ordinary
Boolean evaluation of the \sql{HAVING} predicate $\psi$ on the possible
world~$S_\nu$.

For a Boolean combination $\psi$ of atoms over the same grouping, the
claim extends by structural induction: negation is pushed
to the atoms (Definition~\ref{def:semantics}), where it complements the
comparison operator, and since $\otimes=\land$ and $\oplus=\lor$ act
pointwise on $\mathbb B[X]$ while all aggregate atoms of that grouping
share the same active world $W_\nu$, the predicate provenance
$\predsem{\hat I}{q}{\psi}(u)$ is true under $\nu$ if and only if
$W_\nu\neq\emptyset$ and $\psi$ holds classically on $W_\nu$. A regular
atom mixed into $\psi$ contributes its characteristic value
$\chi_{\op}$, i.e., its classical truth value on the output tuple, so the
same equivalence holds -- except that $\psi$ may then be satisfied in a
world where the group is empty, which is exactly the case in which
Definition~\ref{def:general} keeps the group's pending factor.

It remains to account for those pending factors. A group whose tokens are
never compared, or that is compared by a predicate not entailing its
existence, contributes to the annotation of a row the factor
$\delta\bigl(\bigoplus_{(u,x_u)\in U}x_u\bigr)$ of
Definition~\ref{def:general}. Over $\mathbb B[X]$ the operator $\delta$ is
the identity (Section~\ref{sec:preliminaries}), so this factor is
$\bigvee_{(u,x_u)\in U}x_u$, which is true under $\nu$ if and only if
$W_\nu\neq\emptyset$, i.e., if and only if the group is non-empty in
$S_\nu$ -- precisely the condition under which the grouping produces that
row on $S_\nu$. This is where the choice $\delta\defeq\mathrm{id}$ over
$\mathbb B[X]$ is used: the support indicator would make the factor true
in every world.

Finally, $\nu(\angsem{\hat I}{Q})=\top$ if and only if the annotation
of some output row is true under~$\nu$, i.e., by the two paragraphs
above, if and only if some row of $Q$ on $S_\nu$ survives every
comparison it undergoes and every group it depends on is non-empty in
$S_\nu$; that is, exactly when the output of $Q$ on $S_\nu$ is non-empty.
Property~\eqref{eq:char-prop} is therefore preserved by every operator of
Definition~\ref{def:general} and, together
with~\cite[Theorem~12]{sen2026provsql} for the operators already covered
there, holds for the whole query~$Q$. We conclude that
$\Pr(I\models Q)=\Pr(\angsem{\hat I}{Q})$.
\end{proof}

Unsurprisingly, in semirings whose annotations may be \emph{independent
generators} (so that distinct possible worlds carry distinct, non-cancelling
annotations), evaluating the \sql{HAVING} semantics is intractable. This is
the case of the provenance semirings $\mathbb{B}[X]$ and $\mathbb{N}[X]$, for
which deciding whether a provenance is non-zero is NP-complete, and this
already in \emph{data} complexity, i.e., for a single fixed query.

\begin{propositionrep}\label{prop:having-nphard}
  \lean{HavingComplexity}{Provenance.Complexity.havingSumNonzero_NP_complete}
  Let $\mathbb{K}\in\{\mathbb{B}[X],\mathbb{N}[X]\}$. There is a fixed
  Boolean \sql{HAVING} query~$q$ such that deciding, given as sole input a
  $\mathbb{K}$-instance~$\hat I$ whose tuples are annotated by distinct
  variables, whether the provenance of~$q$ over~$\hat I$ is
  non-$\mathbb{0}_\mathbb{K}$, is NP-complete. Moreover, deciding whether
  the provenance of a fixed \sql{HAVING} query over a
  $\mathbb B[X]$-instance is non-$\mathbb 0$ is in NP.
\end{propositionrep}

\begin{proof}
  \emph{Membership.} Let first $\mathbb K=\mathbb B[X]$ and $Q$ be an
  arbitrary fixed \sql{HAVING} query. Guess a valuation
  $\nu\colon X\to\{\bot,\top\}$ of the variables of~$\hat I$. Evaluation
  under~$\nu$ is an m-semiring homomorphism $h_\nu\colon\mathbb
  B[X]\to\mathbb B$ (it commutes with $\lor$, $\land$, $\land\lnot$ and
  with $\delta=\mathrm{id}$), so by Proposition~\ref{prop:hom},
  $h_\nu(\angsem{\hat I}{Q})=\angsem{h_\nu(\hat I)}{Q}$, the
  $\mathbb B$-provenance of $Q$ over the sub-instance $S_\nu$ of the tuples
  annotated $\top$ (those annotated $\bot$ being absent). By
  Proposition~\ref{prop:pqe-equiv}, applied with all probabilities equal
  to~$1$, this is $\top$ if and only if $S_\nu\models Q$, which is decided
  in polynomial time in data complexity. Since $\angsem{\hat I}{Q}\neq\mathbb
  0$ if and only if some valuation makes it true, the problem is in NP.
  For $\mathbb K=\mathbb N[X]$, the fixed query~$q$ below, and an instance
  annotated by distinct variables, every factor $\mathbb
  1\ominus\bigoplus_{i\notin W}x_i$ equals~$\mathbb 1$, so the provenance is
  a polynomial with nonnegative coefficients whose monomials are exactly the
  valid worlds; guessing one world and checking its sum places the problem
  in NP
  \lean{HavingComplexity}{Provenance.Complexity.havingSumNonzero_mem_NP}.

  \emph{Hardness.} We reduce from subset-sum \cite{DBLP:conf/coco/Karp72}, a
  classic NP-hard problem: given a finite multiset $M=\mset{m_1,\dots,m_n}$ of
  natural numbers and a positive integer $B$, decide whether some
  sub-multiset of $M$ sums to exactly~$B$.

  The reduction outputs an \emph{instance only}: the query below is fixed
  once and for all, and the target~$B$ is carried by the data. We construct a
  $\mathbb{K}$-instance $\hat I$ formed of a binary relation $R(b,v)$
  containing the $n$ tuples $(B,m_1),\dots,(B,m_n)$, annotating tuple
  $(B,m_i)$ with a \emph{distinct variable} $x_i$, and take the query
  \begin{minted}{postgresql}
q: SELECT DISTINCT 1 FROM R GROUP BY b HAVING SUM(v) = b
  \end{minted}
  that is, in our relational algebra,
  \[
    q=\varepsilon\Bigl(\Pi_{1}\bigl(
        \sigma_{\#2=\#1}\bigl(
          \gamma^{\preceq}_{\#1}[\#2{:}f_{\mathrm{SUM}}](R)\bigr)
      \bigr)\Bigr),
  \]
  where $\gamma^{\preceq}_{\#1}[\#2{:}f_{\mathrm{SUM}}]$ groups $R$ by its
  first column and aggregates the $v$-values of each group with \sql{SUM}, so
  that an output tuple is $(b,s)$ with $b$ the group key and $s$ the
  aggregate column; $\sigma_{\#2=\#1}$ compares the aggregate column with the
  regular column~$\#1$, which Definition~\ref{def:semantics} permits since
  the value in a world~$W$ of a side of regular type is $s(u)$, a data value
  read off the output tuple rather than a query constant; and $\Pi_1$
  projects onto the constant~$1$, whose occurrences $\varepsilon$ merges
  (the order $\preceq$ being irrelevant since \sql{SUM} is commutative).

  All $n$ tuples carry the same group key~$B$, so there is a single group.
  Following Definition~\ref{def:semantics}, and writing
  $U\defeq U^{\preceq}_{\angsem{\hat I}{R},(B)}$, the provenance of the
  result $(1)$ is
\[
   P\defeq\bigoplus_{\emptyset\neq W\sqsubseteq U}
  \ann_U(W)\otimes
  \chi_{=}\bigl(\mathrm{agg}_{\#2,f_{\mathrm{SUM}}}(W),\ B\bigr)
=
\bigoplus_{\substack{\emptyset\neq W\sqsubseteq U\\ \sum_{(u,\alpha)\in W} u_2=B}}
\ann_U(W),
\]
where $f_{\mathrm{SUM}}$ denotes the \sql{SUM} aggregate and $u_2$ the
$v$-value of the tuple~$u$. Each possible world
$W$ is a sub-multiset of~$M$, and because the $x_i$ are distinct variables its
annotation is non-$\mathbb{0}_\mathbb{K}$ and the worlds do not cancel: in
$\mathbb{B}[X]$, $\ann_U(W)=\bigwedge_{i\in W}x_i\wedge\bigwedge_{i\notin
W}\lnot x_i$ are distinct minterms
\lean{HavingComplexity}{Provenance.Complexity.havingSumProvBool_ne_zero_iff}; in
$\mathbb{N}[X]$, the monus factor
$\mathbb{1}\ominus\bigoplus_{i\notin W}x_i$ reduces to $\mathbb{1}$ (the
excluded variables have no constant term), so $\ann_U(W)=\bigotimes_{i\in
W}x_i$ are distinct monomials
\lean{HavingComplexity}{Provenance.Complexity.havingSumProv_ne_zero_iff}. Hence $P\neq\mathbb{0}_\mathbb{K}$ if and only
if some non-empty sub-multiset of~$M$ sums to~$B$, i.e., if and only if the
subset-sum instance has a solution. The reduction is polynomial, the
values being written in binary
\lean{HavingComplexity}{Provenance.Complexity.havingSumNonzeroBool_faithful}%
\lean{HavingComplexity}{Provenance.Complexity.havingSumNonzeroHow_faithful},
and, the
query~$q$ being independent of the subset-sum instance, it establishes
hardness in data complexity, which concludes.
\end{proof}

We shall see next that we can isolate cases where the
semantics can be evaluated tractably.

\section{Algorithms}
\label{sec:algorithms}

In this section we describe specialized procedures for evaluating the valid possible worlds for selections involving aggregate comparisons.


All our algorithms are applied \emph{per group key} $\vec v$, on the
$\preceq$-ordered sequence of group-member
occurrences~$U\defeq\nobreak U^{\preceq}_{\angsem{\hat I}{q'},\vec v}$ of the
\emph{pre-aggregation} sub-query $q'$
(Definition~\ref{def:possible_world}); the order being fixed, we
identify a possible world $W\sqsubseteq U$ with its set of positions, an
index set $W\subseteq [|U|]$. Terms are evaluated on tuples produced
by~$q'$, whose arity is denoted~$k$.

We first provide a dynamic programming algorithm for \sql{SUM} comparisons that incrementally constructs all possible worlds producing the desired sum, named \textsc{ValidSum}.
It processes each tuple iteratively and maintains, for every possible value of the sum, the collection of worlds that realize it.
For a new tuple, previously constructed possible worlds can either remain unchanged or be extended by inclusion of this new tuple, hence creating new worlds having larger sums.
After all the tuples have been encountered, the algorithm returns a collection of non-empty worlds whose total sum satisfies the desired comparison.
The full listing of this algorithm is given in the appendix.

The algorithm works also for the special case where $t=1$, i.e., $\text{\sql{COUNT(*)}}\op C$ comparisons.
However, for this case we can use a more efficient enumeration algorithm, named \textsc{ValidCount}, exploiting the unit-weight sum structure of \sql{COUNT} aggregates.
For $\text{\sql{COUNT(*)}}\op C$ we only rely on the fact that the result of a \sql{COUNT} depends only on the number of occurrences of the group key values of the selected tuples.
Our algorithm recursively constructs worlds by deciding for each tuple-occurrence if it is included in or excluded from the current world.
This enumerates subsets of the desired sizes.
Finally, depending on the comparison operator, we combine the appropriate size-ranges to get a collection of non-empty worlds that satisfy the comparison.
The full listing of this algorithm is given in the appendix.

For other aggregate functions, using brute-force enumeration of non-empty worlds $W\sqsubseteq\nobreak U$ will always give the correct answer by our semantics.
When a side of the comparison mentions several aggregates of the group, as in \sql{SUM(a) > 2 * COUNT(*)}, the valid worlds are those~$W$ on which the comparison holds with both sides evaluated in~$W$. \textsc{ValidSum} and \textsc{ValidCount} do not apply, and our implementation enumerates the worlds by brute force.
\begin{toappendix}
\label{appendix:algorithms}
This appendix gives the full listings of \textsc{ValidSum}
(Algorithm~\ref{alg:sum_dp}) and \textsc{ValidCount}
(Algorithm~\ref{alg:count_enum}), then proves the results of
Section~\ref{sec:algorithms}.
\begin{algorithm}
  \small
  \caption{Enumeration of valid worlds for \sql{SUM}$(t)\op C$ at a fixed group key $\vec v$}
  \label{alg:sum_dp}

  \nlset{}\textbf{\textsc{ValidSum}}:\;

  \Input{a finite set of group-member occurrences $U \defeq U^{\preceq}_{\angsem{\hat I}{q'},\vec v}$, an arbitrary ordering $\preceq$ of $U$, a term $t$ of max-index $\le k$ such that $t(u)\in\NN$ for every $(u,\alpha)\in U$, a constant $C\in\NN$, and an operator ${\op}\in\{=,>,<,\ge,\le,\neq\}$}
  \Output{$\mathcal W =\{\,W\sqsubseteq U\mid W\neq\emptyset,\mathrm{agg}_{t,f_\Sigma}(W)\op C\,\}$}
  \BlankLine
  \Fn{\textsc{ValidSum}$(U,t,C,\op, \preceq)$}{
  \If{$\op$ is $\neq$}{
    \Return $\textsc{ValidSum}(U,t,C,<, \preceq) \cup \textsc{ValidSum}(U,t,C,>, \preceq)$\;
  }
 \BlankLine
  Let $((u_r,\alpha_r))_{r \in [|U|]}$ be the occurrences of $U$ ordered by $\preceq$\\
  $T \gets \mathrm{agg}_{t,f_\Sigma}([|U|]) = \sum_{r=1}^{|U|} t(u_r)$\;
  \If{($\op$ is $=$ \textbf{and} $C > T$) \textbf{or} ($\op$ is $>$ \textbf{and} $C \ge T$) \textbf{or} ($\op$ is $\ge$ \textbf{and} $C > T$) \textbf{or} ($\op$ is $<$ \textbf{and} $C = 0$)}{
    \Return $\emptyset$\;
  }
  \If{${\op}\in\{>,\ge\}$}{
    $J \gets T$\;
  }
  \ElseIf{$\op$ is $<$}{
    $J \gets \min(C-1,\ T)$\;
  }
  \Else{
    $J \gets \min(C,\ T)$\;
  }

  \BlankLine
  Initialize $dp[0] \leftarrow \{\emptyset\}$\;
  \For{$j \gets 1$ \KwTo $J$}{
    $dp[j] \leftarrow \emptyset$\;
  }

  \BlankLine
  $\mathrm{pref\_sum} \gets 0$\;

  \For{$i \gets 1$ \KwTo $|U|$}{
    $w \gets t(u_i)$\;
    $\mathrm{pref\_sum} \gets \mathrm{pref\_sum} + w$\;
    $j_{\max} \gets \min(J,\ \mathrm{pref\_sum})$\;

    \For{$j \gets j_{\max}, j_{\max}-1, \dots, w$}{
      $p \gets j - w$\;
      \If{$dp[p] \neq \emptyset$}{
        let $s \gets |dp[p]|$\;
        \For{$r \gets 1$ \KwTo $s$}{
          append $dp[p][r] \cup \{i\}$ to $dp[j]$\;
        }
      }
    }
  }

  \BlankLine
  \If{$\op$ is $=$}{
    $\mathcal W \leftarrow dp[C]$\;
  }
  \ElseIf{$\op$ is $>$}{
    $\mathcal W \leftarrow \bigcup_{j=C+1}^{J} dp[j]$\;
  }
  \ElseIf{$\op$ is $<$}{
    $\mathcal W \leftarrow \bigcup_{j=0}^{C-1} dp[j]$\;
  }
  \ElseIf{$\op$ is $\ge$}{
    $\mathcal W \leftarrow \bigcup_{j=C}^{J} dp[j]$\;
  }
  \Else{
    $\mathcal W \leftarrow \bigcup_{j=0}^{C} dp[j]$\;
  }

  $\mathcal W \leftarrow \mathcal W \setminus \{\emptyset\}$\;

  \Return $\mathcal W$\;
  }
\end{algorithm}
\end{toappendix}

The following theorem states that \textsc{ValidSum} and
\textsc{ValidCount} return exactly the index set of the $\oplus$-sum of
Definition~\ref{def:semantics}; the predicate provenance is then the
$\oplus$-sum of $\ann_U(W)$ over the returned worlds.
\begin{toappendix}
\begin{algorithm}
  \small
  \caption{Enumeration of valid worlds for \sql{COUNT}$(t)\op C$ at a fixed group key $\vec v$}
  \label{alg:count_enum}

  \nlset{}\textbf{\textsc{ValidCount}:}\;

  \Input{a finite set of group-member occurrences $U\defeq U^{\preceq}_{\angsem{\hat I}{q'},\vec v}$, a constant $C\in\NN$, and an operator ${\op}\in \{=,\neq,>,<,\ge,\le\} $}
  \Output{$\mathcal W =\{\,W\sqsubseteq U\mid W\neq\emptyset,|W|\op C\,\}$}
    \BlankLine
  $\mathcal W \gets \emptyset$\;

  \SetKwFunction{Combinations}{Combinations}
  \Fn{\Combinations{$i, x, W$}}{
    \If{$x = 0$}{
      append $W$ to $\mathcal W$\;
      \Return\;
    }
    \If{$i > |U|$}{
      \Return\;
    }
    \If{$|U|-i+1 < x$}{
      \Return\;
    }

    \tcp{Case 1: do not select occurrence $i$}
    \Combinations{$i+1, x, W$}\;

    \tcp{Case 2: select occurrence $i$}
    \Combinations{$i+1, x-1, W \cup \{i\}$}\;
  }

  \SetKwFunction{AddExact}{AddExact}
  \Fn{\AddExact{$x$}}{
    \If{$x\le 0$ \textbf{or} $x>|U|$}{\Return}
    \Combinations{$1, x, \emptyset$}\;
     }

  \BlankLine
  \Switch{$\op$}{
    \Case{$=$}{
      \AddExact{$C$}\;
    }
    \Case{$>$}{
      \For{$x\gets C+1$ \KwTo $|U|$}{\AddExact{$x$}}
    }
    \Case{$\ge$}{
      \For{$x\gets C$ \KwTo $|U|$}{\AddExact{$x$}}
    }
    \Case{$<$}{
      \For{$x\gets 1$ \KwTo $\min(C-1,|U|)$}{\AddExact{$x$}}
    }
    \Case{$\le$}{
      \For{$x\gets 1$ \KwTo $\min(C,|U|)$}{\AddExact{$x$}}
    }
    \Case{$\neq$}{
      \For{$x\gets 1$ \KwTo $|U|$}{
        \If{$x\neq C$}{\AddExact{$x$}}
      }
    }
  }

  \Return $\mathcal W$\;
\end{algorithm}
\end{toappendix}
\begin{theoremrep}[Correctness of the algorithms]\label{thm:algorithms-correct}
\lean{Algorithms/CountEnum}{CountEnum.countEnum_correct}
\lean{Algorithms/SumDP}{SumDP.sumDP_correct}
Let $q'$ be an aggregate-free query of arity~$k$, $\hat I$ a
$\mathbb K$-instance, $\vec v$ a group key, $\preceq$ an ordering,
$U\defeq\nobreak U^{\preceq}_{\angsem{\hat I}{q'},\vec v}$, $C\in\NN$ and
${\op}\in\{=,>,<,\ge,\le,\neq\}$. Then, on input $(U,C,\op)$ and
$(U,t,C,\op,\preceq)$ respectively, \textsc{ValidCount} and
\textsc{ValidSum} output the set
$\{\,W\sqsubseteq U\mid W\neq\emptyset,\ \mathrm{agg}_{t,f_\Sigma}(W)\op C\,\}$
of valid worlds, where $t$ is the constant term~$1$ for
\textsc{ValidCount} and, for \textsc{ValidSum}, any term of max-index
$\le k$ taking values in $\NN$ on the occurrences of~$U$; $f_\Sigma$
is the \sql{SUM} aggregate, so that $\mathrm{agg}_{1,f_\Sigma}(W)=|W|$.
\end{theoremrep}

\begin{appendixproof}
\paragraph*{Correctness of \textsc{ValidSum}}
Write
\[
  T\defeq \mathrm{agg}_{t,f_\Sigma}([|U|])=\sum_{r=1}^{|U|} t(u_r).
\]

\inlinepar{Case $\op$ is $\neq$.}
If $\op$ is $\neq$, the algorithm returns
\[
\textsc{ValidSum}(U,t,C,<, \preceq)\cup\textsc{ValidSum}(U,t,C,>, \preceq).
\]
By the correctness proved below for ${\op}\in\{<,>\}$, this equals
\[
  \{\,W\subseteq [|U|]\mid W\neq\emptyset,\sum_{r\in W}t(u_r)<C\,\}
\cup
\{\,W\subseteq [|U|]\mid W\neq\emptyset,\sum_{r\in W}t(u_r)>C\,\}.
\]
Since for every integer $x$ one has $x\neq C$ iff $x<C$ or $x>C$, the latter union is exactly
\[
  \{\,W\subseteq [|U|]\mid W\neq\emptyset,\sum_{r\in W}t(u_r)\neq C\,\}.
\]
Hence the theorem holds when $\op$ is $\neq$.

\smallskip\noindent
The remainder of the proof establishes correctness for ${\op}\in\{=,>,<,\ge,\le\}$.

\inlinepar{Early-return cases.}
The algorithm returns $\emptyset$ in the following scenarios:
\begin{itemize}
\item if $\op$ is $=$ and $C>T$;
\item if $\op$ is $>$ and $C\ge T$;
\item if $\op$ is $\ge$ and $C>T$;
\item if $\op$ is $<$ and $C=0$.
\end{itemize}
If $\op$ is $=$ and $C>T$, then every world $W\subseteq [|U|]$ satisfies
\[
  \sum_{r\in W} t(u_r)\le \sum_{r=1}^{|U|} t(u_r)=T<C,
\]
so no world has sum $C$.

If $\op$ is $>$ and $C\ge T$, then every world $W\subseteq [|U|]$ satisfies
\[
  \sum_{r\in W} t(u_r)\le T\le C,
\]
so no world has sum strictly greater than $C$.

If $\op$ is $\ge$ and $C>T$, then again every world has sum at most $T<C$, so no world has sum at least $C$.

If $\op$ is $<$ and $C=0$, since all values $t(u_r)$ lie in $\NN$, every world has nonnegative sum; hence no world has sum $<0$.

Therefore, in each early-return case, the algorithm correctly outputs $\emptyset$.

\inlinepar{DP invariant.}
We can safely assume from now on that none of the early-return cases occur.
For $0\le i\le |U|$, $j\ge 0$, define
\[
\mathcal D_i(j)\defeq
\{\,W\subseteq [i]\mid \sum_{r\in W} t(u_r)=j\,\}.
\]
We claim that after completing the outer-loop iteration with index $i$,
for every $j\in\{0,\dots,J\}$,
\begin{equation}\label{eq:sumdp-invariant}
dp[j]=\mathcal D_i(j).
\end{equation}

We prove this by induction on $i$.

\smallskip\noindent\emph{Base case $i=0$.}
The algorithm initializes
\[
dp[0]=\{\emptyset\}
\qquad\text{and}\qquad
dp[j]=\emptyset \text{ for } 1\le j\le J.
\]
We see,
\[
\mathcal D_0(0)=\{\emptyset\}
\qquad\text{and}\qquad
\mathcal D_0(j)=\emptyset \text{ for } j\ge 1.
\]
Hence \eqref{eq:sumdp-invariant} holds for $i=0$.

\smallskip\noindent\emph{Inductive step.}
Assume \eqref{eq:sumdp-invariant} holds after iteration $i-1$, for some $1\le i\le |U|$.
Fix $j\in\{0,\dots,J\}$.
We show that after iteration $i$, one has $dp[j]=\mathcal D_i(j)$.

We now establish two cases.

\smallskip\noindent\underline{Case 1: $j<t(u_i)$.}
The inner loop does not update $dp[j]$, since the loop runs only for indices
\[
  j=j_{\max},j_{\max}-1,\dots,t(u_i).
\]
Hence after iteration $i$, the value of $dp[j]$ is unchanged, so by the induction hypothesis
\[
dp[j]=\mathcal D_{i-1}(j).
\]
Further, no subset of $[i]$ that contains $i$ can have a sum $j$, because such a subset has a sum at least
\[
  t(u_i)>j,
\]
all summands being nonnegative. Therefore every world in $\mathcal D_i(j)$ must avoid $i$, so
\[
\mathcal D_i(j)=\mathcal D_{i-1}(j).
\]
Consequently,
\[
dp[j]=\mathcal D_i(j).
\]

\smallskip\noindent\underline{Case 2: $j\ge t(u_i)$.}
A subset $W\subseteq [i]$ satisfies
\[
  \sum_{r\in W} t(u_r)=j
\]
iff exactly one of the following two disjoint possibilities hold:
\begin{itemize}
  \item $i\notin W$, in which case $W\subseteq [i-1]$ and
\[
W\in \mathcal D_{i-1}(j);
\]
\item $i\in W$, in which case
\[
W=W'\cup\{i\}
\]
for a unique subset $W'\subseteq [i-1]$ satisfying
\[
  \sum_{r\in W'} t(u_r)=j-t(u_i),
\]
that is,
\[
  W'\in \mathcal D_{i-1}(j-t(u_i)).
\]
\end{itemize}

By the induction hypothesis, before processing index $i$ we have
\[
dp[j]=\mathcal D_{i-1}(j)
\qquad\text{and}\qquad
dp[j-t(u_i)]=\mathcal D_{i-1}(j-t(u_i)).
\]
During the update for this fixed $j$, the algorithm sets
\[
  p\gets j-t(u_i),
\qquad
s\gets |dp[p]|.
\]
It then appends
\[
dp[p][r]\cup\{i\}
\qquad (1\le r\le s)
\]
to $dp[j]$.

Thus the worlds added to $dp[j]$ are exactly those obtained by adding $i$ to the worlds that were already present in $dp[p]$ before any append operation at this $j$.
Since, by the I.H, those pre-existing worlds are exactly the elements of $\mathcal D_{i-1}(j-t(u_i))$, the algorithm adds exactly the worlds in
\[
  \{\,W'\cup\{i\}\mid W'\in \mathcal D_{i-1}(j-t(u_i))\,\}.
\]

$s=|dp[p]|$ is essential when $t(u_i)=0$, for then $p=j$.
In that situation the algorithm appends new worlds to $dp[j]$ while scanning from $dp[j]$, but it scans only the first $s$ entries, i.e., exactly those worlds that were present before the append operations started.
Hence each world in $\mathcal D_{i-1}(j)$ contributes exactly one extended world containing $i$, and none of the newly appended worlds is reused again during the same update.

Therefore, after iteration $i$, the content of $dp[j]$ is exactly:
\begin{itemize}
\item all worlds in $\mathcal D_{i-1}(j)$, corresponding to the case $i\notin W$;
\item together with all worlds of the form $W'\cup\{i\}$ for
  $W'\in\mathcal D_{i-1}(j-t(u_i))$, corresponding to the case $i\in W$.
\end{itemize}
By the disjoint case analysis above, this is precisely $\mathcal D_i(j)$.
Thus
\[
dp[j]=\mathcal D_i(j).
\]

Since both cases have been established, \eqref{eq:sumdp-invariant} holds after iteration $i$.
This completes the proof by induction.

Hence, after the outer loop terminates, for every $j\in\{0,\dots,J\}$,
\[
dp[j]
=
\{\,W\subseteq [|U|]\mid \sum_{r\in W}t(u_r)=j\,\}.
\]

\inlinepar{Selecting satisfying worlds.}
The algorithm forms $\mathcal W$ as a union of those $dp[j]$ whose index $j$ satisfies the comparison with $C$:
\begin{itemize}
\item if $\op$ is $=$, it takes $dp[C]$;
\item if $\op$ is $>$, it takes $\bigcup_{j=C+1}^{J} dp[j]$;
\item if $\op$ is $<$, it takes $\bigcup_{j=0}^{C-1} dp[j]$;
\item if $\op$ is $\ge$, it takes $\bigcup_{j=C}^{J} dp[j]$;
\item if $\op$ is $\le$, it takes $\bigcup_{j=0}^{C} dp[j]$.
\end{itemize}
By the characterization of $dp[j]$ proved above, each such union is exactly the set of worlds
$W\subseteq [|U|]$ whose sum satisfies
\[
  \sum_{r\in W} t(u_r)\op C.
\]
The only possible extra world is $\emptyset$, which belongs to $dp[0]$ and may therefore appear whenever the selected range contains $0$.
The algorithm omits it by setting
\[
\mathcal W\gets \mathcal W\setminus \{\emptyset\}.
\]
It follows that the final output is exactly
\[
  \{\,W\subseteq [|U|]\mid W\neq\emptyset,\sum_{r\in W} t(u_r)\op C\,\}.
\]

This establishes the correctness.

\paragraph*{Correctness of \textsc{ValidCount}}
The computation for \sql{COUNT} depends only on the occurrences in $U$.

\inlinepar{Correctness of \textsc{Combinations}.}
For integers $i$ and $x$, and for a $W\subseteq [|U|]$, consider a call
\[
\textsc{Combinations}(i,x,W).
\]
We claim that this call appends to $\mathcal W$ exactly the sets
\[
  W'\subseteq [|U|]
\]
such that:
\begin{itemize}
\item $W\subseteq W'$;
\item every element of $W'\setminus W$ belongs to $\{i,\dots,|U|\}$;
\item $|W'\setminus W|=x$.
\end{itemize}
Equivalently, it appends exactly the worlds of the form
\[
W\cup S
\quad\text{with}\quad
S\subseteq \{i,\dots,|U|\}
\text{ and } |S|=x.
\]

We prove this by induction on the pair $(|U|-i+1,x)$.

\smallskip\noindent\emph{Base cases.}
If $x=0$, it appends $W$ and returns. This is correct because the only way to extend $W$ by choosing exactly $0$ further occurrences is to keep $W$.

If $i>|U|$, there are no remaining occurrences to choose from. Hence no extension of $W$ by a positive number of further occurrences exists, so returning without appending anything is correct.

If $|U|-i+1<x$, then fewer than $x$ occurrences remain available. Thus it is impossible to choose exactly $x$ more occurrences, so returning without appending is again correct.

\smallskip\noindent\emph{Inductive step.}
Assume now that $x>0$, $i\le |U|$, and $|U|-i+1\ge x$.
Any extension of $W$ by exactly $x$ occurrences chosen from
\[
\{i,\dots,|U|\}
\]
falls into one of two disjoint cases:
\begin{itemize}
\item it does \emph{not} contain $i$;
\item it \emph{does} contain $i$.
\end{itemize}

In the first case, the required worlds are exactly those produced by the recursive call
\[
\textsc{Combinations}(i+1,x,W),
\]
since one must still choose all $x$ required occurrences from the suffix starting at $i+1$.

In the second case, the desired worlds are exactly those produced by
\[
\textsc{Combinations}(i+1,x-1,W\cup\{i\}),
\]
since after selecting $i$, one must select exactly $x-1$ further occurrences from the suffix starting at $i+1$.

By the I.H, the first recursive call appends exactly the worlds of the first kind, and the second appends exactly the worlds of the second kind. Since the two cases are disjoint and exhaustive, \textsc{Combinations} appends exactly all desired worlds, each exactly once.

This proves the claim.

\inlinepar{Correctness of \textsc{AddExact}.}
Consider a call \textsc{AddExact}$(x)$.

If $x\le 0$, the procedure returns immediately. This is correct because the output requires $W\neq\emptyset$, so no world of size $0$ should be output, and no world has negative cardinality.

If $x>|U|$, the procedure also returns immediately. This is correct because no subset of $[|U|]$ can have cardinality greater than $|U|$.

If $1\le x\le |U|$, the procedure calls
\[
\textsc{Combinations}(1,x,\emptyset).
\]
By the claim proved above, this appends exactly all sets
\[
  W\subseteq [|U|]
\quad\text{such that}\quad
|W|=x.
\]
Hence \textsc{AddExact}$(x)$ appends exactly the non-empty worlds of cardinality $x$.

\inlinepar{Selecting satisfying worlds.}
We now inspect the final case distinction on $\op$.

\begin{itemize}
  \item If $\op$ is $=$, the algorithm executes \textsc{AddExact}$(C)$, so by the correctness of \textsc{AddExact} it outputs exactly the non-empty worlds $W\subseteq [|U|]$ such that $|W|=C$.

\item If $\op$ is $>$, the algorithm executes \textsc{AddExact}$(x)$ for every
\[
x\in\{C+1,\dots,|U|\}.
\]
Therefore it returns exactly the non-empty worlds whose cardinality is strictly greater than $C$.

\item If $\op$ is $\ge$, the algorithm executes \textsc{AddExact}$(x)$ for every
\[
x\in\{C,\dots,|U|\}.
\]
Since \textsc{AddExact} ignores values $x\le 0$, this outputs exactly the non-empty worlds whose cardinality is at least $C$.

\item If $\op$ is $<$, the algorithm executes \textsc{AddExact}$(x)$ for every
\[
x\in\{1,\dots,\min(C-1,|U|)\}.
\]
Thus it returns exactly the non-empty worlds whose cardinality is strictly less than $C$.

\item If $\op$ is $\le$, the algorithm executes \textsc{AddExact}$(x)$ for every
\[
x\in\{1,\dots,\min(C,|U|)\}.
\]
Thus it returns exactly the non-empty worlds whose cardinality is at most $C$.

\item If $\op$ is $\neq$, the algorithm executes \textsc{AddExact}$(x)$ for every
\[
x\in\{1,\dots,|U|\}\setminus\{C\}.
\]
Hence it returns exactly the non-empty worlds whose cardinality is different from $C$.
\end{itemize}

In every case, the resulting  $\mathcal W$ is exactly
\[
  \{\,W\subseteq [|U|]\mid W\neq\emptyset,\ |W|\op C\,\}.
\]

This establishes correctness.
\end{appendixproof}

We identify certain aggregate-comparison predicates for which the valid worlds can be constructed in polynomial time:
\begin{propositionrep}
  \label{prop:algorithms}
  The provenance of \sql{HAVING} queries of the following forms can be computed
  in polynomial time in data complexity:
  \begin{itemize}
    \item $\text{\sql{MIN(a)}}\op c$,
    $\text{\sql{MAX(a)}}\op c$,
$\text{\fakesql{PICKFIRST}\texttt{(a)}}\op c$
      with
      ${\op}\in\{<,\leq,=,\geq,>,\neq\}$ in absorptive m-semirings in which
      $\otimes$ distributes over $\ominus$;
    \item \sql{COUNT(*) >= k} and \sql{COUNT(*) > k} for $k$ a fixed constant, in absorptive
      m-semirings;
\item \sql{SUM(a) >= }$c$ and \sql{SUM(a) > }$c$ with $c\in\NN$, for an $\NN$-valued
  attribute \texttt{a} in which every nonzero value $a$ satisfies
  $\frac ca\leq k$ where $k\geq 0$ is a fixed constant, in absorptive
      m-semirings;
    \item \sql{COUNT(*) = k}, \sql{COUNT(*) <= k}, and \sql{COUNT(*) < k} for $k$ a fixed
      constant, in arbitrary m-semirings;
    \item \sql{SUM(a) = }$c$, \sql{SUM(a) <= }$c$, and \sql{SUM(a) < }$c$
      with $c\in\NN$, for an $\NN$-valued attribute \texttt{a} whose
      values are all nonzero and satisfy $\frac ca\leq k$ where $k\geq 0$
      is a fixed constant, in arbitrary m-semirings.
  \end{itemize}
\end{propositionrep}
\begin{proof}
Fix an arbitrary $\mathbb K$-instance $\hat I$ and a group key $\vec v$, and write $U\defeq U^{\preceq}_{\angsem{\hat I}{q'},\vec v}$ for the sequence of group-member occurrences at~$\vec v$.
For each case below, we show that the provenance for $\vec v$ can be computed in time polynomial in $|U|$. Since the queries are fixed, this implies polynomial time in data complexity.
We recall that the provenance of the \sql{HAVING} condition (equiv.\ selection predicates with aggregate comparisons) at group key $\vec v$ is obtained by enumerating the valid non-empty worlds
and then forming the corresponding semiring sum of the $\mathbb K$-annotations associated with these worlds using the possible-world semantics. Thus it suffices to show that, in each of the stated cases, the relevant list of worlds can be computed in polynomial time.

\inlinepar{Case 1: $\text{\sql{MIN(a)}}\op c$, $\text{\sql{MAX(a)}}\op c$, and $\text{\fakesql{PICKFIRST}\texttt{(a)}}\op c$.}
Write $\alpha_i$ and $a_i$ for the annotation and the value of
attribute~\texttt{a} of the $i$-th occurrence of~$U$, occurrences being
indexed in the order~$\preceq$, and
$L\defeq\bigoplus_{a_i<c}\alpha_i$, $L'\defeq\bigoplus_{a_i\le c}\alpha_i$,
$G\defeq\bigoplus_{a_i\ge c}\alpha_i$, $G'\defeq\bigoplus_{a_i>c}\alpha_i$,
$E\defeq\bigoplus_{a_i=c}\alpha_i$. In an absorptive m-semiring in which
$\otimes$ distributes over $\ominus$, the possible-world sum for
$\text{\sql{MIN(a)}}\op c$ has the closed form
\lean{HavingMinMax}{Having.minScan_correct}
\[
\begin{array}{r@{\ }l@{\qquad}r@{\ }l@{\qquad}r@{\ }l}
{<}c: & L, & {\le}c: & L', & {\ne}c: & L\oplus(\mathbb 1\ominus L')\otimes G',\\
{\ge}c: & (\mathbb 1\ominus L)\otimes G, & {>}c: & (\mathbb 1\ominus L')\otimes G', & {=}c: & (\mathbb 1\ominus L)\otimes E;
\end{array}
\]
for instance, $\text{\sql{MIN(a)}}<c$ holds in a world exactly when some
occurrence with $a_i<c$ is present, and absorptivity reduces the sum over
these worlds to the sum over the singletons, whereas $\text{\sql{MIN(a)}}\ge c$
additionally requires that no occurrence with $a_i<c$ be present, whence
the factor $\mathbb 1\ominus L$. The closed form for $\text{\sql{MAX(a)}}\op c$
is the mirror image, exchanging $<$ and~$>$
\lean{HavingMinMax}{Having.maxScan_correct}, and for
$\text{\fakesql{PICKFIRST}\texttt{(a)}}\op c$ (the non-commutative aggregate of
Example~\ref{ex:aggregates}) it is
$\bigoplus_{i:\,a_i\op c}(\mathbb 1\ominus\bigoplus_{j<i}\alpha_j)\otimes\alpha_i$,
one term per satisfying occurrence that comes first
\lean{HavingMinMax}{Having.firstScan_correct}. Each is computed by a
single scan of~$U$ using $O(|U|)$ semiring operations. For the four
forms with a monus factor, both hypotheses are used: absorptivity fails in the tropical m-semiring over
$\mathbb R\cup\{\infty\}$ (idempotent and distributive), where the scan
value differs from the possible-world provenance already for
$\text{\sql{MIN(a)}}\ge c$ on a two-occurrence group
\lean{Semirings/Tropical}{MinTropicalR.minScan_ne_prov}; distributivity of
$\otimes$ over $\ominus$ identifies the world annotation $\ann_U(W)$ with
the expanded form $T_U(W)$ of Section~\ref{sec:results}
\lean{HavingSemantics}{Having.worldAnn_eq_T}, on which the closed forms
are established. The two monus-free forms, ${}<c$ and ${}\le c$, need
absorptivity alone: their valid worlds are those meeting a set of
occurrences, a family closed under supersets whose minimal elements
are the singletons, so
Proposition~\ref{prop:monus-cancellation}(iii) applies
\lean{Having}{Having.sum_ann_meet}.

\inlinepar{Case 2: \sql{COUNT(*) >= k} and \sql{COUNT(*) > k} for fixed $k$, in absorptive m-semirings.}
By correctness of Algorithm~\ref{alg:count_enum}
(Theorem~\ref{thm:algorithms-correct}), the valid worlds are exactly the non-empty worlds
$W\sqsubseteq U$ such that $|W|\ge k\text{ or }|W|>k$, respectively.
  If such a $W$ satisfies the condition, then every superset of $W$ also satisfies it; equivalently these $W$ form an upward-closed family of sets. The m-semiring $\mathbb K$ being absorptive, the possible-world provenance of an upward-closed family of worlds (the $\oplus$-sum of their world annotations) is equal to the $\oplus$-sum of the monomials of its minimal worlds \lean{Having}{Having.upward_closed_collapse}; for the family of worlds of size at least $k$ this collapse is exactly the identity $F_k(U)=S_k(U)$ of Section~\ref{sec:results} \lean{Having}{Having.F_eq_S}. Here, the minimal satisfying worlds are:
\begin{itemize}
\item exactly the worlds of size $k$ for \sql{COUNT(*) >= k};
\item exactly the worlds of size $k+1$ for \sql{COUNT(*) > k}.
\end{itemize}
Since $k$ is fixed, the number of such worlds is $\binom{|U|}{k}\text{ or }\binom{|U|}{k+1}$,
which is polynomial in $|U|$. Therefore the provenance can be computed in polynomial time.

\inlinepar{Case 3: \sql{SUM(a) >= c} and \sql{SUM(a) > c} with $c\ge 0$, in absorptive m-semirings.}
By correctness of Algorithm~\ref{alg:sum_dp}
(Theorem~\ref{thm:algorithms-correct}), the valid worlds are exactly the non-empty worlds
$W\sqsubseteq U$ such that $\sum_{r\in W} t(u_r)\ge c$ (resp.\ $>c$).
Since all values are nonnegative, these families are upward-closed under inclusion. Hence, in an absorptive m-semiring, the provenance collapses to the $\oplus$-sum of the monomials of the minimal satisfying worlds
\lean{Having}{Having.sum_ge_collapse} \lean{Having}{Having.sum_gt_collapse}.
Let $W$ be a minimal satisfying world. A zero-valued occurrence never belongs to $W$: removing it leaves the sum unchanged, contradicting minimality. Hence every value $a$ of an occurrence of $W$ is nonzero and satisfies $\frac ca\le k$, i.e., $c\le k\cdot a$, so any $k$ selected occurrences of $W$ already contribute a sum at least $c$. For \sql{SUM(a) >= c}, this gives $|W|\le k$: if $|W|>k$, then any $k$-element subset of $W$ already has sum $\ge c$, contradicting minimality \lean{Having}{Having.minimal_card_le_of_sum_ge}. For the strict comparison \sql{SUM(a) > c}, a $k$-element subset only guarantees a sum $\ge c$, and one further (nonzero) occurrence makes the inequality strict -- covering in particular the boundary case $c=0$ -- so minimality gives $|W|\le k+1$ \lean{Having}{Having.minimal_card_le_of_sum_gt}.

Therefore the number of minimal satisfying worlds is at most
$\sum_{i\le k}\binom{|U|}{i}$ in the first case
\lean{Having}{Having.card_minimal_sum_ge_le} and
$\sum_{i\le k+1}\binom{|U|}{i}$ in the second; both bounds are polynomial
in $|U|$ for fixed~$k$, so the provenance can be computed in polynomial
time.

\inlinepar{Case 4: \sql{COUNT(*) = k}, \sql{COUNT(*) <= k}, and \sql{COUNT(*) < k} for fixed $k$, in arbitrary m-semirings.}
By correctness of Algorithm~\ref{alg:count_enum}, the valid worlds are exactly:
\begin{itemize}
\item the worlds of size exactly $k$ for \sql{COUNT(*) = k};
\item the worlds of sizes $1,\dots,k$ for \sql{COUNT(*) <= k};
\item the worlds of sizes $1,\dots,k-1$ for \sql{COUNT(*) < k}.
\end{itemize}
Since $k$ is fixed, the number of such worlds is bounded by
  $\binom{|U|}{k},\, \sum_{i=1}^{k}\binom{|U|}{i},\text{ and } \sum_{i=1}^{k-1}\binom{|U|}{i}$, respectively \lean{Having}{Having.card_powerset_filter_card_le}. Each of these bounds is polynomial in $|U|$. Therefore the provenance can be computed in polynomial time in arbitrary m-semirings.

\inlinepar{Case 5: \sql{SUM(a) = }$c$, \sql{SUM(a) <= }$c$, and \sql{SUM(a) < }$c$ with all values nonzero and $\frac ca\le k$, in arbitrary m-semirings.}
By correctness of Algorithm~\ref{alg:sum_dp}
(Theorem~\ref{thm:algorithms-correct}), the valid worlds are exactly the
non-empty worlds $W\sqsubseteq U$ whose sum is equal to, at most, or less
than~$c$. If $c=0$ there is none, all values being nonzero. Otherwise, a
valid world $W$ satisfies $|W|\cdot\frac ck\le\sum_{r\in W}t(u_r)\le c$,
since every value is at least $\frac ck$, hence $|W|\le k$. The valid
worlds are therefore among the $\sum_{i=1}^{k}\binom{|U|}{i}$ worlds of
size at most~$k$ \lean{Having}{Having.card_powerset_filter_card_le},
polynomially many for fixed~$k$, and \textsc{ValidSum}, whose table is
bounded by~$c$, enumerates exactly them. No assumption on the m-semiring
is needed.

\smallskip\noindent
With these 5 cases, provenance can be computed in polynomial time for all the \sql{HAVING} queries described in the proposition.
\end{proof}
\begin{proofsketch}
Fix the group key $\vec v$ and write $U\defeq U^{\preceq}_{\angsem{\hat I}{q'},\vec v}$. The provenance of the \sql{HAVING} condition is the semiring sum of the annotations of the valid non-empty worlds, so it suffices to enumerate the relevant worlds in time polynomial in $|U|$.
For \sql{MIN(a)}, \sql{MAX(a)} and \fakesql{PICKFIRST}\texttt{(a)} (Case~1), the validity of a world is decided occurrence by occurrence, so in an absorptive m-semiring with $\otimes$ distributive over $\ominus$ the provenance is obtained by a single scan \lean{HavingMinMax}{Having.minScan_correct}.
For $\text{\sql{COUNT(*)}}\ge k$ or ${}>k$ (Case~2) and $\text{\sql{SUM(a)}}\ge c$ or ${}>c$ (Case~3), the valid worlds form an upward-closed family, whose provenance, in an \emph{absorptive} m-semiring, collapses to the $\oplus$-sum of the monomials of its minimal worlds \lean{Having}{Having.upward_closed_collapse}. These minimal worlds have bounded cardinality: $k$, or $k+1$ for the strict variants (using the hypothesis $\frac ca\le k$ on nonzero values for \sql{SUM}); so there are only $O(|U|^{k+1})$ of them.
Finally, for $\text{\sql{COUNT(*)}}=k$, $\le k$ or $<k$ (Case~4), the valid worlds are exactly those of cardinality at most $k$, again polynomially many, so no absorptivity is needed; likewise for $\text{\sql{SUM(a)}}=c$, $\le c$ or $<c$ when all values are at least $\frac ck$ (Case~5).
\end{proofsketch}

\begin{example}\label{ex:prop-simplification}
Consider as in Example~\ref{ex:semantics} a group of three occurrences,
now with values $a(u_1)=3$ and $a(u_2)=a(u_3)=2$, and the condition
$\text{\sql{SUM(a)}}\ge 5$. The valid worlds are $\{1,2\}$, $\{1,3\}$
and $\{1,2,3\}$, with possible-world provenance
$(x_1\land x_2\land\lnot x_3)\lor(x_1\land x_3\land\lnot x_2)\lor(x_1\land
x_2\land x_3)$; the family is upward-closed and $\mathbb B[X]$
absorptive, so this equals the $\oplus$-sum over the minimal valid
worlds $\{1,2\}$ and $\{1,3\}$, namely $x_1\land(x_2\lor x_3)$.
Proposition~\ref{prop:algorithms} rests on this collapse: the minimal
worlds suffice, polynomially many for a fixed threshold.
\end{example}

The tractability established by Proposition~\ref{prop:algorithms} is
\emph{symbolic}: it computes the provenance of a \sql{HAVING} comparison in an
arbitrary (absorptive) m-semiring by enumerating the relevant valid worlds. A
complementary form of tractability arises in the \emph{probabilistic} setting. By
Proposition~\ref{prop:pqe-equiv}, evaluating a \sql{HAVING} query over a
tuple-independent instance amounts to computing the probability of its
$\mathbb{B}[X]$-provenance under the possible-world semantics. When the
contributors of a group are independent, this probability need not be obtained
by enumerating worlds, whose number is in general exponential. As Ré and
Suciu~\cite{re2009trichotomy} showed in the base case of their safe plans,
the distribution of the aggregate over the group can be assembled by a
convolution of marginal distributions, yielding the answer in polynomial
time for several aggregates and \emph{all} comparison operators. In our
terms:

\begin{propositionrep}[after~\cite{re2009trichotomy}]
\label{prop:agg-cmp-poly-prob}
Let $\hat I$ be a $\mathbb{B}[X]$-instance whose variables carry
independent probabilities, $\vec v$ a group key and
$U\defeq U^{\preceq}_{\angsem{\hat I}{q'},\vec v}=((u_i,\alpha_i))_{1\le i\le N}$
its occurrences; assume the contributors \emph{independent}, i.e., the
$\alpha_i$ have pairwise disjoint sets of variables, disjoint from those
of every other annotation in $\angsem{\hat I}{q}$, and write
$p_i\defeq\Pr(\alpha_i)$. For every comparison operator $\op$ and
threshold $C\in\NN$, the probability of the aggregate comparison on the
group under the possible-world semantics is computable in time
$O\bigl(N\cdot\min(C,N-C)\bigr)$ for $\text{\sql{COUNT(*)}}\op C$, in
time $O(N\cdot\min(C,S))$ for $\text{\sql{SUM}}(t)\op C$ with
$t(u_i)\in\NN$ and $S\defeq\sum_{i=1}^N t(u_i)$, and in time $O(N)$ for
$\text{\sql{MIN}}(t)\op C$ and $\text{\sql{MAX}}(t)\op C$ over any
totally ordered value domain.
\end{propositionrep}

\begin{proof}
By the independence hypothesis and the equivalence between
probabilistic query evaluation and possible-world summation
(Proposition~\ref{prop:pqe-equiv}), the probability of the comparison
atom factors through the indicators
$X_i\defeq[(u_i,\alpha_i)\text{ present in the world}]$, which form a
family of \emph{independent} Bernoullis with $\Pr[X_i=1]=p_i$.
We exhibit the three procedures.

\inlinepar{(1) \sql{COUNT(*)}\,$\op\,C$.}
The aggregate value on a world $W$ is $|W|=\sum_{i=1}^N X_i$, a
Poisson-binomial variable $B$.  For each $\op$, the satisfaction set
$\{j\in\{0,\dots,N\}\mid j\op C\}$ is an interval (possibly truncated
by $0$, or with the singleton $\{C\}$ removed in the $\neq$ case), so
the answer reduces to two CDF queries among $\Pr[B\le T]$,
$\Pr[B\ge T]$, $\Pr[B=T]$, $\Pr[B=0]$.

Let $\rho_n(j)\defeq\Pr\bigl[\sum_{i\le n}X_i=j\bigr]$.
The classical convolution identity
\lean{HavingProbability}{HavingProbability.countMass_insert_succ}
\[
\rho_n(j)=(1-p_n)\,\rho_{n-1}(j)+p_n\,\rho_{n-1}(j-1)
\]
yields a 1-D rolling DP that fills the row $\rho_N(\cdot)$ up to any
prescribed index $J\le N$ in time $O(N\cdot J)$.
Computing the upper tail $\Pr[B\ge C]$ directly costs
$O(N\cdot(N-C+1))$; computing the lower tail costs $O(N\cdot C)$.
Using $\Pr[B\ge C]=\Pr[B'\le N-C]$ for $B'=\sum_i(1-X_i)$, one can
always evaluate the tail on whichever side of $C$ is shorter, giving
the announced bound $O\bigl(N\cdot\min(C,N-C)\bigr)$.
The empty-world mass $\Pr[B=0]=\prod_i(1-p_i)$ is obtained as a
by-product (linear time) and subtracted in the branches whose
satisfaction interval contains $0$ (i.e., $\le$, $<$, $\neq$, and $\ge$
with $C\le0$).

\inlinepar{(2) \sql{SUM}$(t)\,\op\,C$ with integer
weights.}
The aggregate value on a world $W$ is
$\Sigma\defeq\sum_{i:X_i=1}t(u_i)$, a weighted Poisson-binomial
variable.  The same convolution generalizes
\lean{HavingProbability}{HavingProbability.sumMass_insert_of_le}:
\[
\sigma_n(s)=(1-p_n)\,\sigma_{n-1}(s)+p_n\,\sigma_{n-1}\bigl(s-t(u_n)\bigr),
\qquad
0\le s\le S,
\]
where $\sigma_n(s)\defeq\Pr\bigl[\sum_{i\le n}t(u_i)X_i=s\bigr]$.
Only the indices $s\le\min(C,S)$ are needed, the mass of $\Sigma\ge C$
being obtained by complementation (this is the saturation at~$C$ of the
semiring $S_{C+1}$ of~\cite{re2009trichotomy}), so filling
$\sigma_N(\cdot)$ on $[0,\min(C,S)]$ takes $O(N\cdot\min(C,S))$ time and
gives every $\Pr[\Sigma\op C]$ by a single summation; the empty-world
adjustment is the same $\prod_i(1-p_i)$ as before.

\inlinepar{(3) \sql{MIN}$(t)\,\op\,C$ and
\sql{MAX}$(t)\,\op\,C$.}
Independence makes the events
``$\max_{i:X_i=1}t(u_i)\le C$''
\lean{HavingProbability}{HavingProbability.funcProb_maxLeOnNonempty}
and ``$\min_{i:X_i=1}t(u_i)\ge C$''
\lean{HavingProbability}{HavingProbability.funcProb_minGeOnNonempty}
\emph{factor over $i$}, since each surviving $X_i=1$ contributes its
threshold check independently:
\begin{align*}
  \Pr[\text{\sql{MAX}}\le C\text{ on a nonempty world}]
 &=\Bigl(\textstyle\prod_{i:t(u_i)>C}(1-p_i)\Bigr)
   \,\bigl(1-\textstyle\prod_{i:t(u_i)\le C}(1-p_i)\bigr),\\
   \Pr[\text{\sql{MIN}}\ge C\text{ on a nonempty world}]
 &=\Bigl(\textstyle\prod_{i:t(u_i)<C}(1-p_i)\Bigr)
   \,\bigl(1-\textstyle\prod_{i:t(u_i)\ge C}(1-p_i)\bigr),
\end{align*}
each computable by a single scan in $O(N)$.
The remaining operators are obtained by complementation
(``$\text{\sql{MAX}}>C$'' is the negation of ``$\text{\sql{MAX}}\le C$'' on the nonempty
worlds, etc.) and the equality cases by intersecting two such tails;
$\neq$ subtracts the equality from the nonempty mass
$1-\prod_i(1-p_i)$.  All these algebraic combinations stay $O(N)$.
\end{proof}

\begin{toappendix}
The following example illustrates the convolution used in the proof of
Proposition~\ref{prop:agg-cmp-poly-prob}.
\begin{example}\label{ex:poisson-binomial}
Suppose the three contributors of Example~\ref{ex:semantics} are
independent with marginals~$p_1=\tfrac12$, $p_2=\tfrac14$ and
$p_3=\tfrac13$, and consider $\text{\sql{COUNT(*)}}\ge 2$. Writing $\rho_n(j)$
for the probability that exactly $j$ of the first $n$ contributors are
present, so that $\rho_0(0)=1$ (and~$\rho_0(j)=0$ for~$j>0$), the
recurrence $\rho_n(j)=(1-p_n)\,\rho_{n-1}(j)+p_n\,\rho_{n-1}(j-1)$ fills,
row by row, the table
\begin{center}
\begin{tabular}{ccccc}
\toprule
 & $j=0$ & $j=1$ & $j=2$ & $j=3$\\
\midrule
$\rho_0$ & $1$ & & & \\
$\rho_1$ & $\tfrac12$ & $\tfrac12$ & & \\
$\rho_2$ & $\tfrac38$ & $\tfrac12$ & $\tfrac18$ & \\
$\rho_3$ & $\tfrac14$ & $\tfrac{11}{24}$ & $\tfrac14$ & $\tfrac1{24}$\\
\bottomrule
\end{tabular}
\end{center}
Hence
$\Pr[\text{\sql{COUNT(*)}}\ge 2]=\rho_3(2)+\rho_3(3)=\tfrac14+\tfrac1{24}=\tfrac{7}{24}$,
the probability of the Boolean provenance ``at least two of
$x_1,x_2,x_3$'' obtained in Example~\ref{ex:semantics}. This
Poisson-binomial CDF is exactly the quantity our implementation evaluates
for probabilistic \sql{COUNT} comparisons
(Section~\ref{sec:implementation_experiments}).
\end{example}
\end{toappendix}

\section{Implementation and Experiments}
\label{sec:implementation_experiments}
We implemented our \sql{HAVING} semantics within ProvSQL, in C++:
\textsc{ValidCount}, \textsc{ValidSum}, the closed forms of
Proposition~\ref{prop:algorithms} for \sql{MIN}, \sql{MAX} and
\fakesql{PICKFIRST}, and brute-force enumeration otherwise, with the
prunings that
Proposition~\ref{prop:algorithms} justifies: in an absorptive m-semiring,
only the minimal valid worlds of an upward-closed comparison are
generated, and a \emph{range check} on the extremal values the aggregate
can take answers directly the comparisons that no world can satisfy.
The monus cancellation behind the monotone case of
Section~\ref{sec:results} adds two shortcuts: the monus-free closed
forms ($\text{\sql{MIN(a)}}<c$, ${}\leq c$, and their \sql{MAX} mirrors)
serve in every absorptive m-semiring, distributive or not, and in an
idempotent non-absorptive one such as $\mathrm{Why}[X]$ the valid worlds
of a monotone comparison are enumerated in full but annotated by their
present annotations alone. The
table of \textsc{ValidSum}'s dynamic programming is bounded by the
comparison constant for nonnegative integers (negative values and tables
beyond $10^7$ entries fall back to enumeration); the algorithms of
Proposition~\ref{prop:agg-cmp-poly-prob} are implemented for probability
evaluation; and standard (m-)semiring identities (multiplication
by~$\mathbb 1$, addition of~$\mathbb 0$) are folded when constructing
provenance expressions.

\begin{figure}[t]
    \centering
    \includegraphics[width=0.9\textwidth]{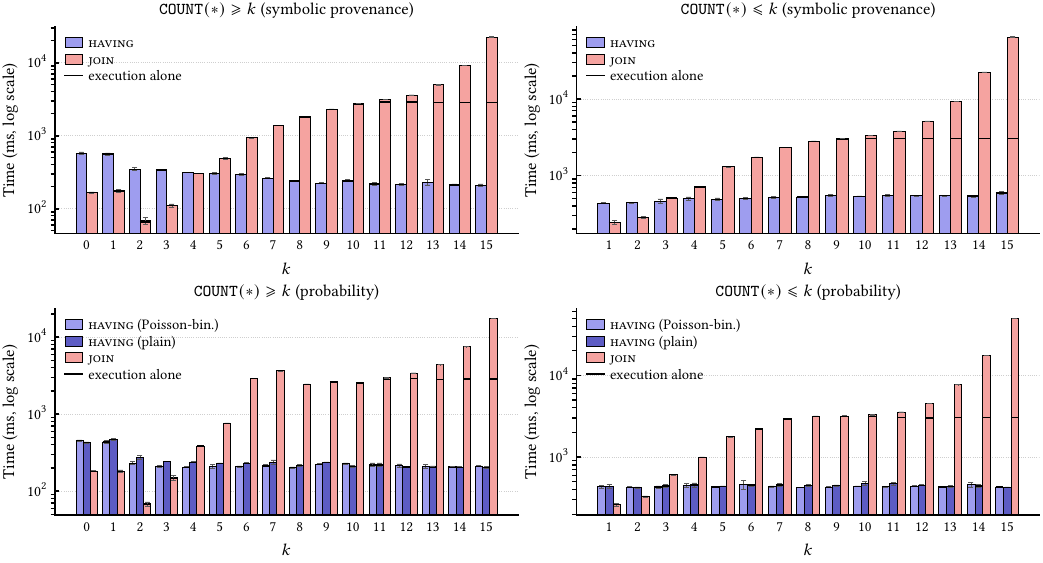}
    \caption{Planning and execution time of the \sql{HAVING} queries
    (\sql{COUNT(*) >= k} on the left, \sql{COUNT(*) <= k} on the right)
    under our semantics and under the self-join rewriting, as a function
    of~$k$: symbolic provenance computation (top) and probability
    evaluation (bottom, with and without the Poisson-binomial optimization
    of Proposition~\ref{prop:agg-cmp-poly-prob}). The height of a bar is
    the planning plus execution time; the black tick marks the execution
    time alone.}
    \label{fig:bench}
\end{figure}

We validate the practical feasibility of this implementation on a
probabilistic image-retrieval database. Following the method
of~\cite{yunus2025using}, the database is built from the objects that
YOLOv5~\cite{ultralytics2021yolov5} detects on the images of a dataset
derived from COCO 2017~\cite{lin2014microsoft}: each detected object is
one tuple of a relation \sql{dataset}, with a unique identifier
\sql{id}, the image identifier \sql{img}, the object class \sql{obj},
and a probability derived from the detector confidence score, treated as
independent probabilistic facts -- $239{,}377$ images and $1{,}570{,}565$
object annotations. We compare, for our possible-world semantics and for
the equivalent self-\sql{JOIN} rewriting (Theorem~\ref{th:correctness})
\ifarxiv
-- the exact set-up of our earlier work~\cite{yunus2025using}, where we had
to use the \sql{JOIN} version as nothing else was available --
\else
-- the exact set-up of \cite{yunus2025using}, where the authors had to
use the \sql{JOIN} version as nothing else was available --
\fi
(i)~provenance tracking, producing a symbolic provenance formula, and
(ii)~probability computation, for \sql{HAVING} queries with the
comparison \sql{COUNT(*) >= k} or \sql{COUNT(*) <= k}, retrieving the
images with at least (at most) $k$ objects of class 72
(\emph{refrigerator}), $k$ ranging from 0 to 15 (no image has more
than~13).

For probabilistic query evaluation, we use the default ProvSQL pipeline:
a cost-based chooser selects, per provenance circuit, the cheapest exact
method it estimates applicable, from linear-time evaluation of read-once
circuits up to knowledge compilation
(d4~\cite{DBLP:conf/ijcai/LagniezM17}) as a last resort. Here the
circuit of a group has at most 13 inputs, and the \sql{HAVING} queries
are handled by the cheapest methods; the larger circuits of the
\sql{JOIN} rewritings also require tree decompositions and, for some
values of~$k$, knowledge compilation.

For each $k$, every query was executed $3$ times on a desktop computer
(Intel Core i7-12700, 64~GB of RAM) running Ubuntu 22.04, with PostgreSQL
18.6; the figures report the mean of the three runs, with error bars at
one standard deviation, on a logarithmic scale. The height of a bar is
the planning plus execution time measured by \sql{EXPLAIN ANALYZE}, and
a black tick marks the execution time alone. We disabled PostgreSQL's
genetic query optimizer, which by default replaces exhaustive join
ordering from 12 relations onwards: the join orders it selected made the
self-\sql{JOIN} rewritings with $k\geq 12$ exceed a five-minute timeout,
whereas the exhaustive planner keeps their execution around $3$ seconds,
at the price of a planning time that grows quickly with the number of
joins.

Figure~\ref{fig:bench} (top) covers provenance tracking: our semantics
constructs a circuit with a \emph{comparison gate} for the aggregate
comparison, and the possible-world semantics is triggered when
evaluating it in the formula pseudo-semiring (\sql{sr_formula}). As this
pseudo-semiring is not absorptive, the pruning of
Proposition~\ref{prop:algorithms} does not apply, and the valid worlds
of a group are enumerated in full -- at most $2^{13}$ here. The
running time of our semantics is stable as $k$ varies, between $0.2$
and $0.6$ seconds; for \sql{COUNT(*) >= k} it even decreases slowly
with~$k$, as the range check discards more groups. The
\sql{JOIN}-based rewriting is faster for $k\leq 4$, then its cost grows
with the number of joins, up to $14$ times that of our semantics at
$k=11$ and $106$ times at $k=15$, where planning the $15$-way self-join
alone takes $20$ seconds. For \sql{COUNT(*) <= k}, whose rewriting
includes a non-monotone \sql{EXCEPT}, the crossover is at $k=2$, and the
ratio reaches $7$ at $k=11$ and $108$ at $k=15$ ($62$ seconds of
planning for the $16$-way self-join): the stability of our runtime cost
pays off as $k$ increases.

Figure~\ref{fig:bench} (bottom) adds probability evaluation to query
execution and provenance computation.
The \sql{HAVING} semantics, with the Poisson-binomial optimization of
Proposition~\ref{prop:agg-cmp-poly-prob} or without it, stays between
$0.2$ and $0.5$ seconds for all~$k$, the two within a few percent of
each other: with at most $13$ objects per image, the DNF of valid worlds
of a group is small enough for the general pipeline to evaluate by world
enumeration or inclusion--exclusion, so the closed form brings no
measurable gain on this dataset (it would on larger groups, where the
DNF has $\binom{n}{k}$ clauses). The \sql{JOIN} rewriting is cheaper
for~$k\leq 3$, then $12$ to $17$ times slower than our semantics on the
monotone $\geq k$ case and $4$ to $8$ times slower on $\leq k$ up
to~$k=11$, as it pays, on top of the joins, the probability evaluation
of larger circuits, some requiring tree decomposition or knowledge
compilation; beyond, planning dominates and the ratio reaches $83$
and~$115$ at $k=15$. The \sql{JOIN} rewriting is also unable to exploit
the fact that the query becomes vacuously true or false when $k$ exceeds
the number of objects in the images, which both our implementations can.

\begin{toappendix}
  \section{Material for Section~\ref{sec:implementation_experiments}
  (Implementation and Experiments)}
\newcommand{\quotedzero}{\PYG{l+s+s1}{\apos$\mathbb{0}$\apos}}
\label{sec:appendix_impl_experiments}

We present the query templates used for the experiments of
Section~\ref{sec:implementation_experiments}. They are SQL queries using
the functions of the ProvSQL extension of PostgreSQL (\sql{provenance()}, \sql{sr_formula},
\sql{probability_evaluate}, etc.)

The \sql{HAVING}-based template is the following, its condition being
instantiated with \sql{COUNT(*) >= k} or with \sql{COUNT(*) <= k}:
\newcommand{\apos}{\texttt{\textquotesingle}}
\begin{minted}[fontsize=\small, breaklines, escapeinside=||]{postgresql}
SELECT * FROM (
  SELECT img, sr_formula(provenance(), 'image_mapping') AS prov
  FROM dataset
  WHERE obj = c
  GROUP BY img
  HAVING COUNT(*) >= k  -- or:  HAVING COUNT(*) <= k
) sub WHERE prov <> |\quotedzero|;
\end{minted}

For probability evaluation, the call to \sql{sr_formula} in the
projection is replaced by \sql{probability_evaluate(provenance())}, and
the final filter, which
discards the images whose provenance simplifies to~$\mathbb 0$, becomes
\sql{prov > 0}, discarding those whose probability is~$0$; under the
\sql{HAVING} semantics, these are the images for which the aggregate
comparison holds in no possible world.

The equivalent self-\sql{JOIN} rewriting for the ``\sql{>= k}'' case uses
a $(k{-}1)$-fold self-join, with the
``\sql{id > }'' chain ensuring distinct matching tuples:
\begin{minted}[fontsize=\small, breaklines, escapeinside=||]{postgresql}
SELECT * FROM (
  SELECT img, sr_formula(provenance(), 'image_mapping') AS prov
  FROM (
    SELECT DISTINCT d1.img
    FROM dataset d1
    JOIN dataset d2 ON d2.img = d1.img AND d2.obj = c AND d2.id > d1.id
    ...
    JOIN dataset dk ON dk.img = d1.img AND dk.obj = c AND dk.id > d|\(_{k-1}\)|.id
    WHERE d1.obj = c
  ) AS t
) sub WHERE prov <> |\quotedzero|;
\end{minted}

The ``\sql{<= k}'' case is non-monotone and is rewritten as
``at least one match''~\sql{EXCEPT}~``at least~$k+1$ matches'':
\begin{minted}[fontsize=\small, breaklines, escapeinside=||]{postgresql}
SELECT * FROM (
  SELECT img, sr_formula(provenance(), 'image_mapping') AS prov
  FROM (
    SELECT DISTINCT d1.img
    FROM dataset d1
    WHERE d1.obj = c
    EXCEPT
    SELECT DISTINCT d1.img
    FROM dataset d1
    JOIN dataset d2 ON d2.img = d1.img AND d2.obj = c AND d2.id > d1.id
    ...
    JOIN dataset d|\(_{k+1}\)| ON d|\(_{k+1}\)|.img = d1.img AND d|\(_{k+1}\)|.obj = c AND d|\(_{k+1}\)|.id > dk.id
    WHERE d1.obj = c
  ) AS t
) sub WHERE prov <> |\quotedzero|;
\end{minted}
The same projection / zero-filter substitution applies to obtain the
\sql{probability_evaluate} variants of both \sql{JOIN} rewritings.
\end{toappendix}

\section{Related Work}
\label{sec:related_work}

\inlinepar{Aggregation in the provenance semiring framework.}
Our approach builds on~\cite{amsterdamer2011provenance}, in which
\sql{HAVING} conditions become formal comparison expressions between
elements of a provenance semimodule; \cite{amsterdamer2011provenance}
stops at producing these expressions, whose evaluation, as well as
probability computation, is not considered. \cite{pintor2025dbms}
follows this approach with a DBMS-independent rewriting system that
stores the expressions as strings. The present work starts where these
two end, by giving the expressions a value in every commutative
m-semiring.

\inlinepar{Negation in the provenance semiring framework.}
Since provenance semirings \cite{green2007provenance,green2017provenance}
do not support non-monotone queries, \cite{geerts2010database} extended
them with a monus, yielding m-semirings. Some m-semirings lack axioms one
may expect~\cite{amsterdamer2011limitations}, notably the distributivity
of $\otimes$ over $\ominus$ discussed in Section~\ref{sec:preliminaries};
the model has nonetheless been fruitful, as the basis of a provenance for
SPARQL~\cite{DBLP:journals/jacm/GeertsUKFC16}, of a semiring
generalization of Codd's theorem~\cite{badia2025codd} and of the
\sql{EXCEPT} operator of ProvSQL~\cite{sen2026provsql}. Quotient
semirings of polynomials with dual
indeterminates~\cite{gradel2025provenance} are an alternative, less
amenable to implementation.

\inlinepar{Provenance systems.}
\ifarxiv
ProvSQL~\cite{sen2026provsql}, our system, which this work extends,
rewrites queries over provenance-tracked relations and stores annotations
as a circuit; it is also a probabilistic database. Before this work, its
aggregation followed~\cite{amsterdamer2011provenance}, recording
comparisons on aggregates as formal expressions that could be displayed
but not evaluated.
\else
ProvSQL~\cite{sen2026provsql}, the system we extend, rewrites queries
over provenance-tracked relations and stores annotations as a circuit;
its aggregation follows~\cite{amsterdamer2011provenance}, recording
comparisons on aggregates as formal expressions that can be displayed but
not evaluated; it is also a probabilistic database.
\fi
GProM~\cite{arab2018gprom} is a rewriting middleware over several
backends that computes several forms of provenance without full semiring
semantics; it joins aggregation results with the provenance of each group
and treats \sql{HAVING} as a plain selection, without recording how the
provenance depends on the comparison, which precludes probabilistic use.

\inlinepar{Aggregation in probabilistic databases.}
\cite{re2009trichotomy} characterizes which aggregate conjunctive queries
with a single \sql{HAVING} atom can be evaluated or approximated in
polynomial time over probabilistic databases, by the shape of the query;
the base case of its safe plans, the convolution of the marginal
distributions of independent tuples in a small semiring, is what
Proposition~\ref{prop:agg-cmp-poly-prob} restates in our setting, whereas
Section~\ref{sec:algorithms} classifies combinations of semirings and
aggregates for symbolic provenance, which~\cite{re2009trichotomy} does
not consider. \cite{abiteboul2011capturing} studies
aggregation over probabilistic XML, for terminal aggregates only. Closest
to us, \cite{fink2012aggregation} extends the construction
of~\cite{amsterdamer2011provenance} for probability evaluation and
supports \sql{HAVING}, with a semantics specific to the Boolean case and
to aggregates over commutative monoids, implemented in a prototype of
SPROUT that was not released. Among the probabilistic database prototypes
MystiQ~\cite{DBLP:conf/sigmod/BoulosDMMRS05},
Trio~\cite{dblp:conf/vldb/agrawalbshnsw06}, Orion~\cite{singh2008orion},
and MayBMS~\cite{huang2009maybms} with its public SPROUT
evaluator~\cite{DBLP:conf/icde/OlteanuHK09}, none is maintained,
aggregation over uncertain data is at best partial
(Trio~\cite{murthy2011aggregation}, Orion), and none gives \sql{HAVING}
an uncertain semantics.
ProbLog~\cite{DBLP:journals/tplp/FierensBRSGTJR15,DBLP:conf/pkdd/DriesKMRBVR15},
a probabilistic programming system, supports aggregation and conditions
on aggregates through grounding and weighted model counting, but no other
(m-)semiring, and does not scale to databases of realistic
sizes~\cite{sen2026provsql}.

\inlinepar{Aggregation over compact representations.}
\cite{bakibayev2013aggregation} adds aggregation and \sql{HAVING} to a
factorized database engine, without provenance and in a deterministic
setting; whether factorized representations could compress our possible
worlds is a natural question.

\section{Conclusion}
\label{sec:conclusion}

We have presented a general semantics for \sql{HAVING} queries that
applies to provenance in any commutative m-semiring, giving a value to
the formal comparison expressions of~\cite{amsterdamer2011provenance}. It agrees with the provenance of the self-join
rewriting in m-semirings that are absorptive and where times distributes
over monus, notably in the Boolean-function semiring used for
probability evaluation; elsewhere the two unavoidably differ, as outside
Boolean provenance the way a query is written changes its provenance. Algorithms for specific aggregates and comparisons, together with
experiments, show the approach feasible, yielding the first end-to-end
provenance-aware implementation of selection on aggregate values. A natural next step is
to lift the safe plans of~\cite{re2009trichotomy}, whose base case
Proposition~\ref{prop:agg-cmp-poly-prob} restates, to provenance
circuits, keeping probabilistic evaluation of safe \sql{HAVING}
queries over joins polynomial.

\begin{toappendix}
  \section{Dependency Graph of the Results}
  \label{sec:proofgraph}

  Figure~\ref{fig:proofgraph} summarizes how the results of this paper depend
  on one another.

  \begin{figure}[htbp]
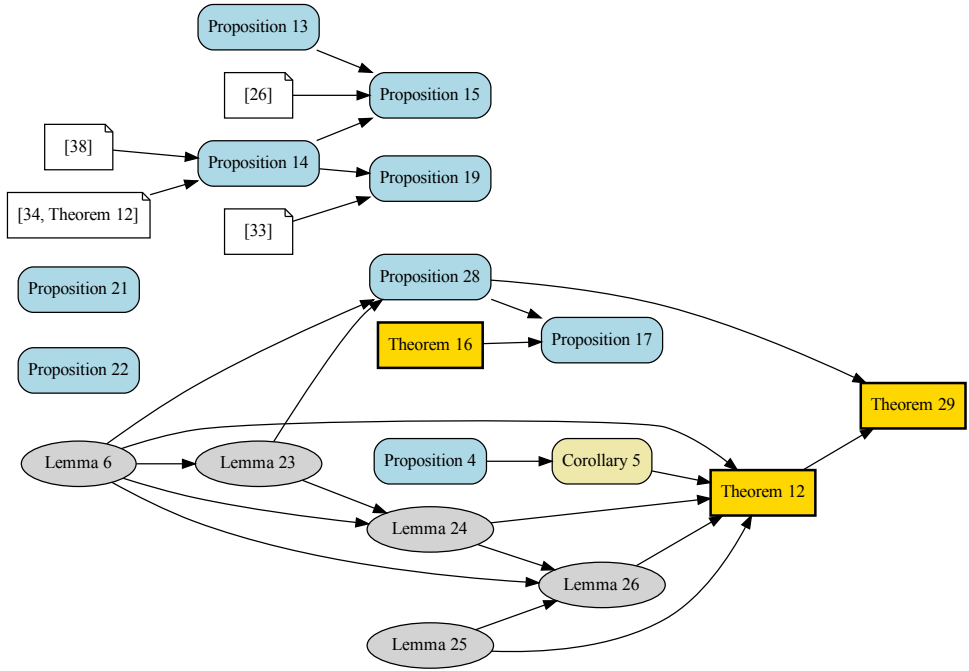

    \centering
    \proofgraph[width=\linewidth]
    \caption{Dependency graph of the results of the paper: an arrow from $A$
    to $B$ means that the proof of $B$ uses $A$.}
    \label{fig:proofgraph}
  \end{figure}
\end{toappendix}

\ifarxiv\else\clearpage\fi
\bibliography{ref}

\begin{thebibliography}{10}

\bibitem{abiteboul2011capturing}
Serge Abiteboul, T.{-}H.~Hubert Chan, Evgeny Kharlamov, Werner Nutt, and Pierre
  Senellart.
\newblock Capturing continuous data and answering aggregate queries in
  probabilistic {XML}.
\newblock {\em {ACM} Trans. Database Syst.}, 36(4):25:1--25:45, 2011.
\newblock \href {https://doi.org/10.1145/2043652.2043658}
  {\path{doi:10.1145/2043652.2043658}}.

\bibitem{dblp:conf/vldb/agrawalbshnsw06}
Parag Agrawal, Omar Benjelloun, Anish~Das Sarma, Chris Hayworth, Shubha~U.
  Nabar, Tomoe Sugihara, and Jennifer Widom.
\newblock Trio: {A} system for data, uncertainty, and lineage.
\newblock In Umeshwar Dayal, Kyu{-}Young Whang, David~B. Lomet, Gustavo Alonso,
  Guy~M. Lohman, Martin~L. Kersten, Sang~Kyun Cha, and Young{-}Kuk Kim,
  editors, {\em Proceedings of the 32nd International Conference on Very Large
  Data Bases, Seoul, Korea, September 12-15, 2006}, pages 1151--1154. {ACM},
  2006.
\newblock URL: \url{http://dl.acm.org/citation.cfm?id=1164231}.

\bibitem{DBLP:journals/tcs/AmarilliBDS19}
Antoine Amarilli, Mouhamadou~Lamine Ba, Daniel Deutch, and Pierre Senellart.
\newblock Computing possible and certain answers over order-incomplete data.
\newblock {\em Theor. Comput. Sci.}, 797:42--76, 2019.
\newblock URL: \url{https://doi.org/10.1016/j.tcs.2019.05.013}, \href
  {https://doi.org/10.1016/J.TCS.2019.05.013}
  {\path{doi:10.1016/J.TCS.2019.05.013}}.

\bibitem{amarilli2016example}
Antoine Amarilli and Mika\"el Monet.
\newblock Example of a naturally ordered semiring which is not an m-semiring.
\newblock \url{http://math.stackexchange.com/questions/1966858}, 2016.

\bibitem{amsterdamer2011limitations}
Yael Amsterdamer, Daniel Deutch, and Val Tannen.
\newblock On the limitations of provenance for queries with difference.
\newblock In Peter Buneman and Juliana Freire, editors, {\em 3rd Workshop on
  the Theory and Practice of Provenance, TaPP'11, Heraklion, Crete, Greece,
  June 20-21, 2011}. {USENIX} Association, 2011.
\newblock URL:
  \url{https://www.usenix.org/conference/tapp11/limitations-provenance-queries-difference}.

\bibitem{amsterdamer2011provenance}
Yael Amsterdamer, Daniel Deutch, and Val Tannen.
\newblock Provenance for aggregate queries.
\newblock In {\em Proceedings of the thirtieth ACM SIGMOD-SIGACT-SIGART
  symposium on Principles of database systems}, pages 153--164, 2011.

\bibitem{arab2018gprom}
Bahareh~Sadat Arab, Su~Feng, Boris Glavic, Seokki Lee, Xing Niu, and Qitian
  Zeng.
\newblock {GProM} - {A} {Swiss Army} knife for your provenance needs.
\newblock {\em {IEEE} Data Eng. Bull.}, 41(1):51--62, 2018.
\newblock URL: \url{http://sites.computer.org/debull/A18mar/p51.pdf}.

\bibitem{badia2025codd}
Guillermo Badia, Phokion~G. Kolaitis, and Carles Noguera.
\newblock Codd's theorem for databases over semirings.
\newblock {\em Proc. {ACM} Manag. Data}, 3(5):277:1--277:26, 2025.
\newblock \href {https://doi.org/10.1145/3767713} {\path{doi:10.1145/3767713}}.

\bibitem{badia2025coddarxiv}
Guillermo Badia, Phokion~G. Kolaitis, and Carles Noguera.
\newblock Codd's theorem for databases over semirings.
\newblock {\em CoRR}, abs/2501.16543, 2025.
\newblock Version~1, January 2025.
\newblock URL: \url{https://arxiv.org/abs/2501.16543v1}, \href
  {https://arxiv.org/abs/2501.16543v1} {\path{arXiv:2501.16543v1}}.

\bibitem{bakibayev2013aggregation}
Nurzhan Bakibayev, Tom{\'{a}}s Kocisk{\'{y}}, Dan Olteanu, and Jakub Zavodny.
\newblock Aggregation and ordering in factorised databases.
\newblock {\em Proc. {VLDB} Endow.}, 6(14):1990--2001, 2013.
\newblock URL: \url{http://www.vldb.org/pvldb/vol6/p1990-zavodny.pdf}, \href
  {https://doi.org/10.14778/2556549.2556579}
  {\path{doi:10.14778/2556549.2556579}}.

\bibitem{DBLP:conf/sigmod/BoulosDMMRS05}
Jihad Boulos, Nilesh~N. Dalvi, Bhushan Mandhani, Shobhit Mathur, Christopher
  R{\'{e}}, and Dan Suciu.
\newblock {MYSTIQ:} a system for finding more answers by using probabilities.
\newblock In Fatma {\"{O}}zcan, editor, {\em Proceedings of the {ACM} {SIGMOD}
  International Conference on Management of Data, Baltimore, Maryland, USA,
  June 14-16, 2005}, pages 891--893. {ACM}, 2005.
\newblock \href {https://doi.org/10.1145/1066157.1066277}
  {\path{doi:10.1145/1066157.1066277}}.

\bibitem{bunemankt01}
Peter Buneman, Sanjeev Khanna, and Wang~Chiew Tan.
\newblock Why and where: {A} characterization of data provenance.
\newblock In Jan~Van den Bussche and Victor Vianu, editors, {\em Database
  Theory - {ICDT} 2001, 8th International Conference, London, UK, January 4-6,
  2001, Proceedings}, volume 1973 of {\em Lecture Notes in Computer Science},
  pages 316--330. Springer, 2001.
\newblock \href {https://doi.org/10.1007/3-540-44503-X_20}
  {\path{doi:10.1007/3-540-44503-X_20}}.

\bibitem{DBLP:conf/pods/DayalGK82}
Umeshwar Dayal, Nathan Goodman, and Randy~H. Katz.
\newblock An extended relational algebra with control over duplicate
  elimination.
\newblock In Jeffrey~D. Ullman and Alfred~V. Aho, editors, {\em Proceedings of
  the {ACM} Symposium on Principles of Database Systems, March 29-31, 1982, Los
  Angeles, California, {USA}}, pages 117--123. {ACM}, 1982.
\newblock \href {https://doi.org/10.1145/588111.588132}
  {\path{doi:10.1145/588111.588132}}.

\bibitem{DBLP:conf/pkdd/DriesKMRBVR15}
Anton Dries, Angelika Kimmig, Wannes Meert, Joris Renkens, Guy~Van den Broeck,
  Jonas Vlasselaer, and Luc~De Raedt.
\newblock Problog2: Probabilistic logic programming.
\newblock In Albert Bifet, Michael May, Bianca Zadrozny, Ricard Gavald{\`{a}},
  Dino Pedreschi, Francesco Bonchi, Jaime~S. Cardoso, and Myra Spiliopoulou,
  editors, {\em Machine Learning and Knowledge Discovery in Databases -
  European Conference, {ECML} {PKDD} 2015, Porto, Portugal, September 7-11,
  2015, Proceedings, Part {III}}, volume 9286 of {\em Lecture Notes in Computer
  Science}, pages 312--315. Springer, 2015.
\newblock \href {https://doi.org/10.1007/978-3-319-23461-8_37}
  {\path{doi:10.1007/978-3-319-23461-8_37}}.

\bibitem{DBLP:journals/tplp/FierensBRSGTJR15}
Daan Fierens, Guy~Van den Broeck, Joris Renkens, Dimitar~Sht. Shterionov, Bernd
  Gutmann, Ingo Thon, Gerda Janssens, and Luc~De Raedt.
\newblock Inference and learning in probabilistic logic programs using weighted
  {Boolean} formulas.
\newblock {\em Theory Pract. Log. Program.}, 15(3):358--401, 2015.
\newblock \href {https://doi.org/10.1017/S1471068414000076}
  {\path{doi:10.1017/S1471068414000076}}.

\bibitem{fink2012aggregation}
Robert Fink, Larisa Han, and Dan Olteanu.
\newblock Aggregation in probabilistic databases via knowledge compilation.
\newblock {\em Proceedings of the VLDB Endowment}, 5(5):490--501, 2012.

\bibitem{geerts2010database}
Floris Geerts and Antonella Poggi.
\newblock On database query languages for {K}-relations.
\newblock {\em J. Appl. Log.}, 8(2):173--185, 2010.
\newblock URL: \url{https://doi.org/10.1016/j.jal.2009.09.001}, \href
  {https://doi.org/10.1016/J.JAL.2009.09.001}
  {\path{doi:10.1016/J.JAL.2009.09.001}}.

\bibitem{DBLP:journals/jacm/GeertsUKFC16}
Floris Geerts, Thomas Unger, Grigoris Karvounarakis, Irini Fundulaki, and
  Vassilis Christophides.
\newblock Algebraic structures for capturing the provenance of {SPARQL}
  queries.
\newblock {\em J. {ACM}}, 63(1):7:1--7:63, 2016.
\newblock \href {https://doi.org/10.1145/2810037} {\path{doi:10.1145/2810037}}.

\bibitem{gradel2025provenance}
Erich Gr{\"{a}}del and Val Tannen.
\newblock Provenance analysis and semiring semantics for first-order logic.
\newblock In Klaus Meer, Alexander Rabinovich, Elena~V. Ravve, and Andr{\'{e}}s
  Villaveces, editors, {\em Model Theory, Computer Science, and Graph
  Polynomials: Festschrift in Honor of Johann A. Makowsky}, Trends in
  Mathematics, pages 351--401. Springer Nature Switzerland, 2025.
\newblock \href {https://doi.org/10.1007/978-3-031-86319-6_21}
  {\path{doi:10.1007/978-3-031-86319-6_21}}.

\bibitem{green2007provenance}
Todd~J. Green, Gregory Karvounarakis, and Val Tannen.
\newblock Provenance semirings.
\newblock In Leonid Libkin, editor, {\em Proceedings of the Twenty-Sixth {ACM}
  {SIGACT-SIGMOD-SIGART} Symposium on Principles of Database Systems, June
  11-13, 2007, Beijing, China}, pages 31--40. {ACM}, 2007.
\newblock \href {https://doi.org/10.1145/1265530.1265535}
  {\path{doi:10.1145/1265530.1265535}}.

\bibitem{green2017provenance}
Todd~J. Green and Val Tannen.
\newblock The semiring framework for database provenance.
\newblock In Emanuel Sallinger, Jan~Van den Bussche, and Floris Geerts,
  editors, {\em Proceedings of the 36th {ACM} {SIGMOD-SIGACT-SIGAI} Symposium
  on Principles of Database Systems, {PODS} 2017, Chicago, IL, USA, May 14-19,
  2017}, pages 93--99. {ACM}, 2017.
\newblock \href {https://doi.org/10.1145/3034786.3056125}
  {\path{doi:10.1145/3034786.3056125}}.

\bibitem{DBLP:journals/iandc/GrumbachM99}
St{\'{e}}phane Grumbach and Tova Milo.
\newblock An algebra for pomsets.
\newblock {\em Inf. Comput.}, 150(2):268--306, 1999.
\newblock URL: \url{https://doi.org/10.1006/inco.1998.2777}, \href
  {https://doi.org/10.1006/INCO.1998.2777} {\path{doi:10.1006/INCO.1998.2777}}.

\bibitem{herschel2017survey}
Melanie Herschel, Ralf Diestelk{\"{a}}mper, and Houssem {Ben Lahmar}.
\newblock A survey on provenance: What for? what form? what from?
\newblock {\em {VLDB} J.}, 26(6):881--906, 2017.
\newblock URL: \url{https://doi.org/10.1007/s00778-017-0486-1}, \href
  {https://doi.org/10.1007/S00778-017-0486-1}
  {\path{doi:10.1007/S00778-017-0486-1}}.

\bibitem{huang2009maybms}
Jiewen Huang, Lyublena Antova, Christoph Koch, and Dan Olteanu.
\newblock {MayBMS}: a probabilistic database management system.
\newblock In Ugur {\c{C}}etintemel, Stanley~B. Zdonik, Donald Kossmann, and
  Nesime Tatbul, editors, {\em Proceedings of the {ACM} {SIGMOD} International
  Conference on Management of Data, {SIGMOD} 2009, Providence, Rhode Island,
  USA, June 29 - July 2, 2009}, pages 1071--1074. {ACM}, 2009.
\newblock \href {https://doi.org/10.1145/1559845.1559984}
  {\path{doi:10.1145/1559845.1559984}}.

\bibitem{jha2010bridging}
Abhay~Kumar Jha, Dan Olteanu, and Dan Suciu.
\newblock Bridging the gap between intensional and extensional query evaluation
  in probabilistic databases.
\newblock In Ioana Manolescu, Stefano Spaccapietra, Jens Teubner, Masaru
  Kitsuregawa, Alain L{\'{e}}ger, Felix Naumann, Anastasia Ailamaki, and Fatma
  {\"{O}}zcan, editors, {\em {EDBT} 2010, 13th International Conference on
  Extending Database Technology, Lausanne, Switzerland, March 22-26, 2010,
  Proceedings}, volume 426 of {\em {ACM} International Conference Proceeding
  Series}, pages 323--334. {ACM}, 2010.
\newblock \href {https://doi.org/10.1145/1739041.1739082}
  {\path{doi:10.1145/1739041.1739082}}.

\bibitem{DBLP:conf/coco/Karp72}
Richard~M. Karp.
\newblock Reducibility among combinatorial problems.
\newblock In Raymond~E. Miller and James~W. Thatcher, editors, {\em Proceedings
  of a symposium on the Complexity of Computer Computations, held March 20-22,
  1972, at the {IBM} Thomas J. Watson Research Center, Yorktown Heights, New
  York, {USA}}, The {IBM} Research Symposia Series, pages 85--103. Plenum
  Press, New York, 1972.
\newblock \href {https://doi.org/10.1007/978-1-4684-2001-2_9}
  {\path{doi:10.1007/978-1-4684-2001-2_9}}.

\bibitem{DBLP:conf/ijcai/LagniezM17}
Jean{-}Marie Lagniez and Pierre Marquis.
\newblock An improved decision-{DNNF} compiler.
\newblock In Carles Sierra, editor, {\em Proceedings of the Twenty-Sixth
  International Joint Conference on Artificial Intelligence, {IJCAI} 2017,
  Melbourne, Australia, August 19-25, 2017}, pages 667--673. ijcai.org, 2017.
\newblock URL: \url{https://doi.org/10.24963/ijcai.2017/93}, \href
  {https://doi.org/10.24963/IJCAI.2017/93} {\path{doi:10.24963/IJCAI.2017/93}}.

\bibitem{DBLP:journals/tcs/Libkin03}
Leonid Libkin.
\newblock Expressive power of {SQL}.
\newblock {\em Theor. Comput. Sci.}, 296(3):379--404, 2003.
\newblock \href {https://doi.org/10.1016/S0304-3975(02)00736-3}
  {\path{doi:10.1016/S0304-3975(02)00736-3}}.

\bibitem{lin2014microsoft}
Tsung{-}Yi Lin, Michael Maire, Serge~J. Belongie, James Hays, Pietro Perona,
  Deva Ramanan, Piotr Doll{\'{a}}r, and C.~Lawrence Zitnick.
\newblock Microsoft {COCO:} {Common Objects in Context}.
\newblock In David~J. Fleet, Tom{\'{a}}s Pajdla, Bernt Schiele, and Tinne
  Tuytelaars, editors, {\em Computer Vision - {ECCV} 2014 - 13th European
  Conference, Zurich, Switzerland, September 6-12, 2014, Proceedings, Part
  {V}}, volume 8693 of {\em Lecture Notes in Computer Science}, pages 740--755.
  Springer, 2014.
\newblock \href {https://doi.org/10.1007/978-3-319-10602-1_48}
  {\path{doi:10.1007/978-3-319-10602-1_48}}.

\bibitem{murthy2011aggregation}
Raghotham Murthy, Robert Ikeda, and Jennifer Widom.
\newblock Making aggregation work in uncertain and probabilistic databases.
\newblock {\em {IEEE} Trans. Knowl. Data Eng.}, 23(8):1261--1273, 2011.
\newblock \href {https://doi.org/10.1109/TKDE.2010.166}
  {\path{doi:10.1109/TKDE.2010.166}}.

\bibitem{DBLP:conf/icde/OlteanuHK09}
Dan Olteanu, Jiewen Huang, and Christoph Koch.
\newblock {SPROUT:} lazy vs. eager query plans for tuple-independent
  probabilistic databases.
\newblock In Yannis~E. Ioannidis, Dik~Lun Lee, and Raymond~T. Ng, editors, {\em
  Proceedings of the 25th International Conference on Data Engineering, {ICDE}
  2009, March 29 2009 - April 2 2009, Shanghai, China}, pages 640--651. {IEEE}
  Computer Society, 2009.
\newblock \href {https://doi.org/10.1109/ICDE.2009.123}
  {\path{doi:10.1109/ICDE.2009.123}}.

\bibitem{pintor2025dbms}
Paulo Pintor, Rog{\'{e}}rio Costa, and Jos{\'{e}} Moreira.
\newblock A {DBMS}-independent approach for capturing provenance polynomials
  through query rewriting.
\newblock {\em CoRR}, abs/2508.14608, 2025.
\newblock URL: \url{https://doi.org/10.48550/arXiv.2508.14608}, \href
  {https://arxiv.org/abs/2508.14608} {\path{arXiv:2508.14608}}, \href
  {https://doi.org/10.48550/ARXIV.2508.14608}
  {\path{doi:10.48550/ARXIV.2508.14608}}.

\bibitem{re2009trichotomy}
Christopher R{\'{e}} and Dan Suciu.
\newblock The trichotomy of {HAVING} queries on a probabilistic database.
\newblock {\em {VLDB} J.}, 18(5):1091--1116, 2009.
\newblock URL: \url{https://doi.org/10.1007/s00778-009-0151-4}, \href
  {https://doi.org/10.1007/S00778-009-0151-4}
  {\path{doi:10.1007/S00778-009-0151-4}}.

\bibitem{sen2026provsql}
Aryak Sen, Silviu Maniu, and Pierre Senellart.
\newblock {ProvSQL}: {A} general system for keeping track of the provenance and
  probability of data.
\newblock In {\em Proc.\ {ICDE}}, Montr\'eal, Canada, May 2026.

\bibitem{senellart2017provenance}
Pierre Senellart.
\newblock Provenance and probabilities in relational databases.
\newblock {\em {SIGMOD} Rec.}, 46(4):5--15, 2017.
\newblock \href {https://doi.org/10.1145/3186549.3186551}
  {\path{doi:10.1145/3186549.3186551}}.

\bibitem{senellart2018provsql}
Pierre Senellart, Louis Jachiet, Silviu Maniu, and Yann Ramusat.
\newblock {ProvSQL}: Provenance and probability management in {PostgreSQL}.
\newblock {\em Proc. {VLDB} Endow.}, 11(12):2034--2037, 2018.
\newblock URL: \url{http://www.vldb.org/pvldb/vol11/p2034-senellart.pdf}, \href
  {https://doi.org/10.14778/3229863.3236253}
  {\path{doi:10.14778/3229863.3236253}}.

\bibitem{singh2008orion}
Sarvjeet Singh, Chris Mayfield, Sagar Mittal, Sunil Prabhakar, Susanne~E.
  Hambrusch, and Rahul Shah.
\newblock Orion 2.0: native support for uncertain data.
\newblock In Jason~Tsong{-}Li Wang, editor, {\em Proceedings of the {ACM}
  {SIGMOD} International Conference on Management of Data, {SIGMOD} 2008,
  Vancouver, BC, Canada, June 10-12, 2008}, pages 1239--1242. {ACM}, 2008.
\newblock \href {https://doi.org/10.1145/1376616.1376744}
  {\path{doi:10.1145/1376616.1376744}}.

\bibitem{suciu2011probabilistic}
Dan Suciu, Dan Olteanu, Christopher R{\'{e}}, and Christoph Koch.
\newblock {\em Probabilistic Databases}.
\newblock Synthesis Lectures on Data Management. Morgan {\&} Claypool
  Publishers, 2011.
\newblock \href {https://doi.org/10.2200/S00362ED1V01Y201105DTM016}
  {\path{doi:10.2200/S00362ED1V01Y201105DTM016}}.

\bibitem{tpch}
{Transaction Processing Performance Council}.
\newblock {TPC} benchmark {H} (decision support) standard specification.
\newblock \url{https://www.tpc.org/tpch/}, 2022.
\newblock Revision 3.0.1.

\bibitem{ultralytics2021yolov5}
Ultralytics.
\newblock {YOLOv5}: {A} state-of-the-art real-time object detection system.
\newblock \url{https://docs.ultralytics.com}, 2021.

\bibitem{widiaatmaja2025demonstration}
Albert~Ariel Widiaatmaja, Belkis Djeffal, Ashish Dandekar, and Pierre
  Senellart.
\newblock Demonstration of {ProvSQL} update provenance through temporal
  databases.
\newblock In {\em Proc.\ {PW}}, Berlin, Germany, 06 2025.
\newblock Demonstration.

\bibitem{yunus2025using}
Fajrian Yunus, Pratik Karmakar, Pierre Senellart, Talel Abdessalem, and
  St{\'{e}}phane Bressan.
\newblock Using {A} probabilistic database in an image retrieval application.
\newblock In Alkis Simitsis, Bettina Kemme, Anna Queralt, Oscar Romero, and
  Petar Jovanovic, editors, {\em Proceedings 28th International Conference on
  Extending Database Technology, {EDBT} 2025, Barcelona, Spain, March 25-28,
  2025}, pages 1106--1109. OpenProceedings.org, 2025.
\newblock URL: \url{https://doi.org/10.48786/edbt.2025.100}, \href
  {https://doi.org/10.48786/EDBT.2025.100} {\path{doi:10.48786/EDBT.2025.100}}.

\end{thebibliography}

\end{document}